\documentclass[%
 reprint,superscriptaddress,
 amsmath,amssymb,
 aps,
]{revtex4-2}
\usepackage{simplewick}
\usepackage{algorithmic}
\usepackage{graphicx}
\usepackage{dcolumn}
\usepackage[bbgreekl]{mathbbol} 
\usepackage{dsfont}
\usepackage{bbold}
\usepackage{bm}
\usepackage{upgreek}
\usepackage{bm}
\usepackage{multirow}
\usepackage[table, xcdraw]{xcolor}
\usepackage{amssymb}
\usepackage{tabularx}   
\usepackage{mathtools,halloweenmath}
\usepackage{physics}
\usepackage{graphicx}
\usepackage{amsmath} 
\usepackage{amsthm} 
\usepackage[version=4]{mhchem}
\usepackage{siunitx}
\usepackage{longtable,tabularx}
\usepackage{braket}
\usepackage{float}
\usepackage{ulem}
\usepackage{comment}
\usepackage{hyperref} 
\usepackage{longtable}
\usepackage{enumitem}
\usepackage{booktabs}
\usepackage{array} 

\usepackage[ruled,lined]{algorithm2e} 

\usepackage{mathtools}

\newtheorem{theorem}{Theorem} 
\theoremstyle{remark} 
\newtheorem{remark}{Remark} 

\begin{document}
\preprint{APS/123-QED} 
\title[]{Quantum linear solvers for quantum chemistry: prospects of exponential quantum advantage} 
\author{Peniel Bertrand Tsemo}
\email{peniel.tsemo.72@tcgcrest.org} 
\affiliation{Centre for Quantum Engineering, Research and Education, TCG CREST, Sector V, Salt Lake, Kolkata 700091, India}
\affiliation{Department of Physics, IIT Tirupati, Chindepalle, Andhra Pradesh 517619, India}
\author{Kenji Sugisaki}
\affiliation{Centre for Quantum Engineering, Research and Education, TCG CREST, Sector V, Salt Lake, Kolkata 700091, India}
\affiliation{Deloitte Tohmatsu LLC, 3-2-3 Marunouchi, Chiyoda-ku, Tokyo 100-8363, Japan}
\author{Ishita Bhattacharjee}
\affiliation{Indian Association for the Cultivation of Science, Kolkata 700032, India}
\affiliation{Donostia international physics centre, Manuel Lardizabal 4
E20018 Donostia / San Sebastián, Spain}
\author{V. S. Prasannaa} 
\email{srinivasaprasannaa@gmail.com}
\affiliation{Centre for Quantum Engineering, Research and Education, TCG CREST, Sector V, Salt Lake, Kolkata 700091, India}
\affiliation{Academy of Scientific and Innovative Research (AcSIR), Ghaziabad- 201002, India}

\begin{abstract} 
Quantum linear solvers (QLSs) offer the potential for exponential quantum advantage in solving quantum chemical problems, but assessing such an advantage hinges on determining the condition number ($\kappa$) scaling, which itself is computationally challenging. While a recent work applied the Harrow-Hassidim-Lloyd (HHL) algorithm to single-reference linearized coupled cluster equations (SRLCC), the validity of the HHL-SRLCC framework is restricted to weakly correlated regimes. A general treatment requires a formulation that can access strongly correlated regions. We thus begin by extending the QLS-SRLCC framework to its multi-reference form, which is based on the internally contracted multi-reference LCC method (QLS-icMRLCC). We then analyze $\kappa$ scaling across a diverse range of atomic and molecular systems as well as single-particle basis sets in both QLS-SRLCC and QLS-icMRLCC formulations, using three complementary diagnostics that range from explicit computations to use of indirect structural indicators: (i) direct calculations of $\kappa$, (ii) scaling of the ratio of maximum to minimum diagonal entries of an $A$ matrix, and (iii) structural analyses of the $A$ matrices based on a recently proposed conjecture that we adapt to the QLS-LCC problem. The three approaches yield broadly consistent predictions, with our numerical results indicating a slow-growing $\kappa$ scaling in system size for the sampled systems. For most sampled systems, the numerical results are consistent with approximately polylogarithmic behaviour, while the SRLCCSD H$_4$ results in the cc-pVQZ basis exhibit an approximately linear dependence on the number of virtual orbitals. For icMRLCC, however, the observed behaviour depends on the reference active space, leaving the broader prospects of quantum advantage with it less certain. This finding, when combined with our arguments on sub-linear scaling of sparsity (when the so-called excitation rank is greater than 2), supports the prospects of exponential advantage using QLSs for the LCC problem. However, the prospects of a QLS-level advantage do not necessarily translate to a speed-up in an end-to-end workflow. Thus, using HHL-SRLCCSD as a representative example, we discuss the challenges associated with realizing such a speed-up. Finally, up to 28-qubit HHL-icMRLCCSD numerical calculations on potential energy curves of model systems containing up to four atoms recover the ground state energies with about $0.012\%$ errors relative to benchmark classical methods. 
\end{abstract} 

\maketitle 

\tableofcontents 

\section{Introduction}\label{sec:introduction} 

Quantum chemistry on quantum computers has emerged as one of the most actively pursued research areas in recent times in the field of quantum information processing, with significant progress driven by approaches such as the variational quantum eigensolver (VQE)~\cite{Peruzzo2014, TILLY20221, Harville2026} and the quantum phase estimation (QPE) algorithms \cite{kitaevqpe1995, QPE_Sethlloyd, Abrams1997, Cleve_1998}. While VQE was the method of choice until recently, QPE has made a strong and well-deserved comeback due to recent quantum hardware advancements and the scalable nature of the algorithm \cite{yamamoto2026,yamamoto_bayesian_qpe}. Such developments prompt a need for exploring new and efficient quantum algorithms that can broaden the scope of achieving quantum advantage in chemistry applications. 

In this context, the recent proposal of using the Harrow-Hassidim-Lloyd (HHL) algorithm \cite{hhl_harrow} to quantum chemistry \cite{Nishanth2023} provides a promising framework for solving the single-reference linearized coupled cluster (SRLCC) equations. HHL can, in principle, offer an exponential advantage in runtime over its classical counterparts. Since the SRLCC equations can be recast as a system of linear equations, $A\vec{x}=\vec{b}$, where $A$ and $\vec{b}$ are built out of Hamiltonian matrix elements between Slater determinants, and are supplied as inputs to HHL, while the solution vector, $\vec{x}$, contains the cluster amplitudes, they are naturally amenable to be solved using quantum linear solvers (QLSs). Furthermore, the need to read the full solution vector is avoided by appending a quantum circuit at the end of the HHL circuit to extract correlation energies. By working with the normal ordered Hamiltonian, the approach outputs correlation energies directly rather than total energies, making it attractive for quantum chemistry applications, where the correlation energy is a small fraction of the total energy, and is thus challenging to resolve on noisy quantum devices. 

Despite the work on connecting HHL and SRLCC being a first step in the direction of using QLSs in quantum chemistry, two challenges remain. While the SRLCC equations can capture weak correlation effects (dynamical correlations), one requires a multi-reference treatment (an MRLCC formalism) to capture \textit{static correlation} effects in order to yield sufficiently accurate results in situations such as making and breaking of covalent chemical bonds. The first step to use QLSs for chemistry would require proposing a QLS-MRLCC framework alongside QLS-SRLCC. The second challenge pertains to the central question of whether or not one needs a QLS at all for the LCC problem (be it SR or MR), that is, the question of whether QLS-LCC is expected to offer a quantum advantage. Addressing this question is as hard as it is important; any attempt to answer the question requires computing the condition number ($\kappa$) scaling of $A$ with system size, which in itself is not a classically easy problem since one needs to learn the eigenvalues of $A$ to find $\kappa$. The current work focuses on answering both of these non-trivial questions. 

With regard to the first question, there have been significant developments in the field of multi-reference coupled cluster theories \cite{Mukherjee01041977, JezioMonkhorst1668, MUKHERJEE1986207, MUKHOPADHYAY1992236, Kirtman1981} and therefore a wide variety of approaches to choose from. However, for our purposes, the most attractive of these are the matrix product states based multi-reference linearized coupled cluster formulation work by Sharma and Alavi that leads to linear equations at each perturbative order \cite{Alavi2015}, and works on the internally contracted MRLCC (icMRLCC) approach by Köhn \cite{Ajit1982, Evangelista2011, Andreas2011}. We note that the latter amounts to an all-order perturbative treatment of selected contributions as opposed to the former's sequence of linear equations for successive perturbative corrections to the wave function. Since the icMRLCC governing equations not only naturally admit a linear system formulation for a given multi-configurational reference but also the existing icMRCC computational framework is fairly mature thereby allowing us to modify current codes and explicitly construct the relevant matrices to investigate their spectral and sparsity properties, we pick the icMRLCC approach as the method of choice for addressing strong correlation effects in our current work. 

As far as the second question is concerned, existing studies that assess quantum advantage with HHL typically do so by numerically studying $\kappa$ scaling for representative problem instances. In fact, with the notable exception of HHL-SRLCC, where the authors find a polylogarithmic $\kappa$ scaling with system size \cite{Peniel2025}, the rest of the applications report at least polynomial in system size growths of $\kappa$ \cite{Jin_2022, Golden2022, Gopal2024, Yalovetzky2024, pareek_2026}. A recent work that involved a large scale numerical survey on graph families also reaches a similar conclusion, suggesting that favourable $\kappa$ scaling and thus quantum advantage with QLSs is not common \cite{shetty2025}. 

While the encouraging results on slow $\kappa$ growth from Ref. \cite{Peniel2025} motivate a deeper analysis of prospects of advantage in QLS-LCC, the difficulty of generalizing beyond limited problem instances led us to employ three different methods ranging from direct (yet computationally costly) to indirect (but computationally cheap) determination of $\kappa$ scaling over some representative atomic and molecular systems and compare results from them: (i) explicit computations of condition number scaling with system size, (ii) numerically assessing the usefulness of our derivation of an upper bound for $\kappa$ of real symmetric matrices (the ones we encounter in the LCC problem), and (iii) use of the recently proposed \cite{shetty2025} edge spawning conjecture to assess if $\kappa$ growth is slow (favourable for advantage) or fast (detrimental to advantage). Our analysis provides broadly consistent indications of slow $\kappa$ growth over the investigated system sizes, with approximately polylogarithmic behaviour for most sampled systems. However, the SRLCCSD H$_4$ results in the cc-pVQZ basis exhibit an approximately linear dependence on the number of virtual orbitals. Together with our sparsity analysis, these results support the prospects of exponential separation between classical and quantum linear solvers, even in certain cases where $\kappa$ grows polynomially. 

Since our analysis of advantage focuses solely on the QLS and the feature extraction (correlation energy computation) steps of an end-to-end workflow, we consider HHL-SRLCCSD as a representative example and discuss the classical pre-processing costs involved in constructing $A$ and $\vec{b}$ as well as broader challenges and future directions for input state preparation and implementing quantum access to $A$. The discussion aims to serve as an identification of additional steps whose costs must be addressed in future works towards realizing an end-to-end quantum advantage. 

To complement the scaling analysis, we also perform proof-of-concept numerical HHL-icMRLCC simulations to determine ground state energies of model molecular systems containing up to four atoms. Our simulations capture ground state energies to within an error of $0.012\%$ of benchmark classical methods, demonstrating the accuracy of the approaches. 

While our assessment of advantage is motivated in part by the approach taken for graph-based linear systems in Ref. \cite{shetty2025}, the electronic structure problem introduces distinct mathematical and computational challenges. In particular, the LCC matrices arise from many-electron Hamiltonians and possess excitation-rank-dependent dimensions (system size), sparsity patterns constrained by Slater–Condon rules, and non-trivial spectral properties. Besides analyzing such structures and their implications for advantage, we extend the QLS-LCC formulation to an internally contracted multi-reference setting, examine complementary diagnostics for $\kappa$ scaling, and discuss the pre- and post-processing as well as quantum resource requirements of the resulting chemistry calculations. 

The subsequent sections are ordered as follows: Section \ref{sec:theory} presents the relevant theoretical framework by discussing QLS-SRLCC and QLS-icMRLCC. Section \ref{sec:scaling} discusses the relevant diagnostics for $\kappa$ scaling, while Sections \ref{sec:sr-scaling} and \ref{sec:icmr-scaling} discuss scaling for QLS-SRLCC and QLS-icMRLCC respectively, with major focus on $\kappa$ scaling using the three methods mentioned above, while also deriving sparsity scaling and system size growth. Section \ref{sec:advantage} puts together our findings for $\kappa$ and sparsity scaling to find prospects for exponential advantage for QLS-LCC. The section also discusses the two-qubit gate counts and $T$ counts. This is followed by Section \ref{sec:preproc-scaling}, where we present pre- and post-processing costs followed by their impact on the end-to-end QLS-LCC workflow. Section \ref{sec:calc} presents proof-of-concept HHL-icMRLCC numerical simulations. We conclude in Section \ref{sec:conclusion}. Fig. \ref{fig1} presents the overall structure of the work. Since the work involves several abbreviations, we have listed them in alphabetical order in Table \ref{tab:acronyms} of the Appendix. 

\begin{figure*}[t] 
\centering 
    \includegraphics[scale=0.37]{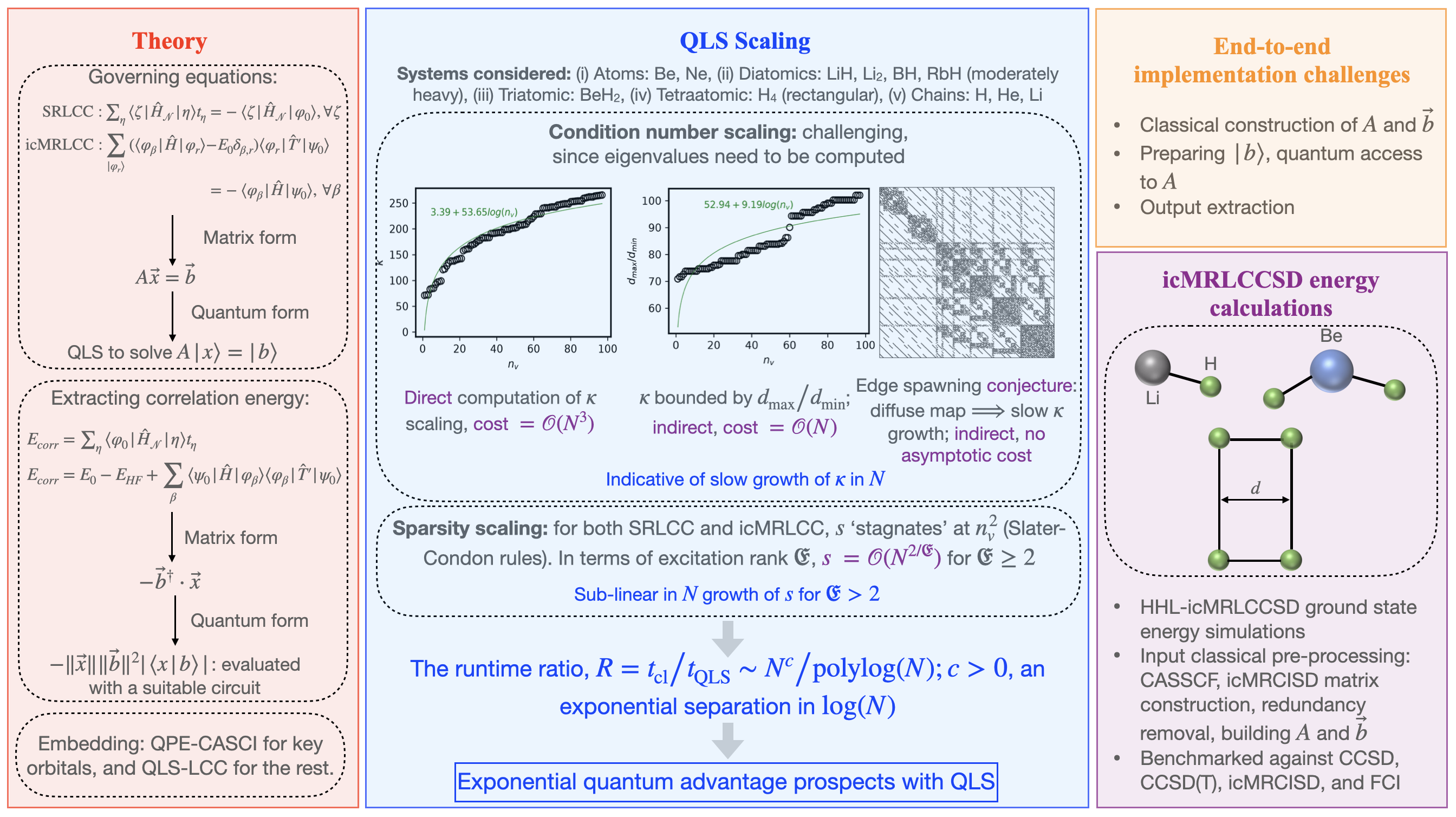} 
    \caption{Overview of the present work and organization. The left-most panel (Theory, see Section \ref{sec:theory} onwards) describes the working equations of single- and internally contracted multi-reference linearized coupled cluster theories abbreviated as SRLCC and icMRLCC respectively, reformulates them in the form $A\vec{x} = \vec{b}$, and proposes embedding quantum phase estimation-complete active space configuration interaction (QPE-CASCI) in QLS-LCC for large molecular calculations. In the middle panel (QLS scaling, see Section \ref{sec:scaling}), we describe three diagnostics for determining the condition number, $\kappa$, of the $A$ matrix, using which we find indications of the quantity scaling at most as a polylogarithmic function of system size for most systems, $N$, based on numerical calculations on several molecules in multiple basis sets (see Sections \ref{sec:sr-scaling} and \ref{sec:icmr-scaling}). We also show a sub-linear upper bound in $N$ for the sparsity, $s$, for excitation rank $\mathfrak{E}>2$. Then we combine our findings on $\kappa$ and $s$ scaling and identify conditions even in cases where $\kappa$ grows polynomially to arrive at the runtime ratio $R \sim t_{\text{cl}} /t_{QLS}$, which suggests an exponential separation in $\log(N)$, thus, prospects of exponential quantum advantage (see Section \ref{sec:advantage}). The top panel on the extreme right (End-to-end implementation challenges) touches upon challenges encountered in an end-to-end workflow by considering HHL-SRLCCSD as an example (see Section \ref{sec:preproc-scaling}). Finally, the bottom panel at the extreme right (icMRLCCSD energy calculations, Section \ref{sec:calc}) gives our numerical HHL simulation results, and benchmarks them with CASSCF (complete active space self-consistent field), icMRCI (internally contracted multi-reference configuration interaction), CCSD (coupled cluster in the singles and doubles approximation), and CCSD(T) (CCSD with inclusion of perturbative triples). $\hat{H}_{\mathcal{N}}$ is the normal ordered Hamiltonian, $\ket{\eta}$ and $\ket{\zeta}$ are excited Slater determinants due to particle-hole excitations that arise from the reference state, $\ket{\varphi_0}$, which is the mean field state for SRLCC (or an entangled state, $\ket{\psi_0}$, for icMRLCC (not shown explicitly in the figure for convenience)}\label{fig1} 
\end{figure*} 

\section{Theory}\label{sec:theory} 

\subsection{Quantum linear solvers}\label{subsec:qls} 

QLSs are scalable quantum algorithms that `solve' for a system of linear equations under certain assumptions, by preparing a quantum state proportional to the solution vector. The problem of solving $A\vec{x}=\vec{b}$ (with $A$ being a known $(N \times N)$ matrix, $\vec{x}$ a $N \times 1$ unknown vector that we want to find, and $\vec{b}$ an $N \times 1$ known vector) is changed to `solving' $A \ket{x}=\ket{b}$, where the matrix $A$ and the quantum state $\ket{b} \in \mathcal{H}^{n_b}$ are inputs, and a QLS protocol prepares an approximation to $\ket{x}$. Here, the $n_b$-qubit state $\ket{b}$ contains information on $\vec{b}$, for example, one could amplitude encode the latter and normalize the resulting state to obtain the former quantity. Often, it is assumed that the matrix elements of $A$ can be loaded efficiently, and an efficient quantum circuit prepares $\ket{b}$ from $\ket{0^{\otimes n_b}}$. Throughout, we too will work under these two standard assumptions, but we will comment on these points in Secs. \ref{sec:preproc-scaling}A(c) and \ref{sec:preproc-scaling}A(d). 

In this work, since our primary interest lies in the runtime complexities of the QLSs for chemistry, we consider two quantum linear solvers: HHL, which is the prototypical scalable QLS whose runtime scales as 

\begin{eqnarray}
t_{HHL} = \mathcal{O}\textbf{\bigg(}\mathrm{log}(N)s(N)^2 \kappa(N)^3\frac{1}{\epsilon(N)}\textbf{\bigg)}, \label{eq:thhl} 
\end{eqnarray}

and the near-optimal in $\kappa(N)$ and $1/\epsilon(N)$ Childs-Kothari-Somma (CKS) algorithm \cite{Childs_2017}, whose runtime scales as 

\begin{eqnarray}
t_{CKS} = \mathcal{O}\textbf{\bigg(}\mathrm{log}(N)s (N) \kappa(N) \nonumber \\ \mathrm{polylog}\bigg(s(N) \kappa(N) \frac{1}{\epsilon(N)}\bigg)\textbf{\bigg)}. \label{eq:tcks} 
\end{eqnarray}

In the above equations, $\kappa = 
\frac{\displaystyle \max_i |\lambda_i(A)|}
{\displaystyle \min_i |\lambda_i(A)|}
$ is the condition number of $A$, where $\lambda_i(A)$ refers to the eigenvalues of $A$. $s$ is the sparsity of $A$, and is defined as the maximum number of non-zero elements per row or column of $A$. $\epsilon$ is the additive error resulting from inadequately representing the eigenvalues of $A$ in their binary form. For a classical algorithm reference, we shall assume that the approach scales similar to conjugate gradient for simplicity: 

\begin{eqnarray} 
t_{\mathrm{cl}} = \mathcal{O}\textbf{\bigg(}Ns(N) \sqrt{\kappa(N)}\mathrm{log}\bigg(\frac{1}{\epsilon(N)}\bigg)\textbf{\bigg)}. \label{eq:tcg} 
\end{eqnarray}

This classical runtime expression is intended as a representative baseline rather than a lower bound on the cost of the best classical electronic structure algorithms for linear systems. Techniques such as preconditioning and/or other problem-specific optimizations can reduce classical computational costs along with additional pre-processing overheads, and their effects are not included in the present study. We also note that all of the runtimes we consider are asymptotic, and they exclude prefactors. Analyzing the effects of prefactors in practical computations is outside the scope of the current study. 

Finally, we define the runtime ratio, 

\begin{eqnarray} 
R(N; QLS) := t_{\mathrm{cl}} \big/ t_{QLS}, 
\end{eqnarray} 

where the QLS can either be HHL or CKS. Henceforth, we will drop mentioning the explicit dependence of $\kappa$, $s$, and $\epsilon$ on $N$ for brevity, but the above expressions show that an understanding of the scaling of these quantities with system size, $N$, is necessary to understand if we will realize an advantage with the quantum solvers over the classical one. We will revisit this central point in Section \ref{sec:scaling} in the context of applying QLSs for chemistry. 

Since we shall be using the HHL algorithm for our numerical simulations later on, we briefly explain the HHL algorithm in Section \ref{app:hhl} of the Appendix. We direct the interested readers to Ref. \cite{morales_qlss2025} for a detailed review on QLSs, including HHL. 

\subsection{Single-reference linearized coupled cluster method with QLS}\label{subsec:qls-srlcc} 

\begin{center}
    \textit{a. SRLCCSD amplitude equations} 
\end{center} 

We now derive the SRLCCSD equations (SRLCC in the singles and doubles approximation) for illustration; the same framework extends in a straightforward manner to higher excitation ranks. We begin with the Schrödinger equation: 

\begin{equation}
    \hat{H} \ket{\psi} = E \ket{\psi}\label{eq:SE}. 
\end{equation} 

$\hat{H}$ is the electronic structure Hamiltonian in our case. We choose the exponential coupled cluster ansatz for the wave function, $\ket{\psi} = e^{\hat{T}} \ket{\varphi_0}$, where $\ket{\varphi_0}$ is the Hartree-Fock (HF) state and $\hat{T}=\hat{T}_1 + \hat{T}_2$ for the CCSD approximation. $\hat{T}_1$ and $\hat{T}_2$ are cluster operators that respectively create singly and doubly excited functions upon acting on $\ket{\varphi_0}$. We could include higher rank excitations, such as triples and quadruples (CCSDT and CCSDTQ) too, but usually CCSD level of truncation suffices to predict molecular properties reasonably well for weakly correlated regimes. By substituting the form of the CCSD ansatz in Eq. (\ref{eq:SE}), multiplying both sides of the resulting equation by $e^{-\hat{T}}$,  
using the Baker-Campbell-Hausdorff (BCH) commutator formula to only first order in $\hat{T}$, and expressing $\hat{H}$ in its normal-ordered form, $\hat{H}_{\mathcal{N}}$, where $\hat{H} =\hat{H}_{\mathcal{N}} +  E_{\text{HF}}$ with $E_{\text{HF}} = \bra{\varphi_0} \hat{H} \ket{\varphi_0}$, we find: 

\begin{align}
(\hat{H}_{\mathcal{N}} + E_{HF} + (\hat{H}_{\mathcal{N}} \hat{T})_c) \ket{\varphi_0}= E\ket{\varphi_0}, \label{eq:intermed1}
\end{align} 

which, after projecting both sides to the various one-body and two-body excited functions, $\bra{\chi_{l}}$, yields the so-called SRLCCSD amplitude equations: 

\begin{equation}
\boxed{
\begin{aligned}
\sum_{1n} \bra{\chi_{l}} \hat{H}_{\mathcal{N}} \ket{\chi_{1n}} t_{1n}
&+ \sum_{2m} \bra{\chi_{l}} \hat{H}_{\mathcal{N}} \ket{\chi_{2m}} t_{2m} \\
&= -\bra{\chi_{l}} \hat{H}_{\mathcal{N}}\ket{\varphi_0}, \quad \forall l.
\end{aligned}
}
\label{eq:srlccsdampeqs}
\end{equation} 

The subscripts `1' and `2' indicate singly and doubly excited functions arising out of $\ket{\varphi_0}$. Furthermore, $\ket{\chi_{1n}}$ and $\ket{\chi_{2m}} \in \{\ket{\chi_{l}}\}$. 

\begin{center}
    \textit{b. SRLCCSD correlation energy} 
\end{center} 

The correlation energy, $E^{SR}_{\text{corr}} = E^{\text{SR}} - E_{\text{HF}}$, is calculated by projecting both sides of Eq. (\ref{eq:intermed1}) onto $\bra{\varphi_0}$ and knowing that $\bra{\varphi_0} \hat{H}_{\mathcal{N}} \ket{\varphi_0}=0$ as: 

\begin{equation}
\boxed{
\begin{aligned}
E^{SR}_{\text{corr}}
&= \bra{\varphi_0} (\hat{H}_{\mathcal{N}}\hat{T})_{\text{c}} \ket{\varphi_0} \\
&= \bra{\varphi_0} \left(
\sum_{1n}  \hat{H}_{\mathcal{N}} \ket{\chi_{1n}} t_{1n}
+ \sum_{2m} \hat{H}_{\mathcal{N}} \ket{\chi_{2m}} t_{2m}
\right).
\end{aligned}
}
\label{eq:SRLCCEcorr}
\end{equation} 

Furthermore, the expression for the total energy is 

\begin{equation}
\boxed{
\begin{aligned}
E^{SR}
&= \bra{\varphi_0} \left(
\sum_{1n}  \hat{H}_{\mathcal{N}} \ket{\chi_{1n}} t_{1n}
+ \sum_{2m} \hat{H}_{\mathcal{N}} \ket{\chi_{2m}} t_{2m} 
\right)\\ +& E_{HF}.
\end{aligned}
}
\label{eq:SRLCCE}
\end{equation} 
\begin{center}
    \textit{c. Connection between QLS and SRLCC} 
\end{center} 

The SRLCCSD amplitude equations given in Eq. (\ref{eq:srlccsdampeqs}) can be compactly written as: 

\begin{equation}
    \sum_{\eta} \bra{\zeta} \hat{H}_{\mathcal{N}} \ket{\eta} t_\eta = -\bra{\zeta} \hat{H}_{\mathcal{N}} \ket{\varphi_0}, \forall \zeta. \label{eq:srlccsdlin}
\end{equation} 

Since the number of indices in $\eta$ and $\zeta$ match, the left hand side of Eq. (\ref{eq:srlccsdlin}) represents the action of a square matrix ($A$) on a column vector ($\vec{x}$ composed of amplitudes $ t_\eta$s, which we solve for), resulting in another column vector ($\vec{b}$), that is, a system of linear equations. 

Analogously, the SRLCCSD correlation energy expression in Eq. (\ref{eq:SRLCCEcorr}) can also be written in the following compact manner: 

\begin{equation}
 E^{SR}_{\text{corr}} =\sum_{\eta}  \bra{\varphi_0} \hat{H}_{\mathcal{N}} \ket{\eta} t_\eta. \label{eq8}
\end{equation} 

The amplitudes $t_\eta$ are obtained from $\vec{x}$ by $A^{-1} \vec{b}$ and the terms $\bra{\varphi_0} \hat{H}_{\mathcal{N}} \ket{\eta}$ are elements of $-\vec{b}^{\dagger}$. As a result, the expression of $E^{SR}_{\text{corr}}$ is: 

\begin{equation}
    E^{SR}_{\text{corr}} = -\vec{b}^\dagger A^{-1} \vec{b} = -\vec{b}^\dagger \cdot \vec{x}. 
\end{equation} 

This form for the correlation energy, expressed as an overlap between the input and output vectors, makes it straightforward to extract the quantity. For our numerical calculations, we use the Hong-Ou-Mandel (HOM) circuit \cite{swaptestHOM}, an ancilla-free version of the CSWAP test \cite{Buhrman2001}, but which involves measuring the state registers, to evaluate the absolute value of the overlap between $\ket{b}$ and $\ket{x}$. The HOM circuit computes $E^{SR}_{\text{corr}}$ as 

\begin{equation}
E^{SR}_{\text{corr}}= - \|\vec{x}\| \|\vec{b}\|^2 |\langle x | b \rangle|, \label{eq21}
\end{equation} 
where $|\langle x | b \rangle|$ is calculated from $
   | \langle x | b \rangle|^2 = \sum_{\alpha, \beta \in \{0,1\}^{n_b}} (-1)^{\alpha \cdot \beta} P(\alpha \cdot \beta)$. $\|\vec{x}\| \|\vec{b}\|^2$ appears because of normalization factors coming from $\ket{x}$ and $\ket{b}$. $\alpha \cdot \beta$ and $P(\alpha \cdot \beta)$ are respectively the bit-wise AND operation between the two $n_b$-length bitstrings obtained after measuring the qubits involved in the HOM circuit and the probability of obtaining $\alpha \cdot \beta$. 

\subsection{Internally contracted multi-reference linearized coupled cluster method with QLS}\label{subsec:qls-mrlcc} 

Just as in the single-reference case, we will only derive the working equations for icMRLCCSD \cite{Shamasundat2011, Evangelista2011, Andreas2011}, and the underlying ideas discussed can be carried over to situations where higher rank excitations are included. 

\begin{center}
    \textit{a. The icMRLCCSD ansatz} 
\end{center} 

The ansatz for the icMRCCSD case is similar in form to the single-reference counterpart, and is given as: 

\begin{equation}
\ket{\Psi} = e^{\hat{T'}} \ket{\psi_{0}}\label{eq1yoyo}. 
\end{equation} 

We immediately see the striking similarity in form with the single-reference case (which is also what motivated us to pick this variant of MRLCC equations to solve using QLS), except that the reference state is now no longer the HF state, $\ket{\varphi_0}$, but $\ket{\psi_{0}}$, which is a linear combination of many determinants. We shall expound on this aspect as well as on the operator $\hat{T'}$ before delving into the amplitude equations and the expression for correlation energy. 

\begin{center}
    \textit{b. The reference state} 
\end{center} 

We expand on the details of the reference state: 

\begin{enumerate} 
\item We start by partitioning the orbital space in the following way: 
\begin{itemize}
\item core orbitals, which are always doubly occupied in all reference determinants (indexed by $i, j, \cdots$), 
\item active orbitals, which have partial occupancies (indexed by $r,s, \cdots $), and 
\item virtual orbitals, which are always empty in all reference determinants ($a,b, \cdots $). 
\end{itemize} 
We will also refer to core and virtual orbitals as non-active orbitals. 
\item We perform a complete active space self-consistent field (CASSCF) calculation to identify the reference determinants in $\mathcal{U}$. 
\item We enumerate all possible configurations in the space of active orbitals. They form the set of reference determinants $\{\ket{\varphi_\alpha}\}$. We define the two idempotent, self-adjoint, and mutually exclusive projection operators that we shall use subsequently, $\hat{U}$ and $\hat{V}$, 
where $\hat{U} = \sum_\alpha U_\alpha= \sum_\alpha \ket{\varphi_\alpha} \bra{\varphi_\alpha}$ and $\hat{V} = \mathbb{1} - \hat{U} = \sum_\beta \ket{\varphi_\beta} \bra{\varphi_\beta}.$ $\{\ket{\varphi_\alpha}\}$ and $\{\ket{\varphi_\beta}\}$ are orthonormal configurations spanning the spaces $\mathcal{U}$ and $\mathcal{V}$ of the reference determinants and excited functions respectively. Furthermore, we note that an arbitrary state in $\mathcal{U}$ can be expressed as $\ket{\psi_\gamma^0} = \sum_{\alpha} C_{\alpha\gamma}\ket{\varphi_\alpha}$. 
\item Our reference function, $\ket{\psi_{\text{0}}}$, is taken to be $\ket{\psi_1^0} =  \frac{1}{N_a!} \sum_\alpha C_{\alpha 1} \ket{\varphi_\alpha} = \frac{1}{N_a!} \sum_\alpha C_{\alpha}  \ket{\varphi_\alpha}$, which is the state corresponding to the lowest eigenvalue obtained from the CASSCF calculation. $N_a$ here represents the number of active electrons. 
\end{enumerate} 

\begin{center}
    \textit{c. The excitation operator} 
\end{center} 

Having defined the reference state, we now return to Eq. (\ref{eq1yoyo}) to describe the excitation operator, $\hat{T'}$. In the icMRCCSD approach, it is given as $\hat{T'} = \hat{T'_1} + \hat{T'_2} = \sum_{IA} t_I^A \hat{a}_A^\dagger \hat{a}_I + \frac{1}{4}\sum_{IJAB} t_{IJ}^{AB} \hat{a}_A^\dagger \hat{a}_B^\dagger \hat{a}_I \hat{a}_J$, where $I$ and $J$ run over the union of core and active spin-orbitals, while $A$ and $B$ run over the union of virtual and active spin-orbitals. We denote the number of active spin-orbitals by $n_a$ and the number of virtual spin-orbitals by $n_v$. We omit active-to-active excitations \cite{Andreas2011} by construction and keep the coefficients in the CASSCF wavefunction expansion fixed. 

\begin{center}
    \textit{d. icMRLCCSD amplitude equations} 
\end{center} 

By substituting the ansatz from Eq. (\ref{eq1yoyo}) in Eq. (\ref{eq:SE}), we arrive at: 
\begin{equation}
\hat{H} e^{\hat{T'}}\ket{\psi_{0}} = E e^{\hat{T'}} \ket{\psi_{0}}. \label{eq10}
\end{equation} 

We pre-multiply both sides of Eq. (\ref{eq10}) with $e^{-\hat{T'}}$, followed by projection of both sides of the equation onto each function contained in the space $\mathcal{V}$, to obtain the following equations: 

\begin{equation}
    \bra{\varphi_\beta} e^{-\hat{T'}} \hat{H} e^{\hat{T'}} \ket{\psi_{0}} = 0 ,\forall \beta. \label{eq11}
\end{equation} 

We immediately see that some excitation operators in $\hat{T'}$ could lead to the same excited function when applied to different determinants that constitute $\ket{\psi_{0}}$. For illustrating this redundancy problem, we consider the example of $\ket{\psi_{0}}$ being a linear combination of only two reference determinants $\ket{\mathrm{(core)}r}$ and $\ket{\mathrm{(core)}s}$, $\ket{\psi_{0}} = c_r \ket{\mathrm{(core)}r} + c_s \ket{\mathrm{(core)}s}$, where $\mathrm{(core)}$ signifies a set of core orbitals which is common to all reference determinants. $r$ and $s$ are singly occupied active orbitals. The action of excitation operators $\hat{a}_t^\dagger \hat{a}_r$ and $\hat{a}_t^\dagger \hat{a}_s$ on $\ket{\psi_{0}}$ lead to the same excitation function $\ket{\mathrm{(core)}t}$. That is, both $\hat{a}_t^\dagger \hat{a}_s$ and $\hat{a}_t^\dagger \hat{a}_r$ are redundant operators, meaning they generate linearly dependent excitation functions ($(c_s \hat{a}_t^\dagger \hat{a}_r - c_r\hat{a}_t^\dagger \hat{a}_s) \ket{\psi_0} = 0$). We project Eq. \ref{eq11} only on linearly independent excitation functions, $\{\ket{\varphi_\beta}\}$. This requires all redundant excited functions to be removed. There is more than one way to remove redundancy, and we discuss the approach that we follow in Section \ref{sec:preproc-scaling}. 

Upon expanding $e^{-\hat{T'}} \hat{H} e^{\hat{T'}}$ using the BCH commutator formula and retaining only the terms that are linear in $\hat{T'}$, we obtain $e^{-\hat{T'}} \hat{H} e^{\hat{T'}} = \hat{H} + [\hat{H}, \hat{T'}]$. Upon substituting the expression in Eq. \ref{eq11}, we get: 

\begin{equation}
    \bra{\varphi_\beta} \hat{H} + [\hat{H}, \hat{T'}] \ket{\psi_{0}} = 0, \forall \beta. 
\end{equation} 

Next, we insert the resolution of identity $\mathbb{1} = \sum_{\ket{\psi^0_\gamma} \in \mathcal{U}} \ket{\psi^0_\gamma}\bra{\psi^0_\gamma} +  \sum_{\ket{\varphi_r} \in \mathcal{V}} \ket{\varphi_r}\bra{\varphi_r}$ in the previous equation and get the following equation: 

\begin{eqnarray}
\bra{\varphi_\beta} \hat{H} \ket{\psi_{0}} &+& \sum_{\ket{\psi^0_\gamma} \in \mathcal{U}}  \bra{\varphi_\beta} \hat{H} \ket{\psi^0_\gamma}\bra{\psi^0_\gamma} \hat{T'} \ket{\psi_{0}} \nonumber\\
&+& \sum_{\ket{\varphi_r} \in \mathcal{V}} \bra{\varphi_\beta} \hat{H} \ket{\varphi_r} \bra{\varphi_r} \hat{T'} \ket{\psi_{0}}  \nonumber \\
&-&  \sum_{\ket{\psi^0_\gamma} \in \mathcal{U}}  \bra{\varphi_\beta} \hat{T'} \ket{\psi^0_\gamma}\bra{\psi^0_\gamma} \hat{H}\ket{\psi_{0}} \label{eq13} \\
&-& \sum_{\ket{\varphi_r} \in \mathcal{V}} \bra{\varphi_\beta} \hat{T'} \ket{\varphi_r} \bra{\varphi_r} \hat{H} \ket{\psi_{0}} =0, \forall \beta. \nonumber 
\end{eqnarray} 

Since we do not consider excitations from active to active orbitals, the action of $\hat{T'}$ on $\ket{\psi_{0}}$ can never produce another function in the space $\mathcal{U}$, therefore, the second term of Eq. (\ref{eq13}) vanishes. On the other hand, the action of $\hat{H}$ on $\ket{\psi_{0}}$ can produce functions of the type $\ket{\psi^0_\gamma}$, therefore, $\bra{\psi^0_\gamma} \hat{H} \ket{\psi_{0}} = \delta_{\gamma, 0} E_{0}$ where  $E_{0} = \bra{\psi_{0}} \hat{H} \ket{\psi_{0}}$ is the CASSCF ground-state energy. Taking these last remarks into consideration, Eq. (\ref{eq13}) reduces to: 

\begin{eqnarray}
\bra{\varphi_\beta} \hat{H} \ket{\psi_{0}} &+& \sum_{\ket{\varphi_r} \in \mathcal{V}} \bra{\varphi_\beta} \hat{H} \ket{\varphi_r} \bra{\varphi_r} \hat{T'} \ket{\psi_{0}}  \nonumber \\
&-&   \bra{\varphi_\beta} \hat{T'} \ket{\psi_{0}} E_{0} \label{eq14} \\ &-& \sum_{\ket{\varphi_r} \in \mathcal{V}} \bra{\varphi_\beta} \hat{T'} \ket{\varphi_r} \bra{\varphi_r} \hat{H} \ket{\psi_{0}} =0, \forall \beta.  \nonumber \label{eq:yo}
\end{eqnarray} 

Next, we make the approximation of neglecting the fourth term at the left hand side of Eq. (\ref{eq14}), $\sum_{\ket{\varphi_r} \in \mathcal{V}} \bra{\varphi_\beta} \hat{T'} \ket{\varphi_r} \bra{\varphi_r} \hat{H} \ket{\psi_{0}}$. 
Eq. (\ref{eq14}), therefore becomes: 

\begin{eqnarray}
 \sum_{\ket{\varphi_r} \in \mathcal{V}} \bra{\varphi_\beta} \hat{H} \ket{\varphi_r} \bra{\varphi_r} \hat{T'} \ket{\psi_{0}} &-& \bra{\varphi_\beta} \hat{T'} \ket{\psi_{0}} E_{0} \nonumber \\ &=& - \bra{\varphi_\beta} \hat{H} \ket{\psi_{0}}, \forall \beta.  \nonumber \label{eq15}
\end{eqnarray} 

Thus, the icMRLCCSD amplitude equations can also be expressed as: 

\begin{equation}
\boxed{
\begin{split}
\sum_{\ket{\varphi_r} \in \mathcal{V}}
(\bra{\varphi_\beta} \hat{H} \ket{\varphi_r}
 - E_{0}\delta_{\beta, r})
\bra{\varphi_r} \hat{T'} \ket{\psi_{0}}
\\
= - \bra{\varphi_\beta} \hat{H} \ket{\psi_{0}},
\quad \forall \beta .
\end{split}
}
\label{eq16}
\end{equation}

This constitutes a linear system of equations, which has as unknowns the excitation amplitudes, $\bra{\varphi_r} \hat{T'} \ket{\psi_{0}}$, while the matrix elements of the $A$ matrix are $\bra{\varphi_\beta} \hat{H} \ket{\varphi_r} - E_{0}\delta_{\beta, r}$. The vector, $\vec{b}$, is made of elements $ - \bra{\varphi_\beta} \hat{H} \ket{\psi_{0}}, \forall \beta$. We make an important clarification at this point: strictly speaking, the governing equations are those for the MRCEPA(0) method (acronym for multi-reference unshifted coupled electron pair approximation), which is very similar yet slightly different from the icMRLCCSD method (see Ref. \cite{JoshuaB2019}), however, for the purposes of this work, we do not make that distinction and use the term QLS-icMRLCC to denote the quantum linear solver formulation employing this MRCEPA(0) approximation. In the single-reference case, the equivalence LCCSD=CEPA(0) is well-established \cite{Ivanic2018}. We also performed preliminary calculations in which the last term in Eq. (\ref{eq:yo}), originating from the $T'H$ contribution to the commutator, was retained. In the cases investigated, the resulting $A$ matrices were non-Hermitian (because of the non-Hermiticity of excitation operators) and exhibited isolated sharp increases in their computed $\kappa$ values for some test systems as the number of virtual orbitals was varied. In contrast, the approximation adopted in Eq. \ref{eq16} yields a Hermitian $A$ matrix. A systematic investigation of the spectral properties of the retained-term formulation is beyond the scope of the present work.

\begin{center}
    \textit{e. icMRLCCSD correlation energy} 
\end{center} 

Having obtained the amplitudes by solving the amplitude equations, we now move to the equations that we use to calculate the icMRLCCSD energy. If we take Eq. (\ref{eq10}), multiply both sides by $e^{-\hat{T'}}$, and then project them onto $\bra{\psi_{0}}$, we obtain the following equation for the total energy: 

\begin{eqnarray}
 E^{\text{icMR}} = \bra{\psi_{0}} e^{-\hat{T'}} \hat{H} e^{\hat{T'}} \ket{\psi_0}.  \nonumber\label{ecorr:mr}
\end{eqnarray}  

After using the BCH commutator formula to expand $e^{-\hat{T'}} \hat{H} e^{\hat{T'}}$ up to linear powers of $\hat{T'}$, then expanding further the commutator $[\hat{H},\hat{T'}]$, we get: 
\begin{equation}
 E^{\text{icMR}} = \bra{\psi_{0}} \hat{H} \ket{\psi_{0}}+ \bra{\psi_{0}} \hat{H}\hat{T'} \ket{\psi_{0}} - \bra{\psi_{0}} \hat{T'}\hat{H} \ket{\psi_{0}}. \label{e:mr}
\end{equation} 
where we have substituted $E_0$ by $\bra{\psi_{0}} \hat{H} \ket{\psi_{0}}$. The term $\bra{\psi_{0}} \hat{T'}\hat{H} \ket{\psi_{0}}  = 0 $ because $\bra{\psi_{0}} \hat{T'} = \left(\hat{T'}^\dagger \ket{\psi_0}\right)^\dagger = 0 $. Now, introducing the resolution of identity $ \mathbb{1} = \sum_{\alpha} \ket{\varphi_\alpha} \bra{\varphi_\alpha} + \sum_{\beta} \ket{\varphi_\beta} \bra{\varphi_\beta}$ in Eq. (\ref{e:mr}), we obtain: 

\begin{eqnarray}
E^{\text{icMR}} = E_{0} &+&\sum_\alpha \bra{\psi_{0}} \hat{H} \ket{\varphi_\alpha} \bra{\varphi_\alpha} \hat{T'} \ket{\psi_{0}} \nonumber \\ 
&+& \sum_\beta \bra{\psi_{0}} \hat{H} \ket{\varphi_\beta} \bra{\varphi_\beta} \hat{T'} \ket{\psi_{0}}. \nonumber
\end{eqnarray} 

In this expression, the terms including $\bra{\varphi_\alpha} \hat{T'} \ket{\psi_{0}}$ are 0 because the action of $\hat{T'}$ on $\ket{\psi_{0}}$ produces functions contained in $\mathcal{V}$. Thus, the final expression of the total energy in the icMRLCCSD framework is: 

\begin{equation}
\boxed{
E^{\text{icMR}} = E_{0} + \sum_\beta \bra{\psi_{0}} \hat{H} \ket{\varphi_\beta} \bra{\varphi_\beta} \hat{T'} \ket{\psi_{0}}.} \label{e19}
\end{equation} 

The correlation energy expression is obtained by subtracting both sides of Eq. (\ref{e19}) by $E_{HF}$, thus obtaining the following expression: 

\begin{equation}
\boxed{
E^{\text{icMR}}_{\text{corr}}
= E_0 - E_{HF} + \sum_\beta \bra{\psi_{0}} \hat{H} \ket{\varphi_\beta}
\bra{\varphi_\beta} \hat{T'} \ket{\psi_{0}}.
}
\label{eq20}
\end{equation} 

\begin{center}
    \textit{f. Connection between QLS and icMRLCC} 
\end{center} 

Similar to the single-reference case, Eq. (\ref{eq16}) can be written in the compact form: 

\begin{equation}
\sum_{\eta}  (\bra{\zeta}\hat{H}\ket{\eta} - E_{0}\delta_{\eta,\zeta} ) t_\eta = -\bra{\zeta} \hat{H}\ket{\psi_{0}}, \forall \zeta, \label{eq19}
\end{equation} 

where $\bra{\zeta}\hat{H}\ket{\eta} - E_{0}\delta_{\eta,\zeta} $ represents a matrix element (between $\ket{\eta}$ and $\ket{\zeta}$) of the Hamiltonian with $E_{0}$ subtracted from its diagonal. $\ket{\eta}$ and $\ket{\zeta}$ represent $\ket{\varphi_r}$ and $\ket{\varphi_\beta}$ respectively, while $t_\eta$s are the excitation amplitudes, $\bra{\varphi_r} \hat{T'} \ket{\psi_{0}}$. Since we have as many $\ket{\varphi_\beta}$ as $\ket{\varphi_r}$ functions, we also have as many $\ket{\eta}$ as we have $\ket{\zeta}$ functions, and thus the left hand side of Eq. (\ref{eq19}) represents $A\vec{x}$ and its right-hand side, $-\bra{\zeta} \hat{H} \ket{\psi_{0}}$, are the elements of the $\vec{b}$ vector. In short, Eq. (\ref{eq19}) represents the equation $A\vec{x} = \vec{b}$ to be solved by a QLS upon suitable encoding. Analogously to how the correlation energy is calculated for the single-reference case, Eq. (\ref{eq20}) suggests $E_{{\text{corr}}}^{\text{icMR}} = E_0 -E_{HF} + \sum_\beta \bra{\psi_{0}} \hat{H} \ket{\varphi_\beta} \bra{\varphi_\beta} \hat{T'} \ket{\psi_{0}}$ , where $\sum_\beta \bra{\psi_{0}} \hat{H} \ket{\varphi_\beta} \bra{\varphi_\beta} \hat{T'} \ket{\psi_{0}}$ is a dot product between $-\vec{b}^\dagger$ and $\vec{x}$ vector, and since $\vec{x}$ is nothing but $A^{-1} \vec{b}$, the correlation energy expression in terms of $\vec{x}$ and $\vec{b} $ is therefore: 

\begin{equation}
    E_{{\text{corr}}}^{\text{icMR}} = E_0 -E_{HF} -\vec{b}^\dagger A^{-1} \vec{b}. 
\end{equation} 

This energy can also be calculated using the HOM circuit as we saw in Eq. (\ref{eq21}). For completeness, we add that we will in the later sections find $A$ for the icMRLCC case to not necessarily be positive definite. In such cases, $\vec{b}^\dagger A^{-1} \vec{b}$ need not be positive and recovering the correct sign would involve going beyond the traditional HOM module, which we shall defer for a future study. 

\subsection{The role of QLS-LCC in quantum chemistry on quantum computers}\label{subsec:lit} 

The QPE algorithm (as well as its variants) is widely regarded as one of the most promising scalable fault tolerant era quantum algorithms for quantum chemistry applications \cite{yamamoto2026, kitaevqpe1995, QPE_Sethlloyd, Abrams1997, Cleve_1998, Azpuru2005, Lanyon2010-pa, Miroslav2007, SmithJoseph2022, yamamoto_bayesian_qpe, Hardware_QPE, reductive_qpe}. In particular, QPE is typically used to solve the complete active space configuration interaction (CASCI) equations. Although VQE ruled the algorithms landscape for chemistry in the past decade, there has been renewed interest in QPE due to recent advances on quantum hardware, including proof-of-concept QPE-CASCI calculations on quantum hardware in a partial fault-tolerant manner, underscoring the algorithm's importance for quantum chemistry in the years ahead \cite{yamamoto_bayesian_qpe, yamamoto2026}. However, despite all this progress and the algorithm's promise of exponential speed-up over the best known classical methods, we note that 
\begin{itemize}
\item QPE typically requires a number of state register qubits that is equal to the number of spin-orbitals included in the active space chosen. Consequently, it is naturally suited to treating only the most important (chemically relevant) orbitals in a calculation. 
\item CASCI itself, while being the FCI solution within an active space, does not capture dynamical correlation arising from the rest of the orbital space. 
\end{itemize} 

These observations suggest a possible role for QLS-LCC; while LCC within the same orbital space as CASCI is an approximation to FCI and hence less accurate relative to CASCI, it can be executed on the entire (or a large part of the) orbital space and thus can capture dynamical correlation that CASCI cannot. Furthermore, QLS-LCC admits state register sizes that are logarithmic in the number of rows of $A$ due to amplitude encoding, thus allowing one to work with larger orbital spaces. A natural approach could be to propose a combination of the high-accuracy QPE-CASCI solver within a reduced orbital subspace with the QLS-LCC within a bigger orbital space. In particular, one could treat only the most important small set of orbitals using QPE on a CASCI Hamiltonian, while describing the remaining external subspace with QLS-LCC. Additionally, the LCC treatment could be SR or icMR depending on the situation at hand. We leave detailed investigations of such embedding techniques for future studies. 

\section{Diagnostics for $\kappa$ scaling}\label{sec:scaling} 

In this section, we describe the methodology used to investigate the dependence of $\kappa$ on system size in the numerical calculations we present in the next two sections. 

As a precursor to our discussion, we note that favourable $\kappa$ scaling appears to be relatively uncommon in search for applications where QLSs could be potentially useful. The numerous attempts in literature typically involve investigating specific problem instances within an application rather than an entire application, reflecting the inherent difficulty in making such general assessments of $\kappa$ scaling. Furthermore, across problem instances in diverse applications, including the DC and AC power flow problems \cite{pareek_2026, pareek_2026_2}, computational Earth sciences (fracture types in hydrological systems) \cite{Golden2022}, portfolio optimization problem (data from the arXiv version (v4) of Ref. \cite{Yalovetzky2024}), and computational fluid dynamics (Hele-Shaw problem) \cite{Gopal2024}, $\kappa$ is found to grow at least polynomially (sometimes even exponentially) with system size. To the best of our knowledge, the only notable exception was recently reported by some of the present authors (Ref. \cite{Peniel2025}), where it was found that the quantity grows as a polylogarithmic function of the number of spin-orbitals for KH, RbH, and CsH in the Sapporo basis sets. It is worth noting that they choose to increase the virtual orbital space not arbitrarily but via orbital character analysis. They also argue on grounds of chemical intuition that the slow growth of $\kappa$ is driven by the saturation in the decrease of the smallest eigenvalue of $A$. These observations provide us with additional motivation for a careful examination of $\kappa$ scaling in the broader problem of QLS-LCC. 

We estimate $\kappa$ scaling using three different diagnostics, direct and indirect, as we shall detail below: 

\subsection{Direct computation of $\kappa$} \label{aubaec:kappadir} 

For each molecule, we extract the $A$ matrix, the details of which are provided in Section \ref{sec:preproc-scaling}, at each system size. We go from one system size to the next by adding one virtual orbital at a time. At this juncture, it is worth clarifying the notion of system size used throughout this work. Although one specifies system size in terms of the dimension of the linear system, $N$, for the purposes of formal complexity analyses, numerical scaling studies in quantum chemistry require a physically meaningful parameter(s) through which $N$ is varied. A reasonable parameter for this purpose is the number of orbitals, which in practice is the number of single particle basis set elements (to within a factor of 2). For large molecular systems, growing the system size would involve increasing both occupied and virtual spin-orbitals at each increment. Determining $N$ scaling in terms of number of orbitals as well as expressing $s$ in terms of the same quantity is straightforward as we could analytically arrive at such expressions. However, given the computational cost of obtaining $\kappa$, we do not attempt to study $\kappa$ scaling across different numbers of occupied orbitals; we always consider all occupied orbitals for convenience and vary only the number of virtual orbitals for each molecule considered, and study the behaviour of condition number as we increase the number of virtuals. Our motivation for fixing a molecule and varying orbitals rather than moving across different systems stems not only from the inherent difficulty in what moving across different systems means in a general case (outside of chains, where this is natural) but also from our intended applications of QLS-LCC. As seen in the earlier section, we propose that QLS-LCC may serve as complementary treatment to high accuracy QPE-CASCI. In such cases, one might be interested in predicting properties of a target molecule with increasing accuracy, where improving the chemical description of the problem in question would typically involve enlarging the virtual space. In such situations, it is natural to assess $\kappa$ and $s$ scaling using the number of spin-orbitals. Having said that, we also investigate molecular chains, for which increasing the number of atoms provides a natural definition of increasing the system size. 

We note that the virtuals are arranged according to the irreducible representation (irrep) that they belong to. At each system size, we diagonalize the matrix to then compute the condition number. This is followed by finding a suitable fit function to the data points. We note that the full diagonalization of an $N \times N$ Hermitian matrix involves $\mathcal{O}(N^3)$ arithmetic operations. We also add that it is not just sufficient to find only the extremal eigenvalues when a Hermitian matrix is indefinite when the quantity of interest is the condition number, as the relevant quantity in the denominator in the expression of $\kappa$ would involve an eigenvalue closest to zero. 

\subsection{Indirect computation of $\kappa$ via an upper bound} \label{aubaec:kappaindir1} 

It would be ideal to substitute $\kappa$ with some other suitable proxy that is relatively easier to compute. We find that for the specific matrices that we are interested in the context of SRLCC, $d_{\mathrm{max}}\big/ d_{\mathrm{min}}$ can be one such quantity. Here, $d_{\mathrm{max}}$ and $d_{\mathrm{min}}$ are the maximum and minimum diagonal entries of $A$ respectively. We derive an upper bound on $\kappa$ in terms of $d_{\mathrm{max}}\big/ d_{\mathrm{min}}$.  We put forth two theorems, one that assumes a diagonally dominant matrix, and the other which relaxes the assumption. 

\begin{theorem}\label{theorem:1}
Let $A\in\mathbb{R}^{N\times N}$ be a real symmetric matrix with positive diagonal entries, that is, $a_{ii} = d_i > 0$. Let
\[
\upepsilon:=\max_i \frac{R_i}{d_i}<1, 
\]
where $R_i := \sum_{j\neq i} |a_{ij}|$, and the condition number of $A$, $
\kappa(A)
=
\frac{\lambda_{\max}(A)}{\lambda_{\min}(A)}.
$ Then, 
\[
\kappa(A)\le
\frac{d_{\max}(1+\upepsilon)}
     {d_{\min}(1-\upepsilon)}.
\]
\end{theorem} 

\begin{proof} 
Since $\upepsilon:=
\max_i \frac{R_i}{d_i} < 1$, we have $R_i \leq \varepsilon d_i \ \forall i$. Invoking the Gershgorin circle theorem, we see that 

\begin{eqnarray}
\lambda_{\mathrm{max}} &\leq& d_k + R_k \leq d_k + \upepsilon d_k \nonumber \\ 
&=& d_k(1+\upepsilon) \leq d_{\mathrm{max}}(1+\upepsilon)
\end{eqnarray} 
and 
\begin{eqnarray}
\lambda_{\mathrm{min}} &\geq& d_m - R_m \geq d_m - \upepsilon d_m \nonumber \\ 
&=&d_m(1-\upepsilon) \geq d_{\mathrm{min}}(1-\upepsilon). 
\end{eqnarray}

Thus, $\kappa \leq \frac{d_{\mathrm{max}}(1+\upepsilon)}{d_{\mathrm{min}}(1-\upepsilon)}$. 
\end{proof} 

\begin{remark}
The condition on positive diagonal entries is an empirical observation for the SRLCC matrices from our data. 
\end{remark} 

\begin{remark}
For diagonal matrices, the bound is exact since $\kappa=d_{\max}\big/d_{\min}$.
\end{remark} 

\begin{remark} 
Noting that $\frac{1+\upepsilon}{1-\upepsilon}
=
1+2\upepsilon+\frac{2\upepsilon^2}{1-\upepsilon}$, we have for strongly diagonally dominant matrices, where $0 \leq \varepsilon \leq 1/2$, the condition $\frac{1+\upepsilon}{1-\upepsilon} \leq 3$. Thus, in such a regime, $\kappa(A)
\leq 3\frac{d_{\max}}{d_{\min}}$.  
\end{remark} 

The assumption of diagonal dominance, in the absence of which the theorem breaks down, has to be assessed numerically with data. For the four molecules that we consider (both in SRLCC and icMRLCC situations, and for $\mathfrak{E}=2,3,4$), we observe that the condition $a_{ii} > R_i \ \forall i$ is not satisfied for large $n_v(n_{vir})$ cases, as Fig. \ref{Aii-Ri_srmr:lihbeh2} of the Appendix shows, where we plot $\mathrm{min}_i\ a_{ii}-R_i$ for two representative systems (LiH for the SRLCCSD case and BeH$_2$ for the icMRLCCSD scenario). Thus, the theorem presented above, though interesting, is not satisfied for all matrices considered. To that end, we discuss in the theorem presented below the validity of $d_{\max}\big/d_{\min}$ as a proxy to $\kappa$ without including diagonal dominance, although we shall see soon that the contributions (or lack of them) from off-diagonal entries still matter in another form. 

Before proceeding further, we make one observation. We could force diagonal dominance by introducing a sufficiently large shift to the diagonal entries of $A$. However, determining the magnitude of such a shift necessitates evaluating $R_i$ for all $N$ rows. Furthermore, it changes the linear systems problem, and a bound on a shifted matrix does not establish the desired bound for the original matrix. Thus, we avoid that approach to justifying Theorem \ref{theorem:1}. 

\begin{theorem}\label{theorem:2}
Let $A=D+V$, where $A\in\mathbb{R}^{N\times N}$ is symmetric, $D=\mathrm{diag}(d_1,\ldots,d_N)$, and $V = A-D$. Define $d_{\min}:=\min_i d_i$ and $d_{\max}:=\max_i d_i$. Then, $d_{\min}-\|V\|
\le
\lambda_{\min}(A)
\le
d_{\min}+\|V\|$ and $d_{\max}-\|V\|
\le
\lambda_{\max}(A)
\le
d_{\max}+\|V\|$, where $\|V\|$ is the spectral norm of $V=A-D$. Consequently, if $\|V\| \ll d_{\min}$, then $\kappa(A)$ can be well-approximated by $d_{\max}\big/d_{\min}$. 
\end{theorem} 

\begin{proof}
Since $A=D+V$, the spectral stability corollary of Weyl's inequality implies that $|\lambda_i(A)-\lambda_i(D)|
\le
\|V\| \ \forall i$. Therefore, picking $i=1$ and $i=N$ for the extremal eigenvalues gives 
\begin{eqnarray}
|\lambda_{\min}(A)-\lambda_{\min}(D)|
\le
\|V\|
\end{eqnarray}
and 
\begin{eqnarray}
|\lambda_{\max}(A)-\lambda_{\max}(D)|
\le
\|V\|.
\end{eqnarray}

Using the fact that $|x| \leq y\ \iff\ -y \leq x \leq y$, and also that $\lambda_{\min}(D)=d_{\min}$ and $\lambda_{\max}(D)=d_{\max}$, we obtain 
\begin{eqnarray}
d_{\min}-\|V\|\leq \lambda_{\min}(A) \leq d_{\min}+\|V\|
\end{eqnarray}
and
\begin{eqnarray}
d_{\max}-\|V\|\leq \lambda_{\max}(A) \leq d_{\max}+\|V\|. 
\end{eqnarray}

Therefore, if $\|V\| \ll d_{\min}$, then we get $\lambda_{\min}(A)\approx d_{\min}$ and $\lambda_{\max}(A)\approx d_{\max}$. Therefore, $\kappa(A)
=
\frac{\lambda_{\max}(A)}
     {\lambda_{\min}(A)}
\approx
\frac{d_{\max}}
     {d_{\min}}$. 
\end{proof} 

\begin{remark}\label{remark:yo}
In the above theorem, there is no demand on the strict diagonal dominance of $A$. However, there is a condition on the norm of $V$ in relation to $d_{\min}$. Indeed, in the subsequent sections, we will encounter a small fraction of the systems considered where the diagonal ratio is not a good proxy for $\kappa$. 
\end{remark} 

\begin{remark}
The condition that $\|V\| \ll d_{\min}$ implicitly implies that $d_i>0, \forall i$. This is because by construction, $\|V\|$ has to be non-negative, implying that $d_{\min}$ has to be positive, therefore leading to all diagonal entries of $A$ being positive. 
\end{remark} 

\begin{remark}
A remark pertaining to the previous remark is that $d_i>0, \forall i$ alone does not imply positive definiteness of $A$. However, a stronger consequence of Theorem \ref{theorem:2} is that when $\|V\| \ll d_{\min}$, $\lambda_{\min}(A) \geq d_{\min}-\|V\| > 0$, thus demanding that $A>0$. 
\end{remark} 

The cost of determining $d_{\max}/d_{\min}$ is $N$ (by simple inspection; we also ignore the cost of certifying the norm condition). Computing this indirect diagnostic requires fewer operations than explicitly determining $\kappa$, but limited by the condition on $\|V\|$. 

\subsection{Indirect computation of $\kappa$ via a conjecture} \label{aubaec:kappaindir2} 

Finally, we employ another indirect method that is computationally much cheaper than the previous one, to predict $\kappa$ scaling with system size. In particular, we invoke the edge spawning conjecture that was proposed in Ref. \cite{shetty2025} for unweighted graph adjacency matrices. In particular, in our adaptation of the conjecture, the many-electron basis states correspond to the vertices of a graph, non-zero off-diagonal entries correspond to edges, and their numerical values correspond to edge weights. However, since the conjecture was formulated for unweighted graphs, we replace all the edge weights by 1 for this exploratory diagnostic. 

The data-driven edge spawning conjecture was proposed by looking for one common pattern across 30 different graph families' unweighted graph Laplacians, and in particular, the pattern formed by the non-zero elements of the matrix serve as a proxy for $\kappa$ scaling; if the matrix elements are scattered across the matrix in the so-called diffuse pattern, then it is conjectured that $\kappa$ grows polylogarithmically, while `sharp' patterns indicate fast $\kappa$ scaling. The conjecture operates under the implicit assumption that a diffuse pattern remains diffuse across all system sizes and that sharp remains sharp across all system sizes. 

If the conjecture is applicable to LCC matrices, then this approach to assessing $\kappa$ scaling only requires one to examine the diffuse or sharp patterns for a few small system sizes to obtain a qualitative indication of potentially favourable $\kappa$ scaling. Due to this factor, there is negligible scaling cost associated with this method when compared to explicit $\kappa$ computations or finding $d_{\max}\big/d_{\min}$, and therefore, one can check for patterns even up to LCCSDTQ level of theory. 

\subsection{Relation between $\kappa$ and spectral gap} 

One may wonder if the spectral gap of the normal ordered Hamiltonian could serve as a proxy to $\kappa$. We recall that the normal ordered Hamiltonian is a shifted Hamiltonian (in the SRLCC case, which we shall restrict ourselves to for simplicity, this is $H_{\mathcal{N}}$, as discussed in the earlier sections) and its principal sub-matrix is our $A$ matrix. Thus, one may be led to think (and reasonably so) that $\kappa(A)$, the condition number of $A$, could be related to the condition number of the shifted Hamiltonian, $\kappa(H_{\mathcal{N}})$, and therefore $\kappa$ might be connected to the spectral gap of the shifted Hamiltonian. We address this issue in the next paragraph. 

\begin{theorem}
Let $H_\mathcal{N}$ denote the shifted Hamiltonian, given by $H_\mathcal{N} = H - E_{\mathrm{HF}} I$, with eigenvalues $\Lambda_1 \leq \Lambda_2 \leq \cdots \leq \Lambda_{N+1}$. $H_{N}$ is represented in the space spanned by the HF reference and the LCC excitation functions at a chosen excitation rank. Let $A$ be the principal sub-matrix (LCC matrix) obtained from $H_\mathcal{N}$, by removing the Hartree-Fock reference row and column. Let the eigenvalues of $A$ be $\lambda_1 \leq \lambda_2 \leq \cdots \leq \lambda_N$. Assume that $A$ is positive definite, that is, $\lambda_1>0$. Then, $0 < \lambda_1 \le \Lambda_1 + \Delta(H_\mathcal{N})$, where $\Delta(H_\mathcal{N})=\Lambda_2-\Lambda_1$. As a consequence, $\kappa(A)=\frac{\lambda_N}{\lambda_1} \geq
\frac{\lambda_N}{\Lambda_1+\Delta(H_\mathcal{N})} \geq \frac{\Lambda_N}{\Lambda_1+\Delta(H_\mathcal{N})}$. 
\end{theorem} 

\begin{proof}
Since $A$ is a principal sub-matrix of $H_\mathcal{N}$, Cauchy’s interlacing theorem gives $\Lambda_1 \le \lambda_1 \le \Lambda_2$. Using the fact that $\Delta(H_\mathcal{N})=\Lambda_2-\Lambda_1$, we see that $\Lambda_2=\Lambda_1+\Delta(H_\mathcal{N})$ and thus $\Lambda_1 \leq \lambda_1 \leq \Lambda_1+\Delta(H_\mathcal{N})$. However, since $A$ is positive definite, $\lambda_1>0$, and hence $0<\lambda_1\le \Lambda_1+\Delta(H_\mathcal{N})$. 

It follows from here that $\frac{1}{\lambda_1} \geq \frac{1}{\Lambda_1+\Delta(H_\mathcal{N})}$. Multiplying both sides by $\lambda_N$, we obtain $\kappa(A) = \frac{\lambda_N}{\lambda_1} \geq \frac{\lambda_N}{\Lambda_1+\Delta(H_\mathcal{N})}$. We note that the denominator as well as the numerator are already positive, and thus is the same as its absolute value. 

Finally, applying the interlacing theorem at the upper end of the spectrum, $\Lambda_N \le \lambda_N \le \Lambda_{N+1}$, and thus $\kappa(A) \geq \frac{\Lambda_N}{\Lambda_1+\Delta(H_\mathcal{N})}$. 
\end{proof} 

The above relation that lower bounds $\kappa$ immediately indicates that knowledge of the gap alone is insufficient to determine or explain the growth of the condition number. Furthermore, even if a relation were to be established between $\kappa$ and $\Delta(H_\mathcal{N})$, computing the gap is as much effort as finding the extremal eigenvalues and thus, from a practical standpoint, the spectral gap is not a computationally cheaper proxy to $\kappa$. 

Our theorem also shows that knowledge of the condition number scaling of the Hamiltonian does not suffice to understand the condition number scaling of $A$. Our data shown in Figs. \ref{kappaA_kappaHsr:allmolecules} and \ref{kappaA_kappaHmr:chains} support this observation; the scaling of $\kappa(H_{\mathcal{N}})$ and that of the $A$ matrix are not necessarily correlated for both the model molecules as well as chains that we consider. Additionally, Figs. \ref{kappaA_kappaHmr:allmolecules} and \ref{kappaA_kappaHmr:h_chain} also shows that this behaviour is carried forward in the icMRLCC case, although we do not derive an analogous theorem for it and defer it for a later study. 

\begin{table*}[t]
\centering
\caption{Systems considered for our SRLCC and icMRLCC
calculations. All distances are specified in units of bohr. The notation cc-pV$\{D,T,Q\}$Z denotes the cc-pVDZ,
cc-pVTZ, and cc-pVQZ basis sets. All SRLCC calculations use $C_{2v}$ symmetry, while the icMRLCC calculations do not use symmetry. For the icMRLCC column, we also specify the reference active space chosen in the format ($N_a,n_a/2$), where $N_a$ is the number of active electrons and $n_a/2$ is the number of active orbitals considered. } 
\label{tab:systems}
\renewcommand{\arraystretch}{1.25}
\setlength{\tabcolsep}{5pt}
\begin{tabular}{p{0.10\textwidth}
                p{0.42\textwidth}
                p{0.38\textwidth}}
\hline
\textbf{System} & \textbf{SRLCC} & \textbf{icMRLCC} \\
\hline
Be &
cc-pV$\{D,T,Q\}$Z & cc-pVQZ, (2,3) 
 \\
Ne &
cc-pV$\{D,T,Q\}$Z & cc-pVQZ, (2,2) 
 \\
LiH &
cc-pV$\{D,T,Q\}$Z;
bond length $3.013$ \cite{Stwalley1993} & cc-pVTZ; bond length $7.558$, (2,2) 
 \\
Li$_2$ &
cc-pV$\{D,T,Q\}$Z;
bond length $5.051$ \cite{Huber1979} & cc-pVTZ; bond length $10.121$, (4,3) 
 \\
BH &
cc-pV$\{D,T,Q\}$Z;
bond length $2.328$ \cite{Huber1979} & cc-pVTZ; bond length $ 5.327$, (2,4) 
 \\
RbH &
Sapporo DZP 2012 and DKH3-DZP 2012; 
bond length $4.472$ \cite{Geum2001}  
& - \\
BeH$_2$ &
cc-pV$\{D,T,Q\}$Z;  
linear geometry with Be--H bond length $2.506$ \cite{Koput2006} 
& cc-pVTZ; bond length corresponds to point E (see Table \ref{tabbeh2}), (2,2) \\
H$_4$ &
cc-pV$\{D,T,Q\}$Z;
rectangular geometry with  $d=1.800$ in Fig. \ref{fig1}
and the other side length  $2.000$ \cite{sugisaki2022variational} 
& cc-pVTZ; rectangular geometry with $d$= 4.610 in Fig.~1, and the other side length 4.650, (4,4) 
\\
\hline
\end{tabular}
\end{table*}

\section{QLS-SRLCC scaling}\label{sec:sr-scaling} 

\subsection{System size growth}\label{subsec:srsyssize-scaling} 

The system size growth in the SRLCC case is approximately equal to the size of the CI matrix. In the case of CISD, for example, the highest excitation rank being 2, the size of the CISD matrix is approximately the number of possible excited determinants of excitation rank 2, which is estimated as 

\begin{eqnarray}
N \sim \binom{n_{occ}}{2} \binom{n_{vir}}{2} \sim n_{occ}^2 n_{vir}^2,  
\end{eqnarray} 

where $n_{occ}$ and $n_{vir}$ are respectively the number of occupied and virtual spin-orbitals. The additive factor $1+n_{occ}n_{vir}$ is ignored, assuming we work with large spaces, where the doubles sector dominates. Since we fix the number of occupied orbitals and vary the number of virtuals for the purposes of our numerical calculations, we conclude that the system size $N$ in the single-reference case grows as $n_{vir}^2$ in such a situation. 

In general, the system size growth is polynomial in the number of spin-orbitals, and is expected to scale as 

\begin{eqnarray}
N \sim \sum_{m=1}^{\mathfrak{E}}\binom{n_{occ}}{m} \binom{n_{vir}}{m}, 
\end{eqnarray}

where $\mathfrak{E}$ is the excitation rank. When $n_{occ}$ is fixed, $N \sim n_{vir}^\mathfrak{E}$. 

\subsection{Scaling of condition number}\label{subsec:kappa-scaling} 

We find $\kappa$ scaling using three methods (as explained in Sec. \ref{sec:scaling}), for several closed-shell model systems, all chosen in their ground electronic states, with symmetry set to $C_{2v}$ point group (even when the system possesses higher symmetry, for the purposes of simplicity), and in geometries where the HF state is the dominant configuration. The details on the atomic and molecular systems, basis sets, and geometries are provided in Table \ref{tab:systems}, while those for the atomic chains are given separately in Table \ref{tab:systems-chains} of the Appendix. We note that for RbH, we do not use Dunning basis as they are not available for Rb. Furthermore, we only consider DZ quality basis due to computational constraints. In the subsequent paragraphs, we provide only the QZ results for all the systems (except chains where we work with the STO-3G and 6-31G basis sets and RbH where we use DZ quality basis). The rest of the relevant results are provided in Figs. \ref{overall_srlcc_kappa_dmaxbydmn_ccpvdz} and \ref{overall_srlcc_kappa_dmaxbydmn_ccpvtz} of the Appendix. 

\begin{figure*}[t]
   \centering
    \includegraphics[scale=0.55]{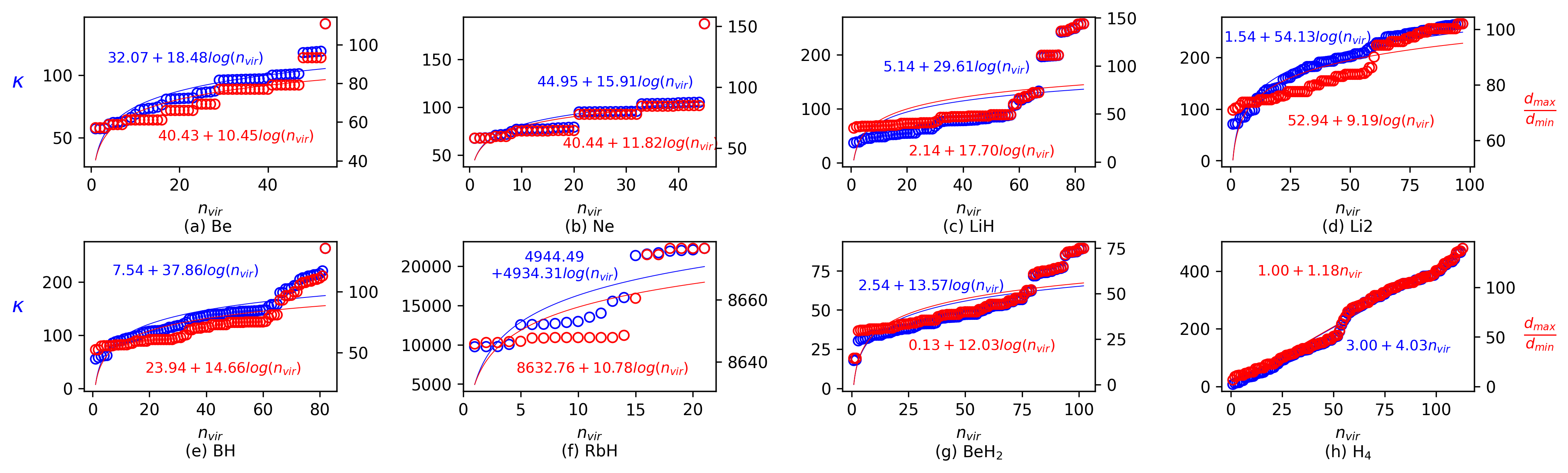}
    \caption{This figure shows the scaling of the condition number, $\kappa$, and the ratio of the maximum and minimum diagonal entries of $A$, $d_{\mathrm{max}}\big/d_{\mathrm{min}}$, with the number of virtuals, $n_{vir}$. $d_{\mathrm{max}}$ and $d_{\mathrm{min}}$ refer to maximum and minimum diagonal entries of $A$ respectively, all for the SRLCCSD case. All the data is presented in the cc-pVQZ basis set, except for the RbH molecule, which is presented in the Sapporo DZP 2012 basis set. }
    \label{fig:sr_overall_allmolecules} 
\end{figure*} 

\begin{center}
    \textit{a. Explicitly computing condition number scaling} 
\end{center} 

Our main result is that $\kappa$ grows slowly relative to system size, $N$, as seen from Fig. \ref{fig:sr_overall_allmolecules} for atoms and molecules in QZ quality basis sets, Figs. \ref{overall_srlcc_kappa_dmaxbydmn_ccpvdz} and \ref{overall_srlcc_kappa_dmaxbydmn_ccpvtz} for DZ and TZ bases respectively, Fig. \ref{overall_excrank} (which compares LCCSD, LCCSDT, and LCCSDTQ cases for LiH in DZ basis), and Figs. \ref{overallchain:h}(a), \ref{overallchain:he}(a), and \ref{overallchain:Li}(a) for chains. For most systems, the observed behaviour is consistent with approximately polylogarithmic growth over the sampled ranges. However, H$_4$ in the cc-pVQZ basis exhibits a different trend, yielding an approximately linear dependence, $\kappa\approx 3.00+4.03n_{vir}$, as shown in Fig. \ref{fig:sr_overall_allmolecules}(h). For other systems, we observe a staircase-like behaviour consistent with an underlying polylogarithmic trend in $\kappa$ over the investigated system sizes. Examining the symmetry labels of the virtuals indicates that transitions between successive `plateaus' coincide with the point where we run out of virtuals in an irrep and begin adding virtuals from the next. The staircase-like behaviour may therefore be understood in terms of symmetry composition in the virtual space. As virtuals belonging to an irreducible representation (irrep) are added, the extremal eigenvalues of $A$ initially `react' to the inclusion of virtuals by adjusting the maximum and minimum eigenvalues. However, beyond a certain point, additional virtuals from the same symmetry do not appear to contribute to the eigenvalues of interest, thus causing $\kappa$ to stagnate. When we run out of virtuals of an irrep and begin adding ones from a new irrep, the spectrum undergoes a rather modest reorganization, eventually leading to the same kind of a plateau behaviour after a certain point, and so on. 

We now report other notable observations. All of the sampled $A$ matrices from the systems considered are Hermitian. Furthermore, all of the systems also satisfy $A > 0$, that is, the $A$ matrices are positive definite. With regard to $\kappa$, we observe that there is no discernible pattern on whether it is governed primarily by $\lambda_{\max}$ or $\lambda_{\min}$, since our data shows that in some cases, one quantity dominates whereas for other cases, both influence the condition number. Furthermore, we also observe that as in RbH, the condition number can go as high as 22000 for the all-orbital case, but the scaling itself is still indicative of polylogarithmic behaviour. Freezing occupied orbitals yields the same polylogarithmic growth but with smaller $\kappa$ values, as also seen from Ref. \cite{Peniel2025}. Although larger systems need not necessarily correlate with larger $\kappa$ (noting that we do expect $\kappa$ to depend on the atomic species with the highest atomic number for a given molecule; this is because the highest eigenstate of the $A$ matrix is generally characterized by two-electron excitations from the $1s$ orbital of the heaviest atom to the highest-energy orbitals. The energy of the $1s$ orbital is roughly proportional to $Z^2$, where $Z$ is the atomic number), it is worth carrying out a numerical survey as a part of a future study on how large the quantity could be with and without freezing some occupied orbitals. The observation that $\kappa$ itself can be large while the scaling is still favourable also highlights an intrinsic limitation of QLSs, which in its best case scenario would scale linear in $\kappa$ (for example, see Ref. \cite{opt1, opt2}). We also study $\kappa$ scaling for the representative case of LiH in cc-pVDZ basis in LCCSDT and LCCSDTQ levels of theory, and find that the trend remains step-wise polylogarithmic in $n_v$ even as the excitation rank increases. 

We also carry out analogous analyses on H, He, and Li chains with SRLCCSD, with $n_v(n_{vir})$ replaced by the number of atoms ($n_{\mathcal{A}}$) (see Figs. \ref{overallchain:h}(a) and \ref{overallchain:he}(a), as well as Fig. \ref{overallchain:Li}(a) of the Appendix). We once again find that $\kappa$ scaling is indicative of a polylogarithmic function in the number of atoms. We tabulate in Table \ref{tab:srlcc-N-kappa-data} of the Appendix the ranges of $N$ and $\kappa$, and also by how much the quantities grow for all of the systems considered, in order to give a feel for not only the system sizes we considered but also by how much $\kappa$ grows in proportion to $N$. 

\begin{center}
    \textit{b. Indirect estimation of condition number via a bound} 
\end{center} 

We now move to the indirect diagnostic of using the ratio $d_{\max}/d_{\min}$ for determining $\kappa$ scaling. The close agreement between $\kappa$ and $d_{\max}\big/d_{\min}$ for a majority of the considered molecular systems as seen in Figs. \ref{fig:sr_overall_allmolecules}, \ref{overall_srlcc_kappa_dmaxbydmn_ccpvdz} and \ref{overall_srlcc_kappa_dmaxbydmn_ccpvtz}, and for chains as shown in Figs. \ref{overallchain:h}(b), \ref{overallchain:he}(b), and \ref{overallchain:Li}(b), indicates that this diagnostic is reasonable at providing information on $\kappa$ scaling. Based on our data, we expect that the effect of $V$ on the extremal eigenvalues is relatively small for a majority of the systems that we consider (although we do not numerically check if the norm of $V$ is much less than $d_{\min}$ (see Theorem \ref{theorem:2})). 

A useful way to understand the observed agreement is to utilize the benefit of hindsight and factorize $\kappa = \lambda_{\max}/\lambda_{\min} = (d_{\min}/\lambda_{\min})(\lambda_{\max}/d_{\max})(d_{\max}/d_{\min})= \varpi (d_{\max}/d_{\min})$. Figs. \ref{lambdamaxbydmaxdminbylambdamin_ccpvdz}, \ref{lambdamaxbydmaxdminbylambdamin_ccpvtz}, and \ref{lambdamaxbydmaxdminbylambdamin_ccpvqz} show that $\lambda_{\max}/d_{\max}$ remains very close to 1, whereas $d_{\min}/\lambda_{\min}$ exhibits a (system-dependent) departure from 1. Therefore, the distinction between $\kappa$ and $d_{\max}/d_{\min}$ scaling curves arises predominantly from the minimum eigenvalue factor. Having said that though, our data still indicate that the multiplicative factor $\varpi$ remains of order unity over our sampled ranges, and thus indicates that one may use $d_{\max}\big/d_{\min}$ as a proxy to $\kappa$ but which is cheaper to compute ($ \mathcal{O}(N)$). 

\begin{center}
    \textit{c. Invoking the edge spawning conjecture} 
\end{center} 

Figs. \ref{heatmapsr:be} through \ref{heatmapsr:h4} present the $A$ matrix structure for 3 randomly sampled system sizes for SRLCCSD, SRLCCSDT, and SRLCCSDTQ cases per system considered. All of them display a diffuse pattern, which, according to our extended edge spawning conjecture, indicates a slow $\kappa$ growth. We report similar findings in chains, as shown in Figs. \ref{overallchain:h}(c), \ref{overallchain:he}(c), and \ref{overallchain:Li}(c). 

The observed diffuse patterns are broadly consistent with the relatively slow $\kappa$ growth found through our direct calculations and are qualitatively in agreement with the predictions of the edge-spawning conjecture. Together with the $d_{\max}/d_{\min}$ proxy, these observations provide complementary indications of favourable $\kappa$ scaling over the sampled system sizes. 

\subsection{Scaling of sparsity}\label{subsec:srs-scaling} 

For the SRLCC formulation in the Slater determinant basis, the sparsity grows at most as $\sim n_{occ}^2 n_{vir}^2$, because for any particular row or column of the CI matrix, an element is zero if the determinants on each side of the Hamiltonian differ by more than two orbital substitutions. This scaling is a consequence of the two-body nature of molecular Hamiltonians, and thus its functional form dependence on the number of occupied and virtual orbitals remains unchanged with excitation rank, $\mathfrak{E}$. One can express sparsity in terms of system size and excitation rank as $s = \mathcal{O}( N^{2/\mathfrak{E}})$ for $\mathfrak{E} \geq 2$ ($s=\mathcal{O}(N)$ for $\mathfrak{E}=1$). We add that $s$ from determinant-based CI upper-bounds that from GUGA based CI \cite{Shavitt-guga}. Later on, when we put things together in evaluating $R(N)$, we shall see that this observation of sparsity `stagnating' has important implications. 

\section{QLS-icMRLCC scaling}\label{sec:icmr-scaling} 

\subsection{System size growth} \label{subsec:mrsyssize-scaling} 

The system size before redundancy removal in the case of icMRCISD is (assuming for simplicity $N_a \approx n_a$) $N \sim n_{c}n_a + n_an_v +n_{c} n_v + n^2_{c}n^2_a + n^2_an^2_v +n^2_{c} n^2_v  $ where the first three terms account for the single excitations and the rest for double excitations. The redundant excitations incurred are of two main types: 1. those due to the repetition of pure single excitations as double excitations involving one active to active excitation, and 2. the ones due to the contracted nature of the reference function $\ket{\psi_0}$. To account for the first redundancy type, we simply remove the term $n_{c}n_a + n_an_v +n_{c} n_v$ in $N$. This leads us to the new expression of $N \sim n^2_{c}n^2_a + n^2_an^2_v +n^2_{c} n^2_v$. For the second type, we diagonalize the overlap matrix, $\mathcal{S} \in \mathbb{R}^{n\times n}$ built out of excitation functions, such that its matrix elements are given by the expression $\mathcal{S}_{pq}= \bra{\psi_{0}} \hat{\tau}_p^\dagger \hat{\tau}_q \ket{\psi_{0}}$, where $\hat{\tau}_p$ and $\hat{\tau}_q$ are excitation operators that give rise to various excited functions upon acting on $\ket{\psi_{0}}$.  Following this, we get rid of the eigenvalues below $10^{-6}$. Excitation operators that share the same number of created holes  ($h$) and particles ($p$) form blocks of the overlap matrix, and a pair `($h,p$)' is called an excitation class. It can be shown that $\mathcal{S}$ is decomposed as $\mathcal{S} = \oplus_{(h,p)} \mathcal{S}^{(h,p)}_{\text{a}} \otimes \mathbb{1}^{(h,p)}_{\text{na}}$ where   $\mathcal{S}^{(h,p)}_{\text{a}}$ is the active part of a given excitation class ($h,p$) involving only active orbital indices, whereas  $\mathbb{1}^{(h,p)}_{\text{na}}$ (`na' denotes that only non-active orbital indices are involved) is the corresponding antisymmetric unit tensor in that excitation class (for details, see Section II of Ref. \cite{Andreas2011}). In practice, redundancies are removed by diagonalizing the active part of $\mathcal{S}$, after which repeated excitations operators are discarded, therefore the number of redundancies removed can be estimated from the size of the largest active blocks of $\mathcal{S}$, namely $\mathcal{S}^{(0,1)}_{\text{a}}$ and $\mathcal{S}^{(1,0)}_{\text{a}}$, which is $\binom{n_a}{2} \binom{n_a}{1} \sim n_a^3$. The expression of $N$ in icMRCISD becomes 
$N \sim  n^2_{c}n^2_a + n^2_an^2_v +n^2_{c} n^2_v - n_a^3 
\sim  (n^2_a +n^2_{c} ) n^2_v + n^2_{c}n^2_a - n_a^3$. Since $n_v \gg n_{c}, n_a$, the final expression for system size is, 

\begin{eqnarray}
N \sim (n^2_a +n^2_{c}) n^2_v. 
\end{eqnarray} 
When we fix the number of occupied orbitals, $n_c$ is fixed. It is also reasonable in practice to assume that $n_a$ be fixed for a particular molecule, and thus $N \sim n_v^2$. That is, asymptotically, both SRLCCSD and icMRLCCSD exhibit quadratic growth of $N$ with the number of virtuals. 

In general, for an excitation rank $\mathfrak{E}$, one can show using similar arguments that one needs to remove $\binom{n_a}{\mathfrak{E}} \binom{n_a}{\mathfrak{E}-1}$ to account for redundancies. Since $n_a$ is assumed fixed for a molecule, one expects $N \sim n_v^\mathfrak{E}$ when the number of occupied spin-orbitals are fixed, consistent with the SRLCC scaling estimate. 
\begin{figure*}[t]
   \centering
    \includegraphics[scale=0.55]{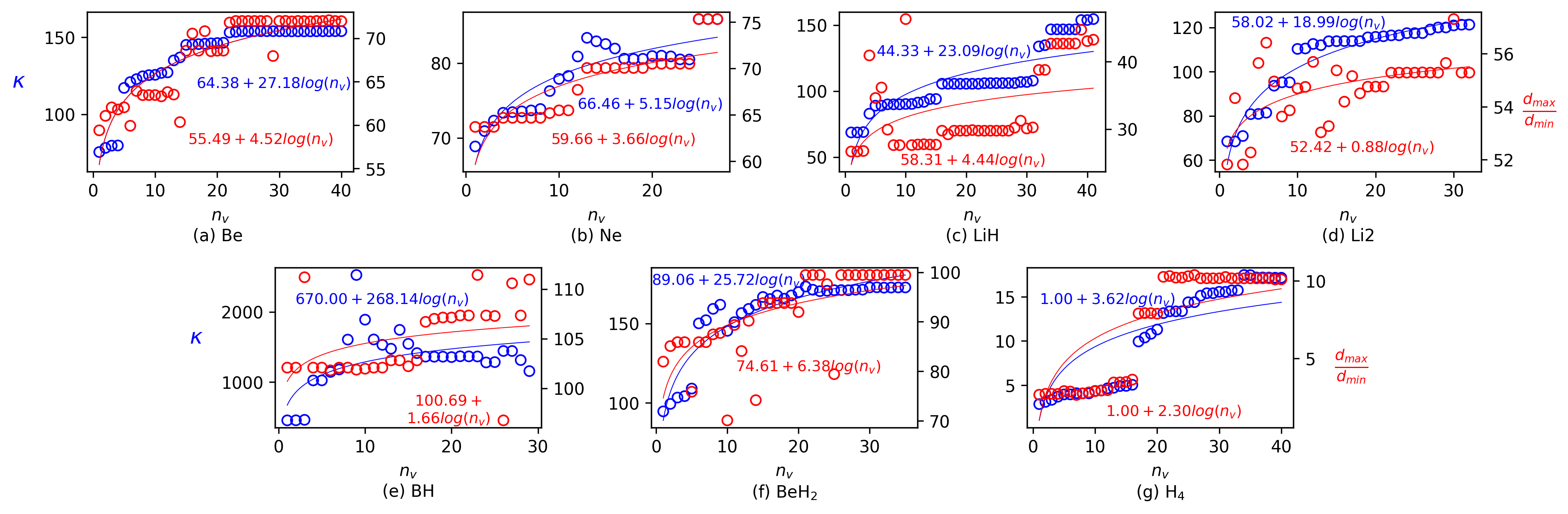}
    \caption{This figure shows the scaling of the condition number, $\kappa$, and the ratio of the maximum and minimum diagonal entries of $A$, $d_{\mathrm{max}}\big/d_{\mathrm{min}}$, with the number of virtuals, $n_{v}$. $d_{\mathrm{max}}$ and $d_{\mathrm{min}}$ refer to maximum and minimum diagonal entries of $A$ respectively, all for the icMRLCCSD case. Be and Ne are in the cc-pVQZ basis, whereas the rest are in the cc-pVTZ basis. }
    \label{fig:mr_overall_allmolecules} 
\end{figure*} 

\subsection{Scaling of condition number} \label{subsec:mrkappa-scaling} 

\begin{center}
    \textit{a. Explicitly computing condition number scaling} 
\end{center} 

Fig. \ref{fig:mr_overall_allmolecules} presents the scaling of $\kappa$ with $n_v$ for the icMRLCCSD matrices considered for our analysis. We begin by focusing on Be, Ne, LiH, Li$_2$, BeH$_2$, BH, and H$_4$, for which all sampled $A$ matrices are positive definite for our active space choice (see Table \ref{tab:systems}). Unlike the SRLCCSD case, we do not consider these systems in DZ through QZ quality bases due to computational constraints. RbH molecule itself is not considered for the same reasons. Be and Ne are considered in the cc-pVQZ basis, while the other five systems are considered in the cc-pVTZ basis, as Table \ref{tab:systems} shows. As summarized in Table \ref{tab:icmr_scaling}, while $N$ goes up by 21 to 129, $\kappa$ only increases between 1.25-6.04 between the first and final sampled points, showing a slow growth of $\kappa$ relative to $N$. In particular, although H$_4$ exhibits the largest relative increase in $\kappa$, its absolute value remains below 18 throughout the sampled range. While polylogarithmic fits are a reasonable qualitative description to our data sets, they do not capture all the observed jumps equally well. Thus, while our results suggest favourable $\kappa$ scaling over the investigated ranges and are consistent with approximately polylogarithmic growth, they do not establish a unique asymptotic scaling pattern. We defer a detailed investigation of the staircase-like structure to future studies. 

For BH, our calculations employing a (2,4) reference active space, for which the $A$ matrices remain positive definite over the sampled range, $\kappa$ exhibits pronounced non-monotonic behaviour, ranging from approximately 453 to 2528. These variations arise predominantly from changes in the minimum eigenvalue, which ranges between approximately 0.0083 and 0.0460, whereas the maximum eigenvalue remains nearly constant at approximately 20.87 to 20.89. Thus, $A>0$ alone does not guarantee favourable conditioning, although for an indefinite case, $\kappa$ at each point would depend significantly on the eigenvalue closest to 0. On this note, we emphasize that the choice of reference active space plays an important role in determining the spectral properties of $A$. In preliminary calculations employing alternative active spaces for BH and BeH$_2$, we encountered cases where $A$ became indefinite beyond certain values of $n_v$. These cases also exhibited isolated but pronounced increases in $\kappa$, despite otherwise relatively slow variation across the sampled range. We did not identify a systematic dependence of the onset of indefiniteness on the active space from these preliminary calculations. A comprehensive investigation of this observation is beyond the scope of the present work and is deferred to a future study. Consequently, the favourable $\kappa$ behaviour reported here should not be interpreted as a property of arbitrary choices of reference active space. 

We also investigate the behaviour of $\kappa$ for chains, just as in the SRLCC case. We consider only the linear hydrogen chains containing 2 to 8 atoms as a representative example, in the 6-31G basis set, as shown in Fig. \ref{overall_plot:hchain-icmr}. The reference active spaces are chosen separately for each chain length, as specified in Table \ref{tab:systems-chains}. We observe that $\kappa$ initially increases with the number of atoms before exhibiting a mild decrease and remaining nearly constant for the larger chains considered. The diagonal ratio, $d_{\max}/d_{\min}$, exhibits a qualitatively similar behaviour. Furthermore, the $A$ matrix for the representative H$_6$ chain displays a diffuse pattern of non-zero elements. While the modest variation of $\kappa$ is encouraging, the limited number of data points does not permit us to reliably establish its functional scaling with system size. 

We note that our icMRLCC calculations (and not necessarily a conventional icMRLCC calculation) are significantly more computationally demanding than their SRLCC counterparts. In particular, determining the spectral properties of the icMRLCC $A$ matrices requires explicit construction of the internally contracted Hamiltonian and overlap matrices, removal of redundancies, transformation to the full rank representation, and subsequent diagonalization of the resulting matrices. These steps become increasingly demanding in computational time and memory as the orbital space grows (see Sec. \ref{subsec:preproc}). Hence, our numerical survey is restricted to comparatively smaller system sizes, and even among the systems chosen, we do not consider the heavier RbH for the icMRLCC analyses. 

Finally, we emphasize that throughout our numerical scaling analysis, for a given system, the reference active space is kept fixed while the virtual orbital space is made larger systematically. Although this provides a consistent framework for assessing $\kappa$ scaling, it need not correspond to the most suitable reference active space at every system size. Allowing the active space to vary with $n_v$, subject to chemically motivated accuracy requirements, could lead to different spectral properties of $A$, and would constitute worthwhile future directions worth exploring. 

\begin{center}
    \textit{b. Indirect estimation of condition number via a bound} 
\end{center} 

We now move to the usefulness of the diagonal ratio, $d_{\max}/d_{\min}$, as an indirect diagnostic of $\kappa$ scaling. Fig. \ref{fig:mr_overall_allmolecules} shows that the diagnostic exhibits similar changes as $\kappa$ for Be and Ne, whereas for LiH and H$_4$, it captures the broad behaviour of $\kappa$. There is no clear pattern for Li$_2$. Just as in the SRLCC case, we can understand the quantity by factorizing $\kappa=(\lambda_{\max}/d_{\max})(d_{\min}/\lambda_{\min})(d_{\max}/d_{\min})$. Fig. \ref{lambdamaxbydmaxdminbylambdamin_mr_ccpvtz} shows that $\lambda_{\max}/d_{\max}$ generally remains relatively close to 1, whereas $d_{\min}/\lambda_{\min}$ exhibits larger, system dependent variations. These observations explain the diagonal ratio not always tracking $\kappa$, particularly in the case of Li$_2$. Furthermore, Fig. \ref{kappaA_kappaHmr:allmolecules} shows that $\kappa$ of the full rank Hamiltonian, $\kappa(H^{\mathrm{fr}})$, need not follow the same trend as $\kappa(A)$, especially for H$4$. Not surprisingly, while $d{\max}/d{\min}$ provides a computationally inexpensive empirical diagnostic for selected systems, its reliability is not universal. We also note that the spectral-norm condition of Theorem \ref{theorem:2} has not been verified for these matrices in view of limitations in computational resources, except for relatively small system sizes for the case of Li$_2$. In particular, we explicitly compute the spectral norm of $V$ for the first 32 sampled system sizes of Li$_2$. We find that $\|V\|$ ranges from 1.164 to 3.032, whereas $d_{\min}$ ranges between 0.102 to 0.106. As a consequence, the ratio $\|V\|/d_{\min}$ ranges from 10.98 to 28.62, violating the sufficient condition $\|V\|\ll d_{\min}$ in Theorem \ref{theorem:2}. 

To circumvent further computations, we design a preliminary test of the spectral norm condition in Theorem \ref{theorem:2}, by employing Weyl's inequality to obtain a lower bound on $\|V\|$ using the differences between the extremal eigenvalues of $A$ and its diagonal entries. That is, $\|V\| \geq \max\{|\lambda_{\min}(A)-d_{\min}|,|\lambda_{\max}(A)-d_{\max}|\}$. Interestingly, this lower bound already exceeds $d_{\min}$ for all sampled Li$_2$ matrices and for a majority of the LiH matrices (30 out of 41), ruling out the sufficient small-norm condition in these cases. We note that the other entries do not establish that it holds; they only do not rule it out. 

\begin{center}
    \textit{c. Invoking the edge spawning conjecture} 
\end{center} 

We report icMRLCCSD heatmaps for the 7 molecules in Figs. \ref{heatmapmr:be_ne_lih} and \ref{heatmapmr:li2_bh_h4_beh2}. All of them show a diffuse pattern. While we note that the SR and icMR $A$ matrices have different diffuse patterns, we only seek to invoke and adapt the conjecture for our purposes, and we do not delve into the reasons behind the specifics of the patterns. 

\subsection{Scaling of sparsity} \label{subsec:mrs-scaling}

We argue that $s$ scaling for the icMRLCC case is determined in the same way as it is for the SRLCC scenario, that is, it is upper bounded by a quadratic scaling in the number of virtuals when we fix the number of occupied spin-orbitals. This is due to the fact that active orbitals can participate in an excitation, but excited functions are determined by virtual orbitals or created particles. We consider an example, and for that purpose, take recourse to a diagrammatic representation of Hamiltonian matrix elements used in Ref. \cite{Andreas2011} for ease of presenting our reasoning. Fig. \ref{matrixelet} shows the diagram of a matrix element of the two-body Hamiltonian (dashed horizontal line) between two doubly excited functions $\ket{\varphi_{ab}}$ and $\bra{\varphi_{cd}}$. The active spin-orbitals, denoted with a solid arrowhead, do not play any role in generating the excited function $\bra{\varphi_{cd}}$, and thus, for our purposes of finding if the matrix element is zero using Slater-Condon rules, we only need to look at the number of spin-orbitals by which the bra and ket differ (see Fig. 3 on Page 7 of \cite{Andreas2011} for more details). As they differ by two virtual spin-orbitals, the matrix element is not zero via Slater-Condon rules. In general, for icMRCI, $s$ is independent of the excitation rank and grows as $\sim (n^2_a + n^2_{c}) n^2_v$. In summary, $s$ scales quadratically in the number of virtuals for both the SRLCC and the icMRLCC cases when we keep the number of core and active spin-orbitals fixed. 

\begin{figure}[!h]
    \centering
  \includegraphics[scale=0.35]{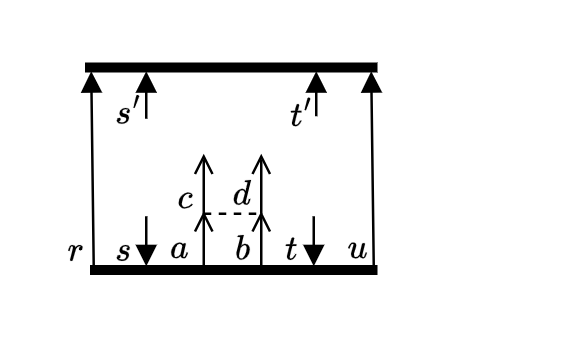}
\caption{Diagram of $\bra{\varphi_{cd}} H^{\text{fr}} \ket{\varphi_{ab}}$. The two solid black horizontal lines at the bottom and the top represent $\ket{\psi_0}$ and $\bra{\psi_0}$ respectively. Filled triangular arrow heads stand for active lines, whereas normal arrow heads stand for particle lines. The dashed horizontal line represents the two-body Hamiltonian.} \label{matrixelet}
\end{figure} 

\section{Quantum advantage with QLS-LCC} \label{sec:advantage} 

\subsection{Runtime ratios and quantum advantage}\label{subsec:R} 

We finally arrive at the runtime ratio estimate to determine if QLS-LCC offers an advantage, both for SRLCC and icMRLCC cases. We emphasize that our runtime ratios are based on asymptotic cost estimates for representative classical and quantum algorithms, while the functional form for $\kappa$ is determined using numerical data from a finite number of systems and relatively small system sizes (we recall that $s$ scaling is analytically obtained). Furthermore, we recall that whereas our SRLCC results provide broader numerical support for favourable $\kappa$ scaling, the icMRLCC results are sensitive to the choice of reference active space, with alternative choices sometimes yielding indefinite $A$ matrices and sharp but isolated increases in $\kappa$. Thus, the resulting runtime ratios are to be interpreted as estimates of potential quantum advantage, rather than establishing such an advantage generally for QLS-icMRLCC. For the rest of the section, we will denote the number of virtual orbitals by $n_v$ in both SRLCC and icMRLCC cases for simplicity. 

Before determining the runtime ratios, it is instructive to obtain the individual runtime scalings under our assumptions of $s \sim N^{2/\mathfrak E}$, $\kappa\sim\log (N)$, and $1/\epsilon\sim\log (N)$, recalling that $\mathfrak E$ is the excitation rank. The assumed sparsity scaling corresponds to saturation of the upper bound, $s=\mathcal{O}(N^{2/\mathfrak E})$, discussed in Sec. \ref{subsec:srs-scaling} and Sec. \ref{subsec:mrs-scaling}, and thus represents a conservative scenario for assessing quantum advantage within our runtime models. A slower growth of $s$ would further improve the runtime ratios. We note that our observations indicating approximately polylogarithmic behaviour in $n_v$ scaling for $\kappa$ is approximated to $\kappa\sim \log(n_v)$ as it does not affect our conclusions regarding exponential separation in $\log(N)$. We further make use of the relation $N \sim n_v^\mathfrak{E}$ when the number of occupied orbitals is fixed ($\binom{n_v}{\mathfrak{E}} = \frac{n_v(n_v-1)\cdots (n_v-\mathfrak{E}+1)}{\mathfrak{E}!}$, which in the limit of large $n_v$ and under the condition of $\mathfrak{E}$ fixed boils down to $\sim n_v^\mathfrak{E}$) to arrive at $\kappa \sim \log(N)$. With regard to $s$, we have $s = \mathcal{O}(n_v^2) = \mathcal{O}(N^{\frac{2}{\mathfrak{E}}})$, where the latter is due to $n_v \sim N^{\frac{1}{\mathfrak{E}}}$. Substitution into Eqs. (\ref{eq:thhl}), (\ref{eq:tcks}), and (\ref{eq:tcg}) yields $t_{\mathrm{HHL}}\sim N^{4/\mathfrak E}(\log (N))^5$, $t_{\mathrm{CKS}}\sim N^{2/\mathfrak E}(\log (N))^2 \log(N^{2/\mathfrak E} (\log(N))^2)$, and $t_{\mathrm{cl}}\sim N^{1+2/\mathfrak E}(\log (N))^{1/2}\log(\log (N))$, respectively, where for CKS we approximate the polylogarithmic factor by a logarithm for simplicity. Thus, the runtime ratio between the reference classical approach and HHL and that between the reference classical approach and CKS is found to be: 

\begin{eqnarray}
R(HHL) &=& \frac{t_{\mathrm{cl}}}{t_{HHL}} \sim \frac{N}{\log(N)}\frac{1}{s \kappa^{5/2}}\frac{\log{(\log{(N)})}}{\log{(N)}} \nonumber \\ &=& \frac{N^{\frac{\mathfrak{E}-2}{\mathfrak{E}}}}{(\log{(N)})^{9/2}}\log(\log(N)),
\end{eqnarray} 
and 
\begin{eqnarray}
R(CKS) &=& \frac{t_{\mathrm{cl}}}{t_{CKS}} \sim \frac{N}{\log(N)}\frac{1}{ \kappa^{1/2}}\frac{\log{(\log{(N)})}}{\log{(s \kappa \log{(N))}}} \nonumber \\ &=& \frac{N}{(\log(N))^{3/2}} \frac{\log(\log(N))}{\log(N^{\frac{2}{\mathfrak{E}}}(\log(N))^2)} 
\end{eqnarray} 

respectively. With HHL, when $\mathfrak{E}>2$ (LCCSDT or beyond), the runtime ratio grows as a positive power of $N$ up to polylogarithmic factors, indicating an exponential separation in $\log(N)$. For CKS, $\mathfrak{E}=2$ (LCCSD) already yields an exponential separation in $\log(N)$. 

In terms of $n_v$, the expressions become 

\begin{eqnarray}
R(HHL) &\sim& \frac{n_v^{\mathfrak{E}-2}}{(\log{(n_v)})^{9/2}}\log(\log(n_v)), \ \mathrm{and} \\ 
R(CKS) &\sim& \frac{n_v^{\mathfrak{E}}}{(\log(n_v))^{3/2}} \frac{\log(\log(n_v))}{\log(n_v^{2}(\log(n_v))^2)}. 
\end{eqnarray} 

The exponent on $n_v$ is $\mathfrak E-2$ in the case of $R(\mathrm{HHL})$, and thus, under the assumed sparsity saturation scenario, $\mathfrak E=2$ does not yield an advantage. However, a slower growth of $s$ could still permit an advantage even at the LCCSD level. For CKS, the exponent on $n_v$ is $\mathfrak E$, indicating an advantage already at LCCSD level of theory. Our result therefore highlights the usefulness of near-optimal QLSs in providing an advantage even in regimes where HHL may not. 

We now relax the assumption of polylogarithmic $\kappa$ scaling, and consider the less restrictive case of $\kappa \sim n_v^\alpha$, where $\alpha>0$. This consideration is in view of some results, for example, $H_4$ in the QZ case (see Fig. \ref{fig:sr_overall_allmolecules}), but also serves the purpose of assessing how much room one has for $\kappa$ scaling to retain advantage prospects. We retain $s$ and $1/\epsilon$ scaling behaviours throughout. The expressions for runtimes yield: 

\begin{eqnarray}
t_{\mathrm{cl}} &\sim& n_v^{2+\mathfrak E+\alpha/2}\log(\log (n_v)), \\ 
t_{\mathrm{HHL}} &\sim& n_v^{4+3\alpha}(\log (n_v))^2, \ \mathrm{and} \\ 
t_{\mathrm{CKS}} &\sim& n_v^{2+\alpha}(\log (n_v))\log\left[n_v^{2+\alpha}\log (n_v)\right]. 
\end{eqnarray} 

This leads to the runtime ratios given by: 

\begin{eqnarray}
R({\mathrm{HHL}}) &\sim&n_v^{\mathfrak E-2-5\alpha/2} \frac{\log(\log (n_v))}{(\log (n_v))^2}, \mathrm{and} \\ 
R({\mathrm{CKS}}) &\sim&n_v^{\mathfrak E-\alpha/2}\frac{\log(\log (n_v))}{(\log (n_v))\log\left[n_v^{2+\alpha}\log (n_v)\right]} \nonumber \\ &\sim&n_v^{\mathfrak E-\alpha/2}\frac{\log(\log (n_v))}{(\log (n_v))^2}. 
\end{eqnarray} 

An exponential separation in $\log(n_v)$ is still retained when the exponent on $n_v$ is positive, and thus: 

\begin{eqnarray} 
\alpha &<& \frac25\left(\mathfrak E-2\right)\quad (\mathrm{HHL}),\\ \alpha &<& 2\mathfrak E\quad (\mathrm{CKS}). 
\end{eqnarray} 

Therefore, polylogarithmic $\kappa$ scaling is not strictly necessary for an exponential separation between classical and quantum linear solvers. Even a polynomial growth of $\kappa$ can preserve such a separation, provided the above conditions are satisfied. 

In view of the results derived above, we now revisit the case of $H_4$ from our SRLCCSD QZ basis calculations (Fig. \ref{fig:sr_overall_allmolecules}(h)). For LCCSD, $\mathfrak E=2$, and thus $\alpha<0$ for HHL and $\alpha<4$ for CKS. Therefore, while the HHL condition is not satisfied, the CKS condition is satisfied for linear growth of $\kappa$. In particular, the corresponding CKS runtime ratio scales as $R(\mathrm{CKS})\sim n_v^{3/2}
\frac{\log(\log(n_v))}{(\log(n_v))^2}$. This result indicates that an exponential separation in $\log(N)$ can be retained even when $\kappa$ grows linearly with the number of virtual orbitals. 

\subsection{Gate counts}\label{subsec:gate_scaling} 

The gate counts of a quantum algorithm depend on the implementation choices, including the decomposition of its constituent subroutines, the target gate set, and the compilation strategy. We therefore consider two complementary analyses for HHL-LCCSD. First, we derive asymptotic logical CX and $T$ gate counts using the standard Clifford+$T$ gate set, ${H,S,CX,T}$ (and their adjoints, as appropriate). Next, we obtain numerical gate counts using gate sets similar to those employed in trapped-ion quantum computers, such as older generation IonQ devices, and in benchmarks such as MQT \cite{mqt}. We restrict our analyses to HHL due to its relatively simple and well-established circuit structure, and to the LCCSD approximation for simplicity. 

Unlike the sparse-access model underlying Eqs. \ref{eq:thhl} (or for that matter Eq. \ref{eq:tcks}), our asymptotic gate count analysis employs a Pauli decomposition of $A$, followed by Trotterization. For our numerical gate count analysis however, an important practical consideration is to lower circuit depth and gate counts so that a classical simulation of HHL becomes computationally manageable for somewhat large system sizes. To that end, we employ Qiskit's NumPyMatrix class, which constructs the matrix exponential classically and decomposes the resulting unitary into the chosen gate set. Furthermore, for the controlled unitary powers required in QPE, instead of implementing $U^{2^j}$ through $2^j$ repeated applications of controlled-$U$, where $U=e^{iAt}$, we directly construct and decompose the equivalent controlled unitary $e^{i2^jAt}$. This avoids repeated controlled-unitary operations and can reduce circuit depth. The resulting gate counts therefore correspond to this particular implementation of HHL and need not exhibit the same scaling as the sparse-access QLS runtime. We also note that a block encoding of $A$ provides a more modern and important alternative to perform Hamiltonian simulation, but we do not consider that option for our study. 

We note that we exclude input state preparation and matrix loading process' gate counts for the asymptotic estimates. Also, the use of $T$ to denote the logical version of the physical $T=\begin{pmatrix}
1& 0 \\ 
0& e^{i \pi/4}
\end{pmatrix}$ gate, instead of the cluster operator $\hat{T}$ as was the case in other sections, is a notation specific to only this sub-section. 

\begin{center}
    \textit{a. Asymptotic CX and $T$ counts} 
\end{center} 

\textit{QPE--} We first recall that $A = \sum_{i=1}^{l} a_i P_i$, where $P_i$ are Pauli words, and $l \leq 4^{n_b} = 4^{\lceil \log_2(N)\rceil} = \mathcal{O}(N^2)$. We first want to simulate $e^{iAt}$. If we assume first order Trotterization, then $e^{iAt} \approx \bigg[ \prod_{i=1}^{l} e^{i a_i P_i t/n} \bigg]^n$ for $n$ Trotter steps. In such a case, one has $ln$ Pauli exponentials, and thus as many Pauli gadgets. For second order Trotter approximation, $e^{iAt} \approx \bigg[ \prod_{i=1}^{l} e^{i a_i P_i t/2n}\prod_{i=l}^{1} e^{i a_i P_i t/2n} \bigg]^n$, thus naively leading to $2ln$ gadgets. In general, for Trotter order $2k$ and Trotter step $n$, we obtain $2ln 5^{k-1}$ gadgets, by using the recursive expression $U_{2k}(t)=U_{2k-2}(S_kt)^2 U_{2k-2}((1-4S_k)t)U_{2k-2}(S_kt)^2$ where $S_k=(4-4^{1/(2k-1)})^{-1}$ (see Section 3 of Ref. \cite{Suzuki}) and $[U_{2k}(t/n)]^n$ for $n$ steps. Now, each $m$-qubit gadget, with $m \leq n_b$, has $2m-2$ CX gates and 1 $R_z$ gate (later, we will convert arbitrary-angle $R_z$ count to $T$ count). By setting $m=n_b$, we obtain a loose upper bound on the number of CX gates as $(2n_b-2)2ln 5^{k-1}$ and the number of $R_z$ gates as $2ln 5^{k-1}$. However, since QPE contains controlled $e^{iAt}$ gates, we make use of the fact that an $n_b$-qubit gadget for $U=e^{iAt}$ leads to 2 $n_b+1$-qubit gadgets for $CU$ \cite{blunt_controlled_pauligadget}. Therefore, $2^0 + 2^1 + 2^2 + \cdots + 2^{n_r-1}\approx 2^{n_r}$ controlled unitaries of QPE has a total of at most $2^{n_r}(4n_b-2)2ln5^{k-1}$ CX gates and $2^{n_r}4ln5^{k-1}$ $R_z$ gates. The QPE$^\dag$ module has the same gate count as QPE, and thus it contributes to a factor of 2. Throughout, we ignore the gates from Hadamard transformation as it neither contributes to CX nor to $T$ count. Using the fact that each controlled phase gate decomposes as 2 CX gates and 3 $R_z$ gates (for example, see Fig. 2(a) of Ref. \cite{cpdecomp}), the CX and $R_z$ cost from each of the two QFT modules is $n_r(n_r-1)$ and $\frac{3}{2}n_r(n_r-1)$ respectively. 

\textit{Controlled rotation module--} A uniformly controlled rotation given by $\sum_{i=0}^{2^{n_r}-1} |i\rangle \langle i | \otimes R_y(\theta_i)$ has a decomposition that involves $\approx 2^{n_r}$ CX gates \cite{mottonen2005} and $\approx 2^{n_r}$ $R_z$ gates (there are $\approx 2^{n_r}$ $R_y$ gates, which can be expressed in terms of $R_z$ and the Cliffords $H$, $S$, and $S^\dag$). 

\textit{HOM module--} We consider the CX cost from the HOM circuit, which has as many CX gates as the number of qubits in the state register, which is $n_b$. The HOM circuit has no non-Cliffords. 

We add the gate counts from QPE, QPE$^\dag$, and the controlled rotation modules to yield at most $2^{n_r+1}(4n_b-2)2ln5^{k-1} + 2n_r(n_r-1) + 2^{n_r}+n_b$ CX gates and $2^{n_r+1}4ln5^{k-1}+3n_r(n_r-1)+2^{n_r}$ $R_z$ gates. 

We finally make some substitutions, $l=\mathcal{O}(N^2)\approx \mathcal{O}(n_v^4)$, $n_b \approx 2\lceil \mathrm{log}_2(n_v) \rceil$, and $n_r=\mathcal{O}(\log(\log(n_v)))$, to obtain CX gate count to be 

\begin{eqnarray}
\mathrm{CX\ count} = \mathcal{O}(\mathrm{polylog}(n_v)n5^{k-1}n_v^4). 
\end{eqnarray}

The gate counts are dominated by $n_v^4$, besides one's choice of error to admit via Trotter order and step. The argument does not change when we also vary $n_o$, and we still would obtain number of CX gates scaling at most as polynomial in the number of MOs. We also finally add that when we go beyond LCCSD by choosing $\mathfrak{E}>2$, we still obtain a CX count upper bounded by a polynomial function of $n_v$. 

For the $R_z$ count, we get a very similar expression, that is, $\mathcal{O}(\mathrm{polylog}(n_v)n5^{k-1}n_v^4) \equiv \mathfrak{N}_{R_z}$. We now switch to (logical) $T$ count by using the fact that each $Rz$ gate can be decomposed into $\mathcal{O}(\log(1/\zeta))$ \cite{Ross2016} without aid from ancillas, where $\zeta$ is the error that we admit in the angles of arbitrary rotation $Rz$ gates. Since we have a total of $\mathfrak{N}_{R_z}$ $R_z$ gates, the $T$ count is upper bounded by 

\begin{eqnarray}
T \mathrm{count}&=&\mathcal{O}(\mathfrak{N}_{R_z} \log(1/\zeta)) \nonumber \\ &=& \mathcal{O}(\mathrm{polylog}(n_v)n5^{k-1}n_v^4 \log(1/\zeta)). 
\end{eqnarray} 

Thus, for a fixed $k$ and $n$, as well as with a fixed or inverse polynomially bounded $\zeta$, the $T$ count remains polynomial in $n_v$. 

It is worth adding that for our particular problem, since Trotterization is involved, the $T$ count might be lowered by exploiting angle-aware compilation as discussed in a recent work \cite{angleT}, but we defer it for a future study. Furthermore, resource estimation exercises such as determining number of physical qubits and runtime, as well as compilation strategies to lower $T$ counts, are beyond the scope of our current work, and we postpone that too for a future study. 

\begin{center}
\textit{b. CX gate counts from numerical simulations} 
\end{center} 

\begin{figure}[t]
    \centering
    \begin{tabular}{cc} 
    \includegraphics[scale=0.43]{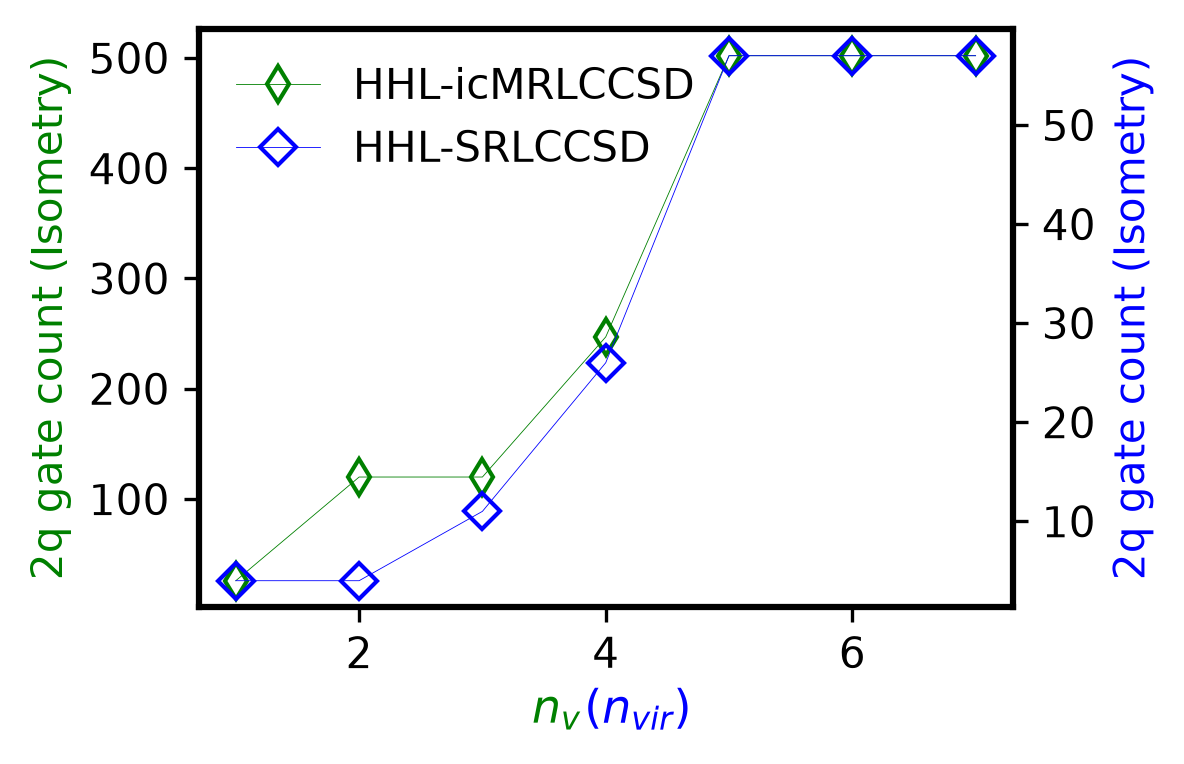} & 
\includegraphics[scale=0.43]{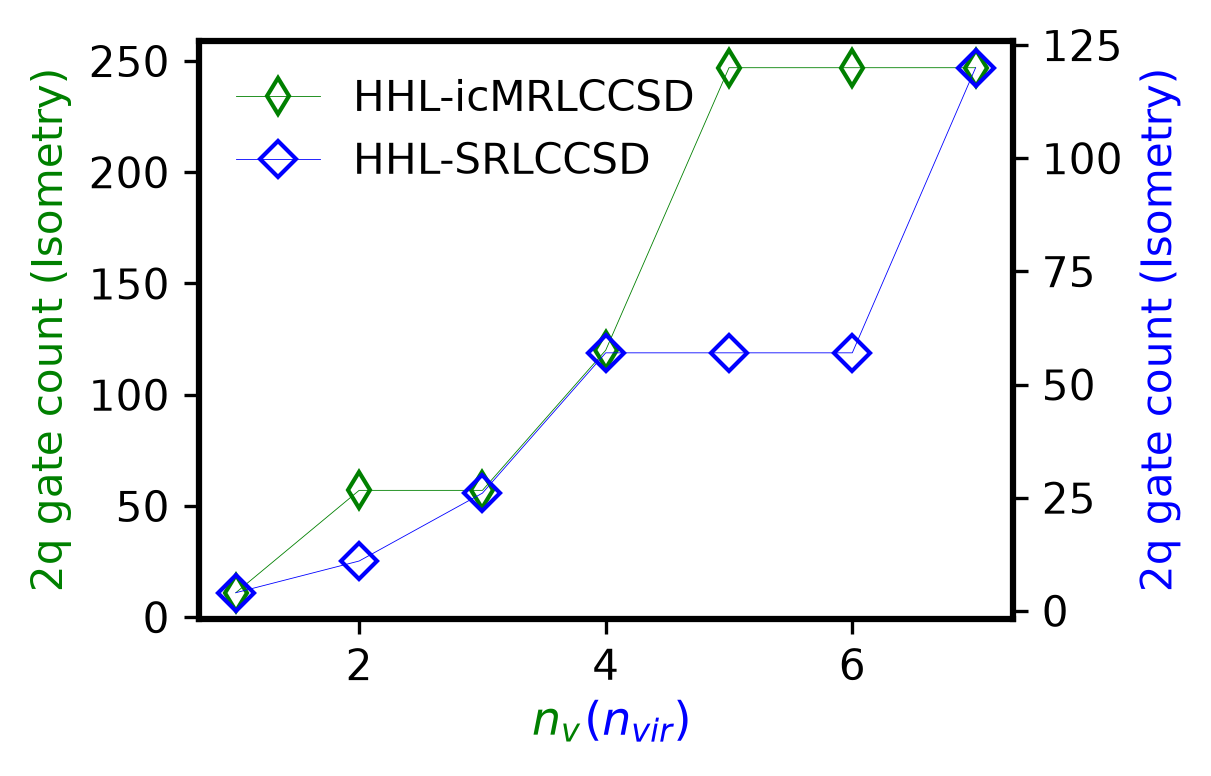} \\
(a) & (b)
    \end{tabular}
\caption{Figure showing the growth of the number of two-qubit gates incurred in the isometry unitary for the case of (a) LiH, (b) H$_4$ molecules versus the number of virtual orbitals, $n_v(n_{vir})$.} \label{isometry_scaling}
\end{figure} 

\begin{figure}[t]
    \centering
    \begin{tabular}{cc}    \includegraphics[scale=0.43]{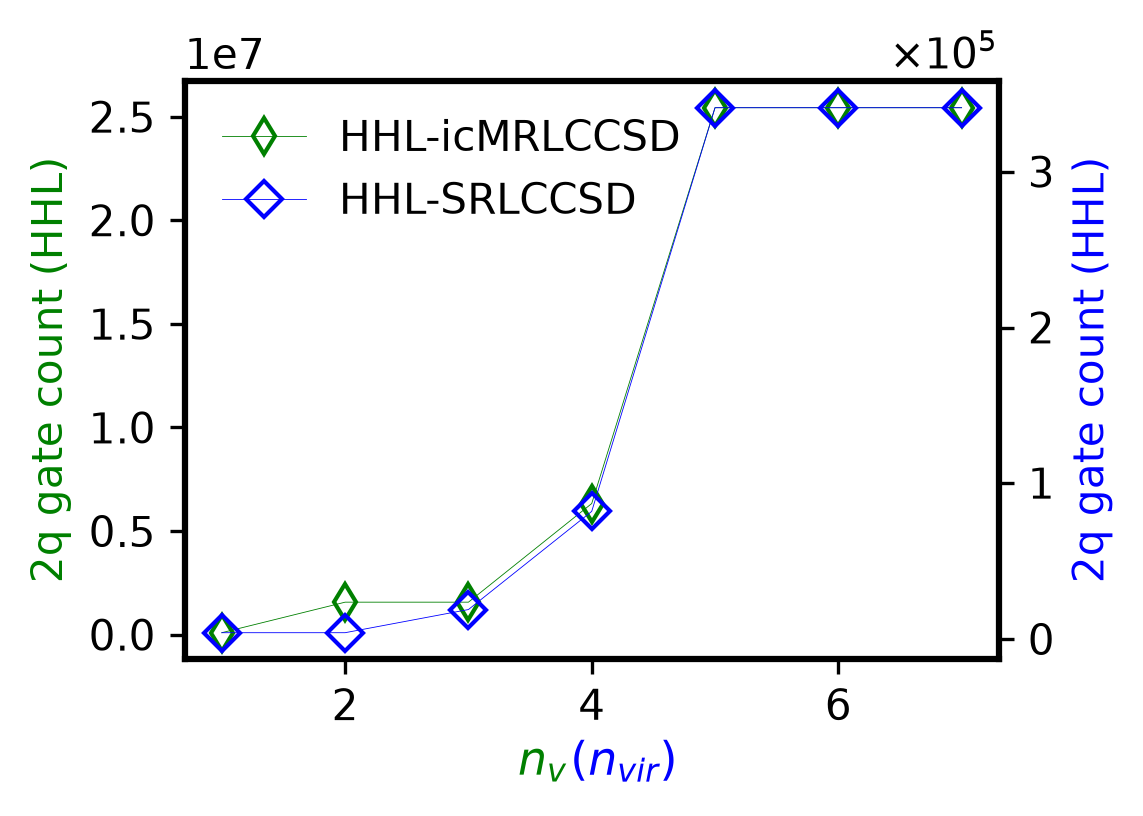} & 
\includegraphics[scale=0.43]{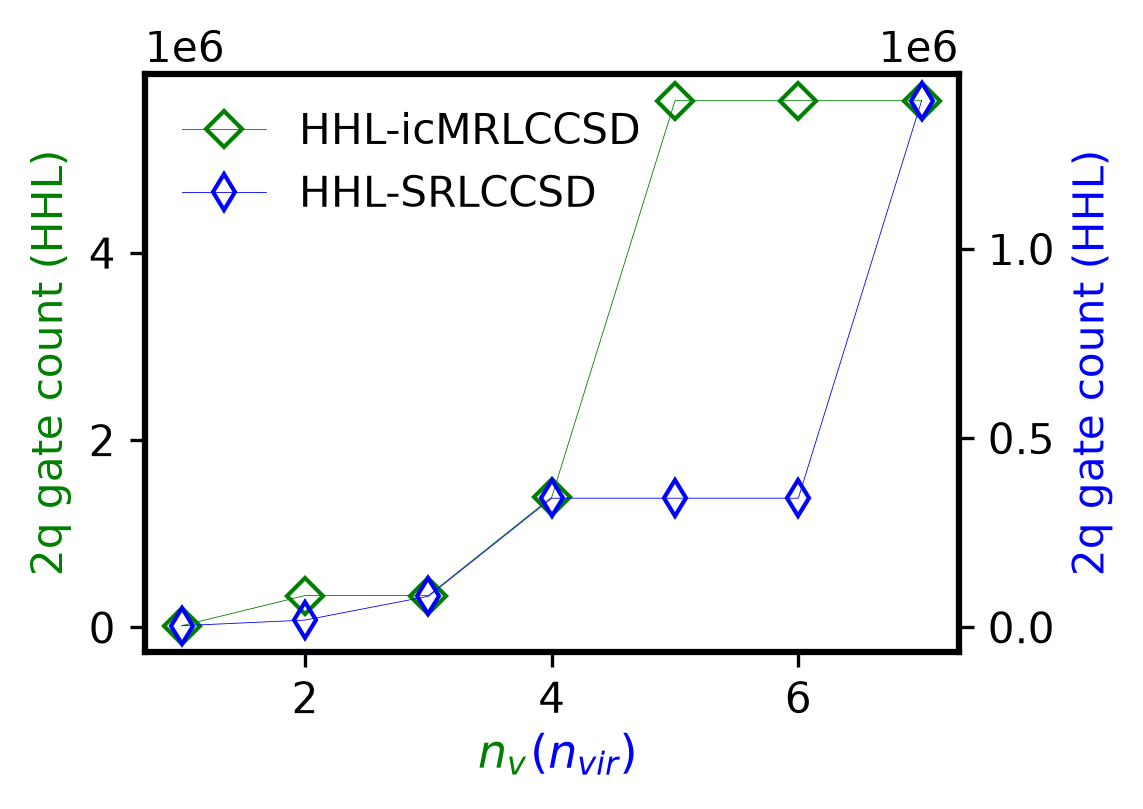} \\
(a) & (b)
    \end{tabular}
\caption{Figure showing the growth of the number of two-qubit gates incurred in the HHL-LCCSD unitary for the case of (a) LiH, and (b) H$_4$ molecules versus the number of virtual orbitals, $n_v(n_{vir})$.} \label{2qhhl_scaling}
\end{figure} 

We estimate the number of two-qubit gates in HHL-LCCSD (both SR and icMR) by considering both the input state preparation and the main circuit (HHL+overlap extraction circuit). For both cases, we simulate and decompose the circuit using the following representative gate set $\{Rx, Ry, Rz, Rxx\}$ where $Rx$ is a shorthand for $Rx(\theta)=e^{-i \theta X/2}$, $Ry(\theta)=e^{-i \theta Y/2}$, and $Rz(\theta)=e^{-i \theta Z/2}$. In the expressions, $X$, $Y$, and $Z$ are the $2 \times 2$ Pauli matrices. Also, $Rxx = e^{-i \theta X \otimes X/2}$. We expect that changing between standard gate sets would not affect gate count estimates significantly. For example, one could go from an $Rxx$ gate to two $CX$ gates by using the fact that $Rxx$ can be decomposed as $Rzz$ and some Hadamard gates, with $Rzz$ being representable by a Pauli gadget with 2 $CX$ gates. Thus, the number of $Rxx$ gates is comparable with the number of $CX$ gates to within a factor of 2 (or more depending on native gates that a target quantum hardware admits). We analyze only the number of two-qubit gates because they are usually the primary source of errors in quantum hardware. 

Figs. \ref{isometry_scaling} and \ref{2qhhl_scaling} show our HHL results for the number of $Rxx$ gates versus the number of virtual orbitals for input state preparation and HHL+overlap extraction circuit respectively, for two representative model systems: LiH and H$_4$, with computational details being discussed in Section \ref{sec:calc}. While the overall trends in gate count are very similar between SRLCCSD and the icMRLCCSD cases, the number of gates grows faster in the latter. We note that the X-axes data points only go up to 7 virtual spin-orbitals due to computational limitations for the transpilation and circuit construction steps for higher number of virtuals. 

\section{Beyond QLS: end-to-end workflow costs}\label{sec:preproc-scaling} 

\subsection{Pre-processing costs} \label{subsec:preproc}

Although a QLS-LCC framework itself may potentially offer an exponential advantage, one needs to carry out all the classical pre-processing steps beginning from HF to atomic orbital (AO) to molecular orbital (MO) integral transformation in constructing $A$ and $\vec{b}$. These steps are not limited to HHL but are necessary for algorithms such as QPE too, except (i) one builds $H$ in place of $A$ in that case, and (ii) constructs an input state to QPE that offers flexibility (build an input state such that the square of the overlap between the input and target state is sufficiently large) unlike in the case of HHL (where $\vec{b}$ and thus its normalized amplitude encoded version are fixed). Naturally, these classical pre-processing steps could be costlier than the quantum algorithmic step in the workflow, and while it is an open problem to find efficient quantizations of these steps, it is worth visiting them briefly to put in perspective the problems to solve before reaping the potential benefits of an exponential advantage that QLS-LCC may offer in an end-to-end workflow. We also add that rather than viewing the cost from pre-processing as a limitation to QLS-LCC, we interpret it as identifying the remaining bottleneck in the end-to-end QLS-LCC workflow. 

It is worth emphasizing that the explicit construction of $A$ is not generally required in conventional electronic structure calculations. In practice, large CI and CC computations often employ iterative algorithms involving Hamiltonian matrix-vector products or evaluations of residuals of relevant amplitudes, thus avoiding the computational expense of constructing and storing the full matrix. Our objective is different; we require explicit access to $A$ to calculate $\kappa$, its diagonal elements, and sparsity structure as the orbital space is enlarged. Consequently, we augment the standard electronic structure computational workflow and modify the necessary routines in the relevant libraries to extract and construct the relevant matrices. The resulting classical pre-processing costs are included in our discussion, although we emphasize that the explicit matrix construction employed for numerical diagnostics is distinct from an efficient quantum-access oracle for $A$. 

\begin{center}
    \textit{a. Obtaining $A$ and $\vec{b}$} 
\end{center} 

We pick HHL-SRLCCSD and HHL-icMRLCCSD for illustration, although the procedure is similar for calculations involving higher excitation ranks. Furthermore, the description of the steps for SRLCC is solely for completeness, as we already discussed this in earlier sections in brief: 

\begin{enumerate}[leftmargin=1.1em] 
\item For HHL-SRLCCSD: 
\begin{itemize}[leftmargin=0.8em]
    \item We carry out an HF calculation to construct Hamiltonian integrals in the MO basis. We use the GAMESS
US program \cite{GAMESS} for this step.
    \item With the information on integrals, we build the CISD Hamiltonian matrix using the Slater-Condon rules. 
    \item We shift the diagonal entries of the CISD Hamiltonian by subtracting the HF energy from each $H_{ii}$. 
    \item We slice away the first row and first column to identify the resulting principal sub-matrix as the $A$ matrix of an SRLCCSD calculation. The sliced column with the first row information removed is the negative of $\vec{b}$. 
\end{itemize}
\item For HHL-icMRLCCSD: 
\begin{itemize}[leftmargin=0.8em]
    \item After choosing an appropriate active space, we carry out a CASSCF calculation, then construct the Hamiltonian integrals out of the optimized molecular orbitals. We employ the Molpro \cite{Molpro} package for this purpose. 
    \item Using the information on integrals, we build the icMRCISD Hamiltonian matrix ($H^{\text{rd}}$) by making suitable modifications to the GeCCo package \cite{gecco2009}. `rd' here denotes `rank-deficient'. As explained before, some excitations are redundant; therefore, the Hamiltonian has some linearly dependent columns and is therefore near-singular. We will now outline the procedure that we use to make $H^{\text{rd}}$ full rank by removing redundancies. 
    \item We remove redundancies in $H^{\text{rd}}$ as follows, with Fig. \ref{A_and_b} outlining the steps: 
    \begin{enumerate}[label=\roman*),leftmargin=0.8em]
        \item We construct $\mathcal{S} \in \mathbb{R}^{n\times n}$ built out of $n$ excitation functions.
        \item We diagonalize $\mathcal{S}$ as follows: 
        \begin{equation}
        U^\dagger \mathcal{S} U = \Omega,
        \end{equation} 
        where $\Omega=\mathrm{diag}{(\omega_1,\omega_2,\cdots,\omega_n)}$ is the diagonal form of $\mathcal{S}$ with all its eigenvalues arranged in the descending order. $U$ is the matrix whose columns are composed of eigenvectors of $\mathcal{S}$ ordered in the same way as in $\Omega$. 
        \item We set a threshold value of $\varepsilon= 10^{-6}$ to slice $\Omega$ such that we keep all $n'$ eigenvalues that are greater than $\varepsilon$ (see Fig. \ref{A_and_b}). 
        \item Let $U_r$ be the $n\times n'$ matrix whose columns are the eigenvectors corresponding to the $n'$ retained eigenvalues of $S$, and let $\Omega_\varepsilon$ be the $n'\times n'$ diagonal matrix containing these eigenvalues. We then construct the transformation matrix $X=U_r\Omega_\varepsilon^{-1/2}$, of dimension $n\times n'$.
        \item Finally, $X^\dagger$ and $X$ are applied on both sides of $H^{\text{rd}}$ the following way to construct the full rank Hamiltonian matrix $H^{\text{fr}}$ (`fr' here denotes `full rank'), 
\begin{equation}
    H^{\text{fr}} = X^\dagger  H^{\text{rd}} X, 
\end{equation}
which has a dimension of $n' \times n'$. 
    \end{enumerate} 
    \item The matrix elements of $A$ in Eq. (\ref{eq16}) can be seen to belong to the principal sub-matrix of $H^{\text{fr}}$, with the excluded row and column that has matrix elements of the type $\bra{\varphi_\beta} H^{\text{fr}} \ket{\psi_{0}}$ and their corresponding conjugates. Because of this, after getting $H^{\text{fr}}$, we identify the row and column where such elements appear and remove them. Lastly, we subtract $E_{0}$ from the diagonal of the resulting matrix to get $A$. Furthermore, Eq. (\ref{eq16}) suggests that the elements of $\vec{b}$ are of the type $-\bra{\varphi_\beta}H^{\text{fr}} \ket{\psi_{0}}$. Therefore, we identify the row/column that has elements of that structure in $A$, remove the entry which has $E_0$ and keep the remaining elements to form $-\vec{b}$. 
\end{itemize}
\end{enumerate} 

We reiterate that unlike conventional electronic structure calculations, our numerical investigations require explicit access to the LCC matrices over systematically enlarged excitation spaces. For the SRLCC case, we construct the relevant CI Hamiltonian matrices using Hamiltonian integrals read from an FCIDUMP file and Slater–Condon rules. For the icMRLCC case, additional implementation work is required to extract the matrix representations and overlap information needed to construct the full rank LCC matrices. These steps require modifications to the conventional electronic structure computational workflow in GeCCo, and are essential for our direct investigations of $\kappa$ and $s$ scaling. 

\begin{figure*}[t] 
\centering
    \hspace{-0.6cm} \includegraphics[scale=0.38]{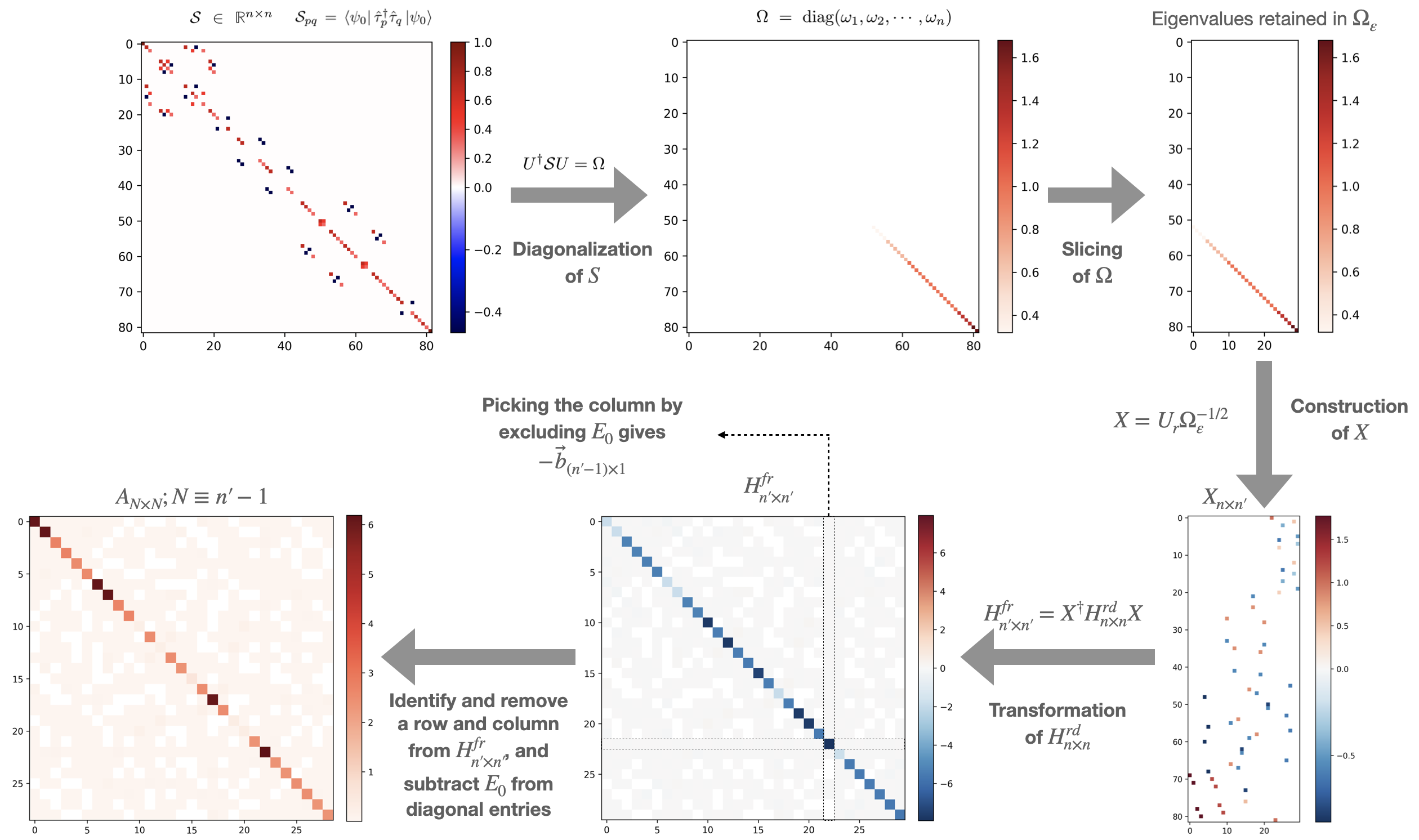}  
\caption{This figure summarizes the procedure used to construct the $A$ matrix and the $\vec{b}$ vector in an icMRLCC calculation, by picking a representative example molecule. The overlap matrix $\mathcal{S}$ is first constructed and diagonalized, which is then sliced such that only the eigenvalues greater than a threshold $\varepsilon$ are retained, leaving $n'$ such eigenvalues. Furthermore, the matrix $U_r$ made of all selected eigenvectors (eigenvectors corresponding to the eigenvalues kept after slicing) of $\mathcal{S}$ is used to construct the transformation matrix $X$, which finally acts on the rank-deficient Hamiltonian, $H^{\text{rd}}$, to form the full rank Hamiltonian, $H^{\text{fr}}$. The $A$ matrix is finally obtained by identifying and removing the row and column of $E_0$, whereas $\vec{b}$ is either that row or column without the entry of  $E_0$. } \label{A_and_b} 
\end{figure*} 

\begin{center}
    \textit{b. Cost of constructing $A$ and $\vec{b}$ classically in SRLCCSD and in icMRLCCSD} 
\end{center} 

\paragraph{SRLCCSD case:} 

The cost of preparing $A$ is dominated by computing the CISD matrix elements, as $A$ is essentially the CISD matrix after slicing the first row and column. We recall that the CISD matrix elements are computed through Slater-Condon rules. The total cost in preparing $A$ is dominated by the number of elements of the form $\bra{\Phi_{ij}^{ab}}H\ket{\Phi_{ij}^{cd}}$, which can be estimated to be $\sim n_{occ}^2 n_{vir}^4$. We assume that evaluating each such matrix element costs $\mathcal{O}(1)$ for simplicity. The net cost of constructing a CI Hamiltonian thus being $=\mathcal{O}(N_s^6)$ for $N_s$ spin-orbitals is also consistent with findings in Ref. \cite{Aashna2025}. Getting $\vec{b}$ costs $\sim n_{occ}^2 n_{vir}^2$ because it amounts to reading a particular row or column of the CISD matrix. 

\paragraph{icMRLCCSD case:} 

In the multi-reference case, the cost of preparing $A$ stems from the following parts: 

\begin{itemize} 
\item $H^{\text{rd}}$ construction: The most expensive terms in getting  $H^{\text{rd}}$ are the ones where the Hamiltonian possesses four particle lines, where two of those lines are involved in a contraction with an excitation operator and two others are open lines, leading to doubly excited functions. In such a situation, the excitation operator that contracts with $H^{\text{rd}}$ has as open lines either two-hole lines, or two active lines, or one active and one hole line, with respective computational costs of $n_c^2n_v^4$, $n_a^2n_v^4$, and $n_cn_an_v^4$, which gives an overall computational cost of $\sim (n_a^2 + n_c^2) n_v^4$ (See Figs. 3 (a), (b) and (c) of \cite{Andreas2011} for more details). The cost of construction $H^{\text{rd}}$ stems from the Hamiltonian contribution in each of these scenarios, which is $\sim n_v^4$.  
\item $\mathcal{S}$ construction: Constructing the overlap matrix depends solely on the cost of evaluating the reduced density matrices that appear in the expressions of each of its blocks,  and contains only active orbital indices. Reduced density matrices are defined by \begin{equation}
    \rho_{rs\cdots}^{tu\cdots} = \bra{\psi_0} a^\dagger_t a^\dagger_u \cdots a_s a_r \cdots \ket{\psi_0}. 
\end{equation}
Consider the example of the matrix element belonging to the excitation class (0,0):  
\begin{equation}
\begin{aligned}
\mathcal{S}_{s, t_1 t_2}^{r, u_1 u_2} &= \bra{\psi_0} (a_s^\dagger a_r)^\dagger (a^\dagger_{u_1} a^\dagger_{u_2} a_{t_1}a_{t_2}) \ket{\psi_0} \\
&= \rho_{t_1 t_2 s}^{r u_1 u_2} + \delta_{su_2} \rho_{t_1 t_2}^{r u_1} + \delta_{s u_1} \rho_{t_1 t_2}^{r u_2},
\end{aligned}
\end{equation}
where only active orbital indices are involved. 
The cost of this matrix element reduces the cost of its most expensive term, which is the three-body reduced density matrix $\rho_{t_1 t_2 s}^{r u_1 u_2}$. Similarly, every matrix element of $\mathcal{S}$ can be expressed in terms of these reduced density matrices involving only active orbitals. Therefore, the rate-determining terms in $\mathcal{S}$ are the $m$-body reduced density matrices $\rho_{rs\cdots}^{tu\cdots} = \frac{1}{(N_a -m)!} \sum_{vw \cdots} C^*_{tu \cdots vw\cdots}  C_{rs \cdots vw \cdots}$, where $C_{\cdots}$ are coefficients that occur in the expansion of $\ket{\psi_0}$. In the expression of $\rho_{rs\cdots}^{tu\cdots}$, $m$ electrons are moved from $r,s \cdots$ to $t,u \cdots$ after excitations, while electrons in $v,w \cdots$ are kept fixed. Thus, the upper bound on the cost of forming such density matrices is $\sim \binom{n_a}{N_a -m} \binom{n_a -(N_a - m)}{m}^2$. In the case of icMRLCCSD, $m$ can go up to 4, thus, the cost incurred is $\sim \binom{n_a}{N_a - 4} \binom{n_a -(N_a - 4)}{4}^2$. 
\item Diagonalization of $\mathcal{S}$: As mentioned earlier, the diagonalization of only active blocks of  $\mathcal{S}$ is needed. The cost of doing so is upper bounded by the cost of diagonalizing the biggest active blocks of $\mathcal{S}$, which is $\left( \binom{n_a}{2} \binom{n_a}{1}\right)^3 \sim n_a^9$. 
\item Construction of $X$: It can be shown that $X$ has the same block structure as $\mathcal{S}$ \cite{Andreas2011}, and therefore, the cost of building the matrix is determined by the cost of building the biggest active block of $X$. The biggest $U$ matrix block has size $n_a^3$. In the worst case where $\Omega_{\varepsilon}$ has the same size as $U$ (that is, no eigenvalue of $\mathcal{S}$ is negligibly small), the cost of building the biggest block of $X$ would be equal to the cost of multiplying $U$ with $\Omega_{\varepsilon}^{-1/2}$ (inverting $\Omega_\varepsilon$ costs $n_a^3$ since the matrix is diagonal) in that block, which is $\sim n_a^6$. 
\item Getting the set of new excitation operators: Let $\hat{T'} = \boldsymbol{\hat{\tau}} \cdot \mathbf{t} = \hat{\tau_1} t_1 + \hat{\tau_2} t_2+ \cdots $, where $\boldsymbol{\hat{\tau}}$ is a row vector containing all excitation operators, and $\mathbf{t}$, a column vector containing all excitation amplitudes. Once $X$ is obtained, the new set of excitation operators is obtained via $\boldsymbol{\hat{\tau'}} = \boldsymbol{\hat{\tau}} X$. Since $X$ has a block structure, all operators in $\boldsymbol{\hat{\tau'}}$ can be obtained by performing $\boldsymbol{\hat{\tau}} X$ in each block of $X$. Therefore, the cost of obtaining $\boldsymbol{\hat{\tau'}}$ is determined by the cost of performing $\boldsymbol{\hat{\tau}} X$ in the biggest block of $X$. Knowing that we have $\sim n_a^3$ elements and $\sim n_a^3$ operators in that block, we expend $\sim n_a^3 n_a^3 = n_a^6$ operations for this task. 
\item Getting the full rank Hamiltonian $H^{\text{fr}}$: Each matrix element of $H^{\text{fr}}$ can be written as $H^{\text{fr}}_{pq} = \bra{\psi_0} \hat{\tau}'^\dagger_p H^{\text{rd}} \hat{\tau}'_q \ket{\psi_0} = \sum_{q',p'} X_{q'p}^* X_{p'q} \bra{\psi_0} \hat{\tau}^\dagger_{q'p} H^{\text{rd}} \hat{\tau}_{p'q} \ket{\psi_0}$. Since the matrix elements of $X$ and $H^{\text{rd}}$ have already been computed in the earlier steps and are stored in memory, and $q'$ and $p'$ run from 1 to $n_a^3$ in the largest active block, the cost of constructing each matrix element of $H^{\text{fr}}$ is estimated to be $\sim n_a^6$. Since the largest active block has $n_a^6$ elements, the net cost of constructing it is $\sim n_a^{12}$. 
\end{itemize} 
Building $\vec{b}$ costs $\sim (n_a^2 + n_c^2) n_v^2$ because it is obtained by reading a particular row or column of the $H^{\text{fr}}$ matrix, which has size $\sim (n_a^2 + n_c^2) n_v^2$. 

\begin{center}
    \textit{c. The fine print: preparing $\ket{b}$} 
\end{center} 

Once $\vec{b}$ has been constructed, the next natural step is to encode it into $\ket{b}$ via amplitude encoding. The presence of structure in $\vec{b}$ can lower the state preparation cost in this step. For example, in the SRLCC case, one can leverage Slater-Condon rules \cite{Slateryo, Condonyo} to show that matrix elements of $\vec{b}$ involving an excited determinant that arises from triple excitations and beyond are all zeroes, as well as invoke Brillouin theorem \cite{Surjan1989} to see that matrix elements involving singly excited determinants are all zeroes. We note that Brillouin theorem based arguments for vanishing single-excitation couplings do not apply to the icMRLCC case. However, it is unclear if there is any asymptotic state-preparation benefit when we perform this exercise, because although many excitation sectors vanish in $\vec{b}$, the double excitation amplitudes may accompany computational basis states across many qubits. It is not obvious if further reduction in gate count is possible in a general setting. 

In addition to leveraging Slater-Condon rules and Brillouin theorem, one could set numerical thresholding so that `small' amplitudes can be discarded in order to possibly save more CX gates as a convenient tool for early fault-tolerant era. We begin with $\ket{b}=\sum_{i} b_i\ket{i}$, where $\sum_{i}|b_i|^2=1$. We note that the coefficients $b_i$ belong to the normalized version of the input vector. We now define $\mathfrak{P}=\sum_{i\in \mathcal{S}}|b_i|^2$, where $\mathcal{S}$ is a subset of the original set of $b_i$. We note that $0<\mathfrak{P}\leq 1$. Then, we define 

\begin{equation}
\ket{\widetilde b}
=\frac{\sum_{i\in S}b_i\ket{i}}
{\sqrt{\sum_{i\in S}|b_i|^2}}
=\frac{1}{\sqrt \mathfrak{P}}\sum_{i\in S}b_i\ket{i}.
\label{eq:b_truncated}
\end{equation} 

We now introduce the discarded part of the input to be $\ket{r} \equiv \ket{b}-\sqrt \mathfrak{P} \,\ket{\widetilde b}=\sum_{i\notin S}b_i\ket{i}$. It follows that $\norm{\ket{r}}^2=1-\mathfrak{P}$. 

Assuming that $A$ is invertible, we introduce $\ket{x}_{\mathrm{un}}=A^{-1}\ket{b}$ and $\ket{\widetilde x}_{\mathrm{un}}=A^{-1}\!\left(\sqrt \mathfrak{P}\,\ket{\widetilde b}\right)$, so that $\ket{x}_{\mathrm{un}}-\ket{\widetilde x}_{\mathrm{un}}
=A^{-1}\ket{r}$. Consequently, 

\begin{equation}
\begin{aligned}
\norm{\ket{x}_{\mathrm{un}}-\ket{\widetilde x}_{\mathrm{un}}}
&\leq \norm{A^{-1}}_2\norm{\ket{r}}\\
&=\norm{A^{-1}}_2\sqrt{1-\mathfrak{P}}.
\end{aligned}
    \label{eq:unnormalized_error}
\end{equation} 

Moreover, since $\norm{\ket{b}}=1$, 

\begin{equation}
\begin{aligned}
 \frac{\norm{\ket{x}_{\mathrm{un}}-\ket{\widetilde x}_{\mathrm{un}}}}
{\norm{\ket{x}_{\mathrm{un}}}}
 &\leq \norm{A^{-1}}_2\norm{A}_2\sqrt{1-\mathfrak{P}}\\
 &=\kappa(A)\sqrt{1-\mathfrak{P}},
\end{aligned}
\label{eq:relative_error}
\end{equation}
where $\kappa(A)=\norm{A}_2\norm{A^{-1}}_2$. 

Next, we define the normalized versions of the solutions, $\ket{x}=\frac{\ket{x}_{\mathrm{un}}}
{\norm{\ket{x}_{\mathrm{un}}}}$, and $\ket{\widetilde x}
=\frac{\ket{\widetilde x}_{\mathrm{un}}}{\norm{\ket{\widetilde x}_{\mathrm{un}}}}$, and their fidelity, $F_x=\bigl|\langle x\vert\widetilde x\rangle\bigr|^2$. We also introduce the orthogonal projector onto the subspace perpendicular to $\ket{\widetilde x}$, $P_\perp=I-\ket{\widetilde x}\bra{\widetilde x}$. Since $P_\perp^\dagger=P_\perp$ and $P_\perp^2=P_\perp$,

\begin{eqnarray}
\norm{P_\perp\ket{x}}^2 
&=&\bra{x}P_\perp^\dagger P_\perp\ket{x} \nonumber \\
&=&\langle x |P_\perp| x\rangle \nonumber \\
&=&1- \langle x | \tilde{x} \rangle \langle \tilde{x} | x \rangle \nonumber \\
&=&1-F_x. 
\end{eqnarray} 

Also, since $P_\perp\ket{\widetilde x}=0$ implies that 
$P_\perp\ket{\widetilde x}_{\mathrm{un}}=0$, we get 

\begin{equation}
\begin{aligned}
1-F_x
&=\norm{P_\perp\ket{x}}^2\\
&=\frac{\norm{P_\perp\ket{x}_{\mathrm{un}}}^2}
{\norm{\ket{x}_{\mathrm{un}}}^2}\\
&=\frac{\norm{P_\perp
\bigl(\ket{x}_{\mathrm{un}}
-\ket{\widetilde x}_{\mathrm{un}}\bigr)}^2}{\norm{\ket{x}_{\mathrm{un}}}^2}. 
\end{aligned}
\label{eq:fidelity_projector}
\end{equation} 

An orthogonal projector cannot increase a vector's norm.
Combining this fact with Eq.~\eqref{eq:relative_error} gives 

\begin{equation}
\begin{aligned}
1-F_x&\leq
\frac{\norm{\ket{x}_{\mathrm{un}}
-\ket{\widetilde x}_{\mathrm{un}}}^2}
{\norm{\ket{x}_{\mathrm{un}}}^2}\\
 &\leq \kappa(A)^2(1-\mathfrak{P}). 
\end{aligned}
\label{eq:fidelity_bound}
\end{equation} 

Therefore, we arrive at our result: 

\begin{eqnarray}
1-F_x\leq\kappa(A)^2(1-\mathfrak{P}). 
\end{eqnarray} 

For a desired solution state fidelity $F_x\geq 1-\delta$,
where $0<\delta\ll 1$, a sufficient condition is
$\kappa(A)^2(1-\mathfrak{P})\leq\delta$, that is, 

\begin{equation}
\mathfrak{P} \geq 1-\frac{\delta}{\kappa(A)^2}. 
\label{eq:retained_probability}
\end{equation} 

The inequality indicates that for large $\kappa(A)$, nearly all the input probability, $\mathfrak{P}$, must be retained. Finally, we note that based on our numerical analyses that indicate a polylogarithmic growth in condition number for LCC matrices, the `discarded' probability, $1-\mathfrak{P}$, needs to scale at least  inverse polylogarithmically in $N$. Had $\kappa$ been poly in $N$, the restriction on $1-\mathfrak{P}$ is much more severe, as the quantity needs to be less than an inverse polynomial function of $N$. 

\begin{center}
    \textit{d. The fine print: loading $A$} 
\end{center} 

Since our methodology involves obtaining the $A$ matrix from the CISD matrix (taking QLS-SRLCCSD as a representative example), we have access to the elements $a_{ij}$, which in this case, enumerate to $n_{occ}^4 n_{vir}^4$. However, we want to finally implement as a Trotterized circuit the gate $R_z(\Theta_i)$, where $\Theta_i=f(h_i,t)$ angle's information is to be supplied as input. Furthermore, we need to know the $P_i$, that is, the Pauli word corresponding to the $i^{\mathrm{th}}$ Pauli gadget. Thus, we want a classical list $\{h_i,P_i\}$, or in other words details on a Pauli decomposition of $A$, that an experimentalist can make use of to implement rotation angles and basis rotation gates for each Pauli gadget. 

A recent work proposed a \textit{classical} algorithm to decompose a square matrix of size $N \times N$ into a sum of Pauli words in time $\mathcal{O}(N^2\log(N))$ \cite{Georges2025}. We note that for our LCCSD case, the resulting classical cost is $\mathcal{\tilde{O}}(n_v^4)$ assuming fixed $n_{\mathrm{occ}}$. Following this work, an article by Yosef and Avron proposed a quantum algorithm \cite{yosef} in the context of encoding matrices using quantum circuits. We build upon their Theorem 25 and accompanying Figs. 7 and 8 to obtain $|h_i|^2$ to within normalization constants and also $P_i$, that is, the pair $\{|h_i|^2,P_i\} \forall i$. We assume that we have access to an amplitude encoded version of $A$ as a $2n_b$ qubit state, $\ket{A} = \frac{1}{\|A\|}\sum_{ij} a_{ij} \ket{i}_r\ket{j}_c$, where subscripts $r$ and $c$ denote row and column indices' qubit registers respectively. $\ket{A}$ is passed as the input to a circuit as described in their Fig. 7. Then, by using the correspondence that each bitstring pair $r_i\ c_i$ can be identified with a $2 \times 2$ Pauli ($00 \rightarrow I$, $11 \rightarrow Z$, $01 \rightarrow X$, and $10 \rightarrow Y$), we can identify each of the $4^{n_b}$ bitstrings with a Pauli word and the corresponding histogram bar's height, that is, the probability of obtaining that bitstring with $|h_i|^2$ as $\mathrm{Prob}(P_i)=\frac{|h_i|^2}{\sum_j|h_j|^2}$. We make the simplistic assumption that coefficients of all Pauli words have the same magnitudes to estimate that the number of shots to observe all bitstrings corresponding to non-zero Pauli coefficients at least once scales as $\mathcal{\tilde{O}}(n_{vir}^4)$. Furthermore, the circuit requires $4n_b$ CX gates, $n_b$ $H$ gates, $2n_b$ $X$ gates, and crucially $3n_b$ $T$ gates. We arrive at the $T$ count by beginning with the two-qubit gate, $\mathcal{D}$, which occurs in the circuit shown in Fig. 7 of Ref. \cite{yosef}; the rest of the gates in the circuit are Cliffords. Fig. 6 of the same reference defines the gate as $\mathcal{D}=(I \otimes X)CS(I \otimes X)$). We now decompose $CS$ in terms of Cliffords and $T$ as $CS = (T \otimes T)CX(I \otimes T^\dag)CX$, and since there are $n_b$ such $\mathcal{D}$ gates, we arrive at the final $T$ count of $3n_b$. We note that we exclude input state preparation cost. However, we recall that we seek $\{h_i,P_i\}$, and not $\{|h_i|^2,P_i\}$. Since the state now contains $4^{n_b}$ amplitudes, obtaining the additional sign information amounts to solving a pure state quantum state tomography problem. Existing approaches permit state reconstruction with $\mathcal{\tilde{O}}(4^{n_b})=\mathcal{\tilde{O}}(N^2)$ queries, thereby performing worse relative to even the classical computing approach by Ref. \cite{Georges2025}. The question of whether the structure of LCC matrices allow for a more efficient extraction of amplitudes remains an open question. We conclude that at this point, the classical algorithm by Ref. \cite{Georges2025} seems to be the approach to take. In other words, while the QLS speed-up indication from our data itself remains, achieving end-to-end speed-up faces this bottleneck when one takes the explicit Pauli decomposition route. 

An alternative approach would be to find explicit quantum circuits that provide sparse access to $A$ without determining its Pauli decomposition. In the SRLCC case, the Slater–Condon rules offer a possible route for determining the locations of potentially non-zero matrix elements, while the corresponding values could be evaluated using stored one- and two-electron integrals. However, whether these operations can be implemented efficiently as quantum circuits, including the cost of coherent access to the integrals, remains to be investigated in future studies. 

\subsection{Post-processing costs} 

One of the oft-cited limitations of HHL is the prohibitive cost involved in extracting the output vector for large system sizes. The application of a QLS to quantum chemical problems need not involve reading $\vec{x}$, since one is typically interested in evaluating molecular properties using the elements of $\vec{x}$. As noted in earlier sections, the correlation energy can be obtained by appending a CSWAP circuit (or an ancilla-free version, the Hong-Ou-Mandel circuit) to the QLS circuit, whose inputs are $\ket{x}$, the output from a QLS, and $\ket{b}$. Although the exercise is beyond the scope of the current work, one could envisage computing molecular properties other than energy by using suitably block encoded versions of the operator of interest and using it in a Hadamard test module. The input to this module would be $\ket{x}$, the output from a QLS. It would be interesting to devise approaches on these lines in future works. 

\section{Proof-of-concept icMRLCC energy calculations}\label{sec:calc} 

Now that our analyses have indicated prospects of quantum advantage for the QLS-LCC problem, a natural follow up given the importance of the field of quantum chemistry on quantum computers is to see if practical calculations yield good accuracies, for example, with limited $n_r$. While quantum hardware computations are outside the scope of our work, we perform simulations of potential energy curves (PECs) on three model molecules, LiH, H$_4$, and BeH$_2$, that are known to exhibit multi-reference nature along their PECs. We pick the simplest QLS, that is, HHL, for our computations. We discuss the computational details and our results in the succeeding paragraphs. 

\subsection{Computational details} 

For the LiH, BeH$_2$, and H$_{4}$ molecules, we used the 6-31G split-valence basis sets. We run all HHL simulations on Qiskit 0.37.2. We systematically increased $n_r$, which we recall is the number of clock register qubits, until we reached satisfactory accuracy in our energies. To that end, we report our results with $n_r$ set to 8 for LiH, 9 for BeH$_2$, and 6 for H$_4$. We also note that the $A$ matrix sizes (which in turn depend on the number of orbitals considered for our calculations) as well as the choice of $n_r$ depended on our access to high-performance computing resources. Our qubit counts stated in the format $n_r + 2n_b + 1$ are given below: 

\begin{itemize}
\item LiH: $8 + 2 \times 8+1=25$, 
\item H$_4$: $6 + 2 \times 6 +1=19$, and 
\item BeH$_2$: $9 + 2 \times 9+1 = 28$. 
\end{itemize} 

Since BeH$_{2}$ has more electrons compared to other systems, which in turn leads to more qubits for our HHL computations, we considered only four $s$ orbitals and 3 $p$ orbitals. 

We now comment on the choice of the reference active space. For CASSCF computations, the choice of the active orbitals is essential to capture the maximum amount of static correlation. By `reference active space', we mean the space composed of active orbitals (having only partial occupancy) and active electrons. We denote a reference active space as $(N_a, n_a/2)$. For the LiH system, we consider the $2s$ orbital of Li and the $1s$ orbital of H as the active orbitals. Therefore, the active space is $(2, 2)$. For H$_{4}$, the $1s$ orbital of each H atom is considered; therefore, the active space is $(4, 4)$. In BeH$_{2}$, the $2s$ and $2p$ orbitals of Be and the $1s$ orbital of each H are included, resulting in a $(4, 6)$ active space. We note that the active space chosen for the BeH$_2$ energy calculations differs from the $(2,2)$ active space employed for the $\kappa$ scaling analysis in Sec. \ref{sec:icmr-scaling} (see Table \ref{tab:systems}). 

\subsection{HHL-icMRLCCSD numerical results} \label{results} 

The comparison between standard single-reference methods, CCSD and CCSD(T), and multi-reference methods, icMRCISD and HHL-icMRLCCSD, is illustrated by plotting $E^{\text{SR}/\text{icMR}} - E_{FCI}$ where $E^{\text{SR}/\text{icMR}}$ represents the total energy obtained from each of these methods. Although $E^{\text{SR}/\text{icMR}}$ and $E_{FCI}$ are total energies, their difference still tells us about how much correlation energy is captured. 

\begin{center}
    \textit{a. LiH} 
\end{center} 

Fig. \ref{Ediff:lih-h4-beh2}(a) presents our findings as we stretch the bond distance between the two constituent atoms of LiH. Throughout the PEC, both the single- and multi-reference methods agree reasonably well, and differ only at the sub-mHa level in their energies. For context, the total FCI energy is $-7.998358$ Ha at equilibrium geometry (3.023 bohr), and the correlation energy in the same geometry is $-0.019036$ Ha, and thus the deviation between single- and multi-reference methods across the PEC is two orders of magnitude smaller than the correlation energy at equilibrium geometry. As the figure shows, the disagreement between the single- and multi-reference methods becomes apparent, especially at stretched geometries. The results from icMRCISD are in close agreement with FCI, indicating that across the range of bond lengths considered, not only do triples and quadruples not contribute to energy significantly (the close agreement between CCSD and CCSD(T) results also reflect the diminished importance of triples), but also the active space built from the two valence electrons (unsurprisingly) matters the most. On the other hand, HHL-icMRLCCSD, while in good agreement most of the time, shows some deviations at specific geometries. This is not necessarily correlated to the ability of icMRLCCSD in capturing multi-reference effects in those geometries, but rather at least in part to the inadequacy in the number of HHL ancilla qubits that we have chosen in capturing eigenvalues. However, we do not attempt to isolate this source of error from inaccuracies inherent in the icMRLCCSD approximation. We finally note that although our interest was in stretching the bond, we also obtain energies when we lower the bond length below the equilibrium value. Our data indicate that a multi-reference character is not too important when the bond is compressed, within the range of considered geometries. In fact, at 2.000 bohr, we find that the two most important CASSCF coefficients correspond to 0.991937 and -0.126724 configurations ($|220\rangle$ and $|202\rangle$). On the other hand, at about 7.550 bohr, 0.810444 and -0.585814 matter the most ($|220\rangle$ and $|202\rangle$). 

\begin{figure*}[t]
\centering
\hspace{-1.5cm}
\includegraphics[scale=0.72]{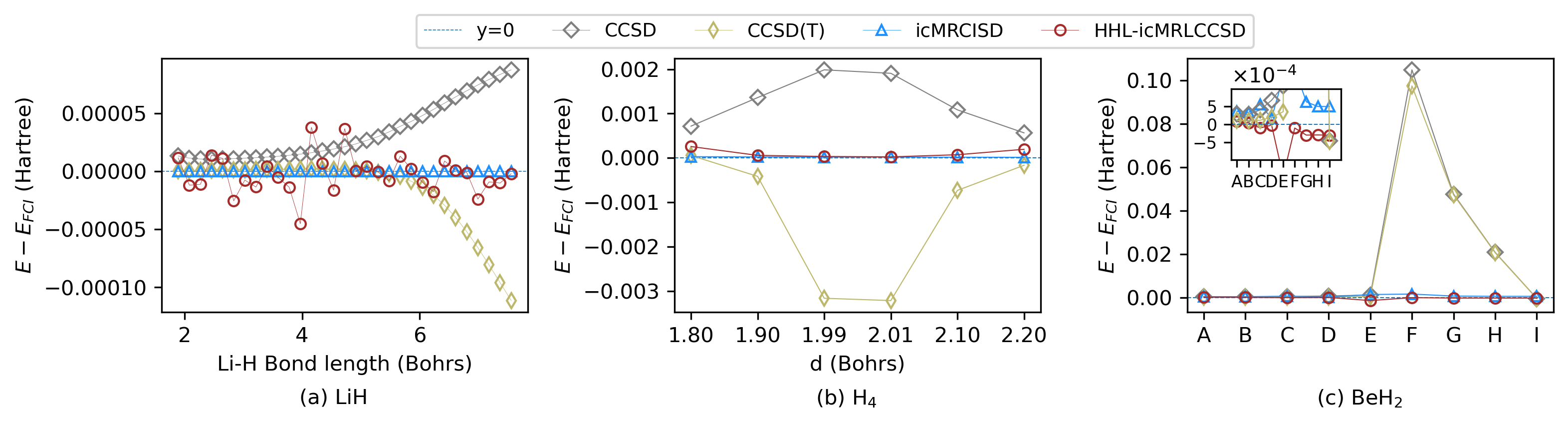}
\caption{Figure showing the energy difference between the total energy and the FCI energy of LiH, H$_4$ and BeH$_2$ between CCSD, CCSD(T), icMRCISD and HHL-icMRLCCSD at given stretched geometries. } \label{Ediff:lih-h4-beh2}
\end{figure*}

\begin{center}
    \textit{b. H$_{4}$} 
\end{center} 

Our data from Fig. \ref{Ediff:lih-h4-beh2}(b) show that the two multi-reference methods, HHL-icMRLCCSD and icMRCISD, are the most stable throughout the range of $d$ values considered for the PEC. We recall that $d$ is the separation between the two H$_2$ sub-systems, as the third panel of Fig. \ref{fig1} shows. Furthermore, the two H atoms in each of the H$_2$ sub-systems are separated by 2.000 bohr. Our range for $d$ varies by $0.200$ bohr on either side of $d=2.000$ bohr, which is the square geometry; these are the distances where strong correlation effects are expected to be present with them being strongest in the vicinity of $d=2.000$ bohr. Returning to our results, we find that the icMRCISD approach provides results that are within $0.022$ mHa agreement with FCI, while HHL-icMRLCCSD gives results within $0.257$ mHa agreement with FCI (see Table \ref{taballmolecules}). We also note that the maximum correlation energy obtained using FCI is $\sim 0.100$ Ha at the two near-square geometries ($d=1.990$ bohr and $d=2.010$ bohr). This indicates that a substantial part of the correlation effects have been captured by both HHL-icMRLCCSD and icMRCISD. With regard to the single-reference theories (CCSD and CCSD(T)), they also provide reasonable results at $d=1.800$ bohr and $d=2.200$ bohr, yielding deviations of $0.718$ mHa and $0.566$ mHa for CCSD, and $0.005$ mHa and $0.165$ mHa for CCSD(T), respectively with respect to FCI. Our results from CASSCF calculations show that even at these geometries, the HF determinant is not the most dominant configuration: $-0.296014 \ket{2200} + 0.932614\ket{2020}$ at $d=1.800$ bohr and $-0.347915 \ket{2200}+0.907368\ket{2020}$ at $d=2.200$ bohr. Interestingly, the inclusion of perturbative triples worsens the agreement with FCI near the square geometry, with the absolute deviations increasing from $1.987$ and $1.909$ mHa for CCSD to $3.164$ and $3.216$ mHa for CCSD(T) at $d=1.990$ and $2.010$ bohr, respectively. Lastly, we notice the superiority of multi-reference methods at the two near-square geometries, as CCSD and CCSD(T) shift from sub-mHa to mHa separation from FCI, whereas icMRCISD deviations remain at $0.001$ mHa, and HHL-icMRLCCSD deviations are $0.033$ and $0.024$ mHa at $d=1.990$ bohr and at $d=2.010$ bohr, respectively. Once again, this is expected because at $d$=1.990 bohr and $d=2.010$ bohr, two configurations are important in the active space. The CASCSF state is $-0.669528 \ket{2200} + 0.709094 \ket{2020}$ at $d=1.990$ bohr, and $-0.669486 \ket{2200} + 0.708673\ket{2020}$ at $d$=2.010 bohr. 

\begin{center}
    \textit{c. BeH$_{2}$} 
\end{center} 

Next, we consider the Be insertion problem \cite{Purvis-1983}, which happens when a Be atom is inserted into H$_2$ molecule to form BeH$_2$. As the Be atom gets closer to H$_{2}$ and the latter's bond gets stretched, the $2p$ orbital of Be interacts with the anti-bonding orbital of H$_{2}$, and the entire compound transitions to BeH$_2$ through an avoided crossing between the ground and first excited singlet state. Fig. \ref{Ediff:lih-h4-beh2}(c) shows that all four methods we consider perform reasonably well across geometries A through D, yielding $0.024$ mHa to $0.670$ mHa energy differences with respect to FCI. At the geometry E, CCSD, icMRCISD, and HHL-icMRLCCSD energy differences start deviating from FCI by giving energy differences of  $1.073$, $1.325$, and $1.454$ mHa respectively (See Table \ref{taballmolecules}). We also note that the deviations are well below the scale of the correlation energy at geometry E, which is $\sim 64.000$ mHa. While other methods tend to perform worse, CCSD(T) performs the best by giving $0.349$ mHa difference with FCI, showing the importance of including perturbative triples. Point F shows that both CCSD and CCSD(T) start performing worse, giving a result of $\sim 104.000$ mHa and  $\sim 97.000$ mHa, respectively, whereas the correlation energy is $\sim 195.000$ mHa. This means that CCSD and CCSD(T) could capture only $\sim  50 \%$ of correlation effects at that geometry. On the other hand, HHL-icMRLCCSD perform very well by giving an energy difference of $\sim  0.100$ mHa, while icMRCISD stays stable with an energy difference of $\sim 1.500 $ mHa. The CASSCF state at that geometry is $0.898607\ket{2220000} -0.300455\ket{2200002} + \cdots$ with two important configurations, showing the importance of using multi-reference theories under such circumstances. The same situation occurs at points G and H with the dominance of icMRCISD and HHL-icMRLCCSD. Finally, at point I, all four methods stabilize within a sub-mHa energy difference from FCI. 

\begin{table}[!h]
\begin{tabular}{ccc}
\hline\hline
Label & \multicolumn{2}{c}{Coordinates of Hydrogens (($x, y, z$) in bohr)} \\\hline
                   & H                      & H                     \\
A             & (0.000, $2.540$, 0.000)       & (0.000, $-2.540$, 0.000)       \\
B            & (0.000, $2.080$, 1.000)     & (0.000, $- 2.080$, 1.000)    \\
C             & (0.000, $1.620$, 2.000)     & (0.000, $-1.620$, 2.000)     \\
D            & (0.000, $1.390$, 2.500)     & (0.000, $-1.390$, 2.500)     \\
E             & (0.000, $1.275$, 2.750)   & (0.000, $-1.275$, $2.750$)    \\
F             & (0.000, $1.160$, 3.000)      & (0.000, $-1.160$, 3.000)     \\
G             & (0.000, $0.930$, 3.500)      & (0.000, $- 0.930$, 3.500)    \\
H            & (0.000, $0.700$, 4.000)     & (0.000, $-0.700$, 4.000)    \\
I            & (0.000, $0.700$, 6.000)    & (0.000, $- 0.700$, 6.000)  \\\hline
\end{tabular}
\caption{Table giving the coordinates of the two hydrogen atoms in BeH$_2$. Be is fixed at the origin throughout. } \label{tabbeh2}
\end{table} 

\section{Conclusion}\label{sec:conclusion} 

In summary, we investigated the prospects of quantum advantage using quantum linear solvers (QLSs) for quantum chemistry, and in particular, `solving’ the linearized coupled cluster (LCC) equations. As the previously proposed single-reference framework (QLS-SRLCC) is suitable only for weakly correlated systems, we first extend QLS-LCC to a multi-reference setting by adopting the QLS-internally contracted multi-reference LCC (QLS-icMRLCC) formalism, more precisely the MRCEPA(0) approximation. We show that both the SRLCC and the icMRLCC amplitude equations can be expressed as linear systems, naturally motivating the use of QLSs for their solution. 

The central objective of this work is to assess how the condition number ($\kappa$) and sparsity ($s$) of the $A$ matrix scale with $N$, since these quantities play a key role in determining prospects of advantage. Our numerical survey encompasses 8 atomic and molecular systems and 3 families of atomic chains, across 7 single-particle basis set choices. For the atomic and molecular systems considered, we increase the number of virtual orbitals to vary $N$, while for chains, we investigate system size growth by increasing the number of atoms from 2 to 14. For the SRLCC calculations, we considered matrix dimensions reaching $N=21882$ (and number of virtual orbitals going up to 113) and encountered $\kappa$ values as large as $22288$. Our results illustrate that even comparatively large $\kappa$ occurrences can accompany a slow growth in system size. For the icMRLCCSD case, focusing on seven systems for which $A>0$ over all sampled system sizes, we investigated matrix dimensions reaching $N=6868$. While $N$ increases by factors of approximately 21–129, $\kappa$ increases only by factors of approximately 1.25–6.04. These observations indicate favourable $\kappa$ scaling over the sampled ranges, although the presence of staircase-like variations and the absence of a unique asymptotic scaling law warrant detailed future investigations. We emphasize that the observed favourable icMRLCC $\kappa$ scaling is dependent on the choice of reference active space, since alternative choices can yield indefinite $A$ matrices and isolated sharp increases in $\kappa$. Understanding the onset of indefiniteness and its implications for QLS-icMRLCC computations is an important direction for future investigations. Furthermore, our present analysis employs a fixed reference active space for each system as the virtual space grows. A more systematic exercise that selects a suitable active space at each system size, while retaining the accuracy needed to describe strong correlation, may yield different $\kappa$ scaling behaviour. Thus, the current observations from our analyses should not be regarded as a fundamental limitation of the icMRLCC framework itself. 

Since explicitly determining $\kappa$ is computationally challenging, we employ three complementary diagnostics: direct computation of $\kappa$, the ratio of the maximum to the minimum diagonal entries of $A$, and an adaptation of the edge-spawning conjecture proposed for graph-based linear systems. With regard to the diagonal ratio, we first establish a bound based on strict diagonal dominance and Gershgorin's theorem. Although this condition is not satisfied by our sampled matrices whose number of virtual orbitals are large, it provides a natural starting point for the importance of diagonal entries of the $A$ matrix in providing information on $\kappa$. Subsequently, we employ Weyl's inequality to establish an alternative spectral norm-based condition under which the diagonal ratio could approximate $\kappa$. Our numerical results indicate that the diagonal ratio can track condition number in a majority of the systems that we consider even in the absence of strict diagonal dominance. Together with the qualitative indications from the edge-spawning conjecture, our three diagnostics provide broadly consistent evidence of slow $\kappa$ growth for the sampled systems, with approximately polylogarithmic trends in most cases. An exception is H$_4$ in the SRLCCSD cc-pVQZ basis, where an approximately linear dependence on the number of virtual orbitals is observed. Our runtime analysis shows that even such polynomial growth can preserve the prospects of exponential separation when using near-optimal QLSs such as CKS. Nevertheless, our numerical results are based on finite system sizes and do not establish general asymptotic scaling guarantees. It is worth placing our findings in the context of previous QLS applications. Several investigations involving problems such as hydrological linear systems, computational fluid dynamics, and DC power flow have identified unfavourable $\kappa$ behaviour or substantial conditioning-related limitations \cite{Golden2022, Gopal2024, pareek_2026}. Against this backdrop, the relatively slow variation of $\kappa$ observed in our numerical survey is noteworthy, including for the selected icMRLCC reference active spaces. Although our results do not establish an asymptotic scaling law as stated earlier, they identify LCC matrices as promising candidates for further investigations of quantum advantage with QLSs. 

We further argue that $s$ is bounded by $O(n_v^2)=O(N^{2/E})$ for $\mathfrak{E}\geq2$ and when we fix the number of occupied orbitals, with $\mathfrak{E}$ being the excitation rank, for both SRLCC and icMRLCC cases. Therefore, it is sub-linear in $N$ when $\mathfrak{E}>2$. Combining our findings on sparsity with the assumed favourable $\kappa$ scaling, we identify regimes supporting exponential separation between the representative classical solver and QLSs at the solver level; for HHL beyond LCCSD, and for near-optimal QLSs such as CKS already at the LCCSD level of theory. We emphasize that the exponential separation discussed here is with respect to the number of qubits in the quantum state register, which scales as $\log(N)$. Since $N\sim n_v^{\mathfrak E}$ for fixed excitation rank, this corresponds to a polynomial separation in the number of virtual orbitals. While our SRLCC results provide encouraging evidence for the realization of such an advantage, its applicability to icMRLCC remains less clear, owing to the sensitivity of the spectral properties of the corresponding matrices to the reference active space choice. Thus, establishing whether a comparable quantum advantage can be realized generally within icMRLCC requires further investigation. 

We also go beyond asymptotic analyses and examine (logical) CX gate counts as well as logical $T$ gate counts for an HHL-SRLCCSD implementation taken as a representative example. We also zoom out and consider the entire end-to-end workflow for an SRLCCSD energy calculation, and identify that pre-processing steps remain as obstacles to an end-to-end quantum advantage. In particular, classical construction of $A$ and $\vec{b}$, implementing quantum access to $A$, input state preparation, and output extraction introduce additional costs that must be accounted for. We also derive a bound relating the solution state fidelity to the probability retained when truncating the input vector, via $\kappa$. Such considerations underscore the importance of evaluating an entire quantum computational workflow rather than only the QLS sub-routine in isolation. 

To complement our scaling analysis, we perform proof-of-concept numerical simulations on LiH, H$_4$, and BeH$_2$ that require up to 28 qubits. The resulting ground state energies show good agreement with FCI benchmark results, and in particular, at the strongly correlated geometries for H$_4$ and BeH$_2$, the multi-reference approaches are substantially more accurate than their single-reference counterparts, underscoring the importance of extending QLS-LCC beyond the single-reference formalism. 

Looking ahead, several opportunities and challenges remain on the road to advantage with QLS-LCC, including developing efficient quantum algorithms for matrix access as well as input state preparation, extending our numerical survey to larger and even more diverse molecular systems, and studying the integration of QPE and QLS via approaches such as embedding. Another interesting direction for future investigation concerns the role of virtual-orbital ordering in determining the observed $\kappa$ scaling. Preliminary investigations in both SRLCCSD and icMRLCCSD indicate that reordering the virtual orbitals often has a relatively low impact on the observed trends. However, exceptions exist, such as H$_4$ in the SRLCCSD cc-pVQZ basis, where reordering yields regions of apparently slow, approximately polylogarithmic growth punctuated by pronounced jumps in $\kappa$, in contrast to the smoother, approximately linear behaviour obtained without reordering. A systematic examination of the influence of orbital ordering on $\kappa$ scaling, including the origin of these jumps, is therefore left for future work.
Furthermore, our findings motivate the exploration of alternative multi-reference electronic structure theories whose linearized formulations yield matrices with more favourable spectral properties, particularly with respect to $\kappa$ scaling. Developing such formulations while retaining the ability to accurately describe strongly correlated systems, together with favourable sparsity and quantum access costs, could be an important direction for future research. Such investigations could help identify multi-reference approaches better suited for achieving quantum advantage using QLSs. Overall, we view the present work as a feasibility assessment of QLS for quantum chemistry (LCC), identifying both the favourable matrix properties that motivate further pursuit of quantum advantage and the outstanding algorithmic and implementation challenges that must be addressed to realize such an advantage in an end-to-end workflow. \\ 

\begin{acknowledgements}
We acknowledge the National Supercomputing Mission (NSM) for providing computing resources of `PARAMSiddhi-AI', under National PARAM Supercomputing Facility (NPSF), C-DAC, Pune, and supported by the Ministry of Electronics and Information Technology (MeitY) and Department of Science and Technology (DST), Government of India. VSP acknowledges support from CRG grant (CRG/2023/002558), as well as from the National Quantum Mission, an initiative of the Department of Science and Technology, Govt. of India. VSP also acknowledges the support provided by the Foundation for QC Innovation (FQCI), DST-NQM T-Hub at IISc Bengaluru, in facilitating this project. 

VSP and PBT thank Prof. Andreas Köhn for valuable help with setting up the icMRLCC computations and helpful comments on the manuscript, as well as VSP's group members for carefully reading the manuscript. VSP used ChatGPT 5.6 for some GAMESS and GeCCO software related discussions, as well as help with the derivation discussed in Section VIIA(c). PBT and VSP also thank Prof. Debashis Mukherjee for his support in the initial days of the work. PBT acknowledges Prof. Ankan Paul for providing access to Molpro at the Indian Association for the Cultivation of Science (IACS), and Dr. Rounak Nath for Molpro access support. PBT acknowledges Dr. Dibyajyoti Chakravarti for useful discussions on icMRCC theories. 
\end{acknowledgements} 

\bibliography{references}

\clearpage 

\onecolumngrid 

\appendix 

\renewcommand{\thefigure}{A\arabic{figure}}
\renewcommand{\thetable}{A\arabic{table}}

\setcounter{figure}{0}
\setcounter{table}{0} 

\begin{table*}[t]
\caption{Acronyms used throughout this work, arranged in alphabetical order in this table. }
\label{tab:acronyms}
\begin{ruledtabular}
\begin{tabular}{ll}
Acronym & Definition \\
\hline
AO & Atomic orbital \\ 
BCH & Baker-Campbell-Hausdorff \\ 
CASCI & Complete active space configuration interaction \\
CASSCF & Complete active space self-consistent field \\
CC & Coupled cluster \\
CCSD & Coupled cluster with singles and doubles \\
CCSD(T) & Coupled cluster with singles, doubles, and perturbatively included partial triples \\
CEPA(0) & Unshifted coupled electron pair approximation\\
Cl & Reference classical algorithm for linear systems \\ 
CI & Configuration interaction \\ 
CISD & Configuration interaction with singles and doubles \\
CKS & Childs-Kothari-Somma \\ 
FCI & Full configuration interaction \\
GUGA & Graphical unitary group approach \\
HF & Hartree-Fock \\
HHL & Harrow-Hassidim-Lloyd \\
HOM & Hong-Ou-Mandel \\
icMR & Internally contracted multi-reference \\ 
icMRCI & Internally contracted multi-reference configuration interaction \\
icMRCISD & Internally contracted multi-reference configuration interaction with singles and doubles \\
icMRLCC & Internally contracted multi-reference linearized coupled cluster \\
icMRLCCSD & Internally contracted multi-reference linearized coupled cluster with singles and doubles \\
icMRLCCSDT & Internally contracted multi-reference linearized coupled cluster with singles, doubles, and triples \\
LCC & Linearized coupled cluster \\
LCCSD & Linearized coupled cluster with singles and doubles \\
LCCSDT & Linearized coupled cluster with singles, doubles, and triples \\
LCCSDTQ & Linearized coupled cluster with singles, doubles, triples, and quadruples \\
MO & Molecular orbital \\ 
MR & Multi-reference \\
MRLCC & Multi-reference linearized coupled cluster \\
NISQ & Noisy intermediate-scale quantum \\
PEC & Potential energy curve \\
QFT & Quantum Fourier transform \\ 
QLS & Quantum linear solver \\
QLS-SRLCC & Quantum linear solver for single-reference linearized coupled cluster \\
QPE & Quantum phase estimation \\
SR & Single-reference \\
SRLCC & Single-reference linearized coupled cluster \\
SRLCCSD & Single-reference linearized coupled cluster with singles and doubles \\
SRLCCSDT & Single-reference linearized coupled cluster with singles, doubles, and triples \\
SRLCCSDTQ & Single-reference linearized coupled cluster with singles, doubles, triples, and quadruples \\
VQE & Variational quantum eigensolver \\
\end{tabular}
\end{ruledtabular}
\end{table*} 

\clearpage 

\section{The HHL algorithm} \label{app:hhl}

We begin by expanding $A$ in terms of its eigenvectors, $\ket{u_i}$, and eigenvalues, $\lambda_i$, to get the expression: 

\begin{equation}
    A = \sum_i \lambda_i \ket{u_i} \bra{u_i}. 
\end{equation} 

Furthermore, $\ket{b}$ can be expanded in the eigenbasis $\{\ket{u_i}\}$ of $A$ in the following way: 
\begin{equation}
    \ket{b} = \sum_i b_i \ket{u_i}; \sum_i |b_i|^2 = 1, 
\end{equation} 

combining both of which leads to: 

\begin{equation}
    \ket{x} = A^{-1} \ket{b} = \sum_i b_i \lambda_i^{-1} \ket{u_i}.
\end{equation} 

The HHL algorithm is therefore designed to extract the eigenvalues of $A$, invert them as shown in the above equation, and extract $\ket{x}$ (or an approximation to it). Although standard HHL assumes a Hermitian matrix $A$, non-Hermitian linear systems can be reformulated as Hermitian systems through a suitable embedding. 

We now enumerate the different steps involved in the HHL algorithm: 

\begin{enumerate} 
\item We begin with three registers: the state register composed of $n_b$ qubits and prepared in the state $\ket{b} = \frac{\sum_i b_i \ket{u_i}}{\|\sum_i b_i \ket{u_i}\|}$, the clock register containing $n_r$ qubits intended to be used as QPE ancillas and initialized to $\ket{0^{n_r}}$, and a single qubit register for the HHL ancilla. 
\item Upon applying QPE on $\ket{0^{n_r}} \otimes \ket{b}$, the phases of $e^{i A t}$ are collected in the form $\frac{\lambda_i t}{2 \pi}\ \forall i \in [1,N]$, that is, 

\begin{equation}
\text{QPE}(\ket{0^{n_r}} \otimes \ket{b}) = \sum_i b_i \bigg|\frac{2^{n_r}\lambda_i t}{2\pi}\bigg\rangle \otimes \ket{u_i}. 
\end{equation} 

The right-hand side can be expressed as: 

\begin{equation} 
\sum_i b_i \ket{\tilde{\lambda}_i} \otimes \ket{u_i}, 
\end{equation} 

where $\tilde{\lambda}_i = \frac{2^{n_r}\lambda_i t}{2\pi}$ is an approximation to the rescaled eigenvalues $\frac{\lambda_i t}{2 \pi}$ in decimal representation.
\item Next, we perform the controlled rotation gate $e^{-i \theta Y}$ with $\theta_i = \arcsin{\left(\frac{C}{\tilde{\lambda}_i}\right)}$. $C$ is a scalar that can be tuned such that the probability of obtaining the outcome $\ket{1}$ at the HHL ancilla is maximal. For our numerical calculations, we set the standard choice of $C=\lambda_{min}$ for convenience, but we note that a suitable rescaling of $A$ in the QPE step, especially for diagonally dominant matrices such as those in chemistry, can circumvent this inefficiency \cite{Nishanth2023}. At the end of this step, we obtain: 

\begin{equation}
\sum_i b_i \left(\sqrt{1-\frac{C^2}{\tilde{\lambda}^2_i}} \ket{0} + \frac{C}{\tilde{\lambda}_i} \ket{1} \right) \otimes \ket{\tilde{\lambda}_i} \otimes \ket{u_i}.
\end{equation} 

\item Uncompute the register $\ket{\tilde{\lambda}_i}$ by applying QPE$^\dagger$ on $ \sum_i b_i \ket{\tilde{\lambda}_i} \otimes \ket{u_i}$ to obtain: 

\begin{equation}
\sum_i b_i \left(\sqrt{1-\frac{C^2}{\tilde{\lambda}^2_i}} \ket{0} + \frac{C}{\tilde{\lambda}_i} \ket{1} \right) \otimes \ket{0^{n_r}} \otimes \ket{u_i}.
\end{equation} 

\item Finally, we measure the HHL ancilla in the computational basis and post-select the state outcome $\ket{1}$. The obtained normalized solution appears on the state register as follows: 

\begin{equation}
\ket{\tilde{x}} = \frac{1}{\sqrt{\sum_i \left|\frac{b_i C}{\tilde{\lambda}_i} \right|^2}}
\sum_i \frac{b_i C}{\tilde{\lambda}_i} \ket{u_i}. 
\end{equation}
\end{enumerate} 

\begin{figure}[!h] 
\centering
   \includegraphics[scale=0.6]{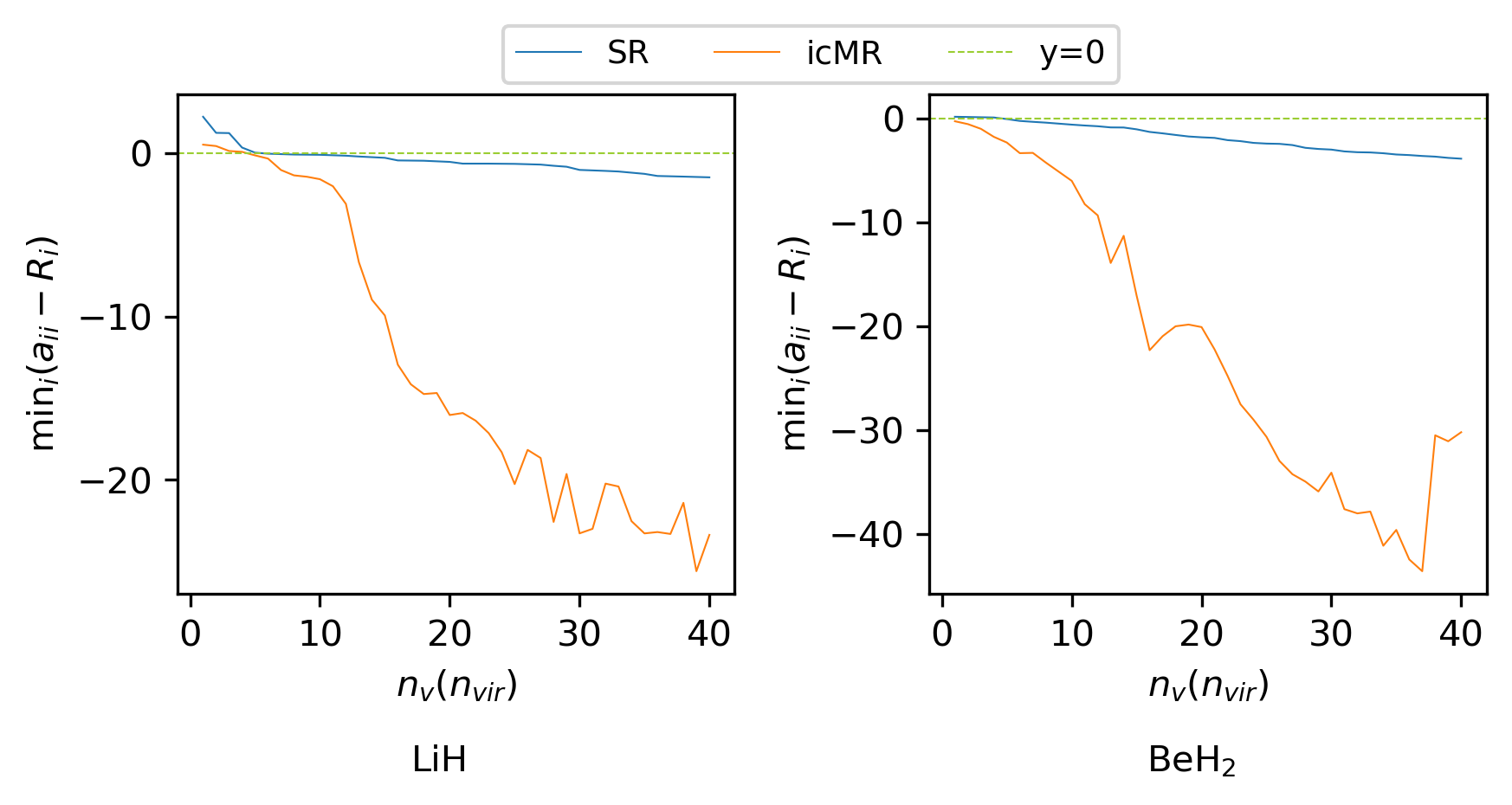}  
\caption{This figure depicts $\text{min}_i (a_{ii} - R_i) $ versus the number of virtuals $n_v (n_{vir})$ included for representative examples from the SR and the icMR cases.} \label{Aii-Ri_srmr:lihbeh2} 
\end{figure} 

\begin{figure}[!h] 
\centering
   \includegraphics[scale=0.52]{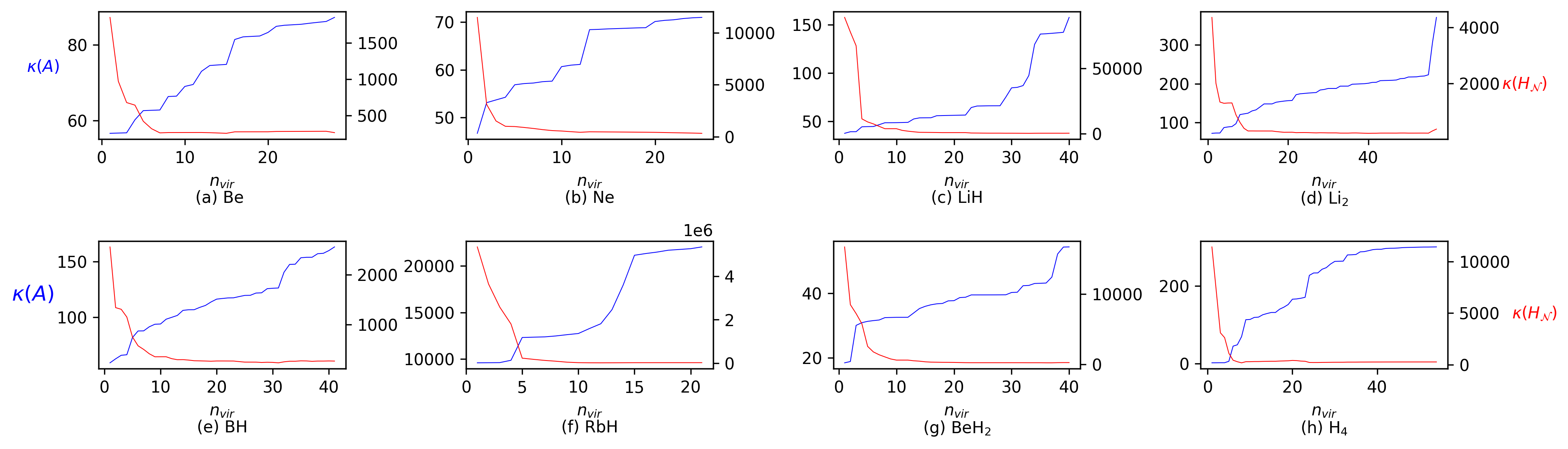}
\caption{This figure depicts condition numbers of the $H_{\mathcal{N}}$, $\kappa(H_\mathcal{N})$, and condition number of $A$, $\kappa(A)$ vs the number of virtual orbitals ($n_{vir}$) and in the SRLCCSD case for the considered molecules. All systems are in the cc-pVTZ basis, except RbH which is in the Sapporo DKH3 DZP 2012 basis.} \label{kappaA_kappaHsr:allmolecules} 
\end{figure} 

\begin{figure}[!h] 
\centering
   \includegraphics[scale=0.5]{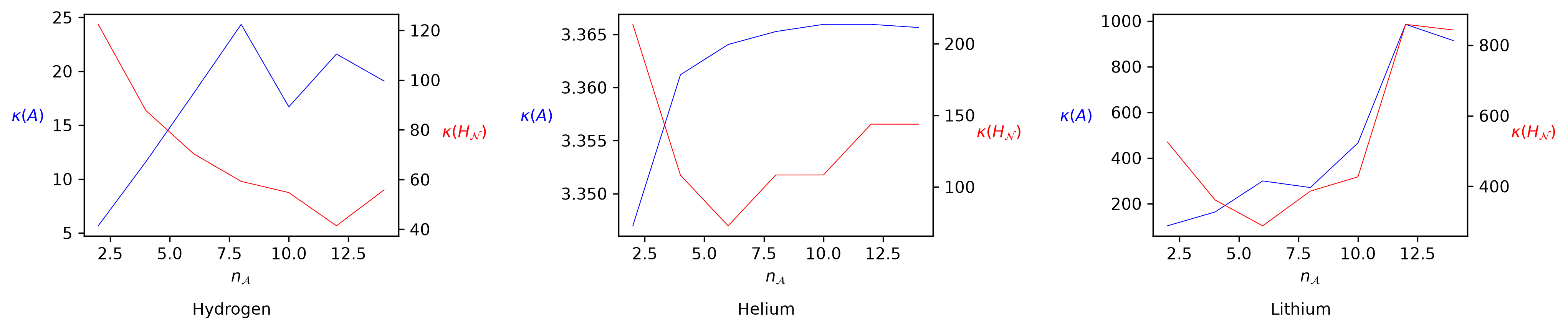}  
\caption{This figure depicts the condition numbers of $H_{N}$, $\kappa(H_\mathcal{N})$, and condition number of $A$, $\kappa(A)$ vs the number of atoms ($n_{\mathcal{A}}$) included in the chain and in the SRLCCSD case.} \label{kappaA_kappaHmr:chains} 
\end{figure} 

\begin{figure}[!h] 
\centering
   \includegraphics[scale=0.55]{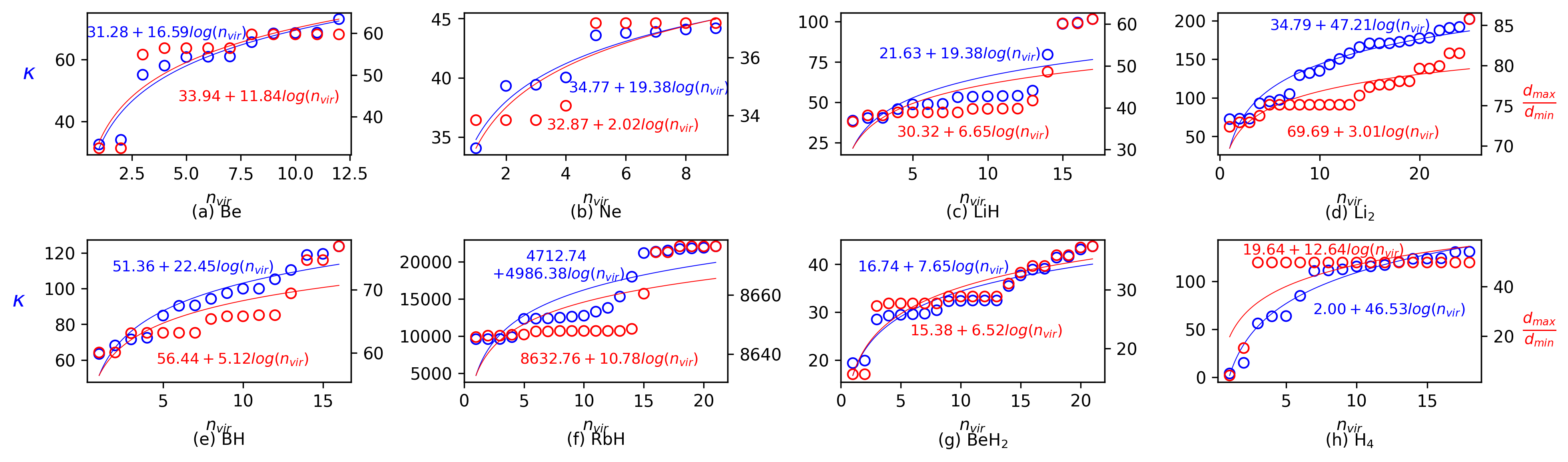}  
\caption{This figure shows the SRLCCSD scaling of the condition number, $\kappa$, and the ratio of the maximum and minimum diagonal entries of $A$, $d_{\mathrm{max}}\big/d_{\mathrm{min}}$, with the number of virtuals, $n_{vir}$. All the data are presented in the cc-pVDZ basis set, except RbH, which is presented in the Sapporo DKH3 DZP 2012 basis set. For H4, the virtual orbitals were reordered (see Ref. \cite{Peniel2025}) to obtain the approximately polylogarithmic trends shown. } \label{overall_srlcc_kappa_dmaxbydmn_ccpvdz} 
\end{figure} 

\begin{figure}[!h] 
\centering
   \includegraphics[scale=0.55]{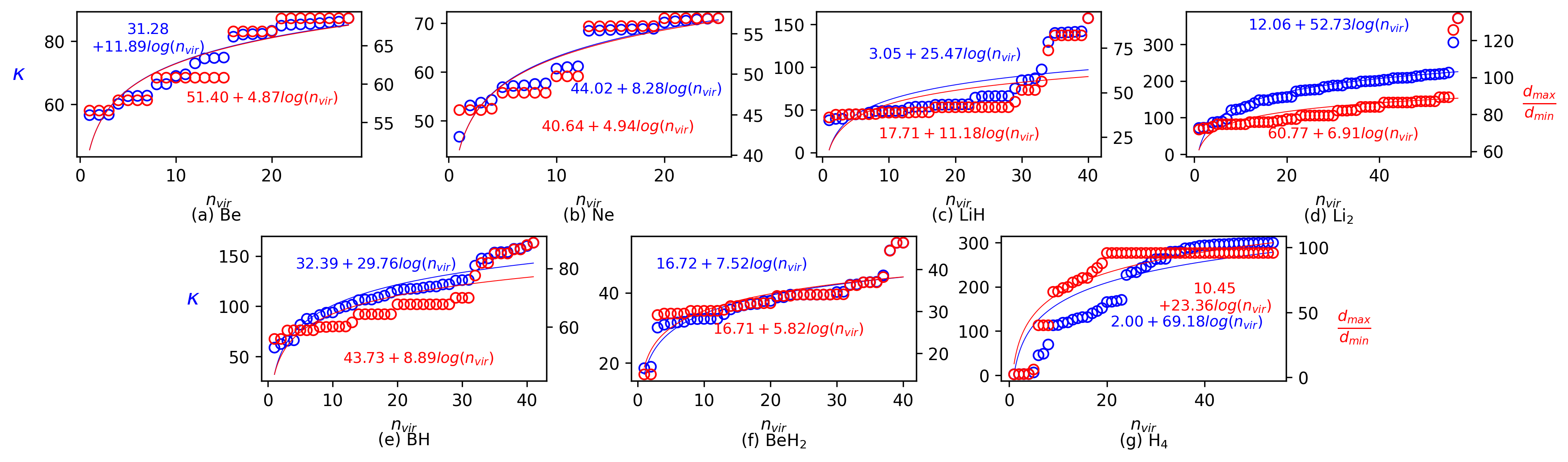}  
\caption{This figure shows the SRLCCSD scaling of the condition number, $\kappa$, and the ratio of the maximum and minimum diagonal entries of $A$, $d_{\mathrm{max}}\big/d_{\mathrm{min}}$, with the number of virtuals, $n_{vir}$. All the data are presented in the cc-pVTZ basis set. For H4, the virtual orbitals were reordered (see Ref. \cite{Peniel2025}) to obtain the approximately polylogarithmic trends shown.} \label{overall_srlcc_kappa_dmaxbydmn_ccpvtz} 
\end{figure}  

\begin{figure}[!h] 
\centering
   \includegraphics[scale=0.55]{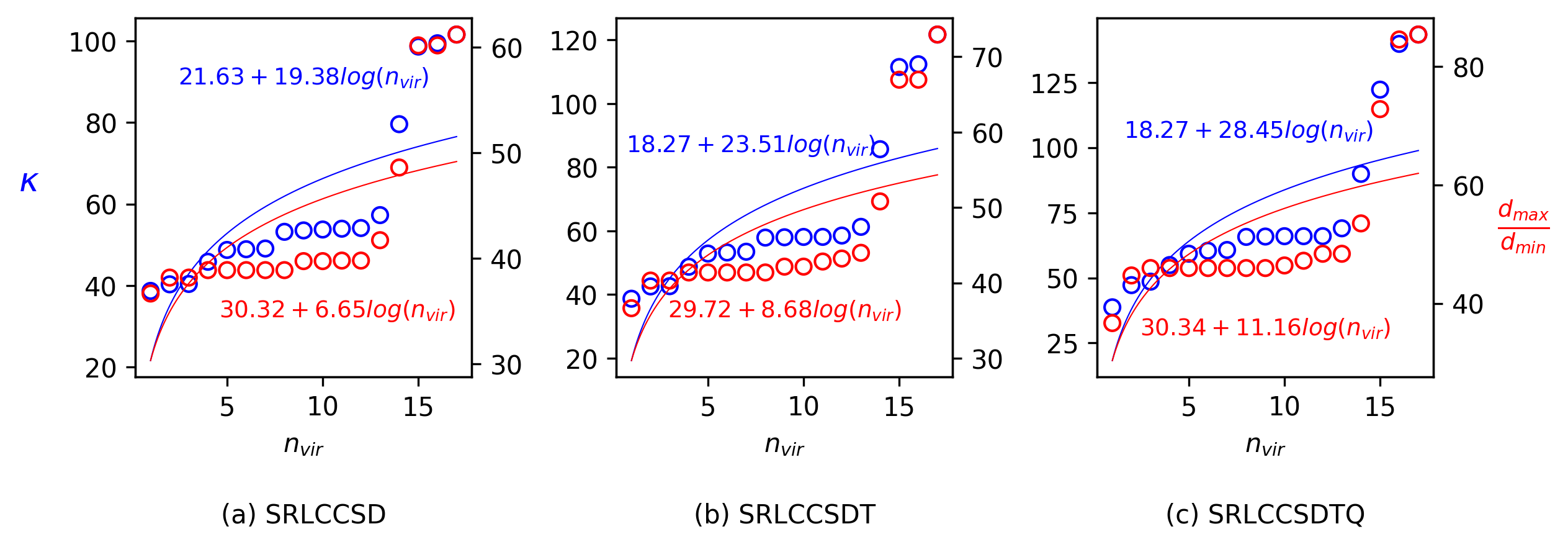}  
\caption{This figure shows the scaling of the condition number, $\kappa$, and the ratio of the maximum and minimum diagonal entries of $A$, $d_{\mathrm{max}}\big/d_{\mathrm{min}}$, with the number of virtuals, $n_{vir}$. Here, we considered the LiH molecule in the cc-pVDZ basis with excitation ranks 2 (Subfigure (a)), 3 (Subfigure (b)) and 4 (Subfigure (c)), all for the SRLCC case. } \label{overall_excrank} 
\end{figure}  

\begin{figure}[!h] 
\centering
   \includegraphics[scale=0.65]{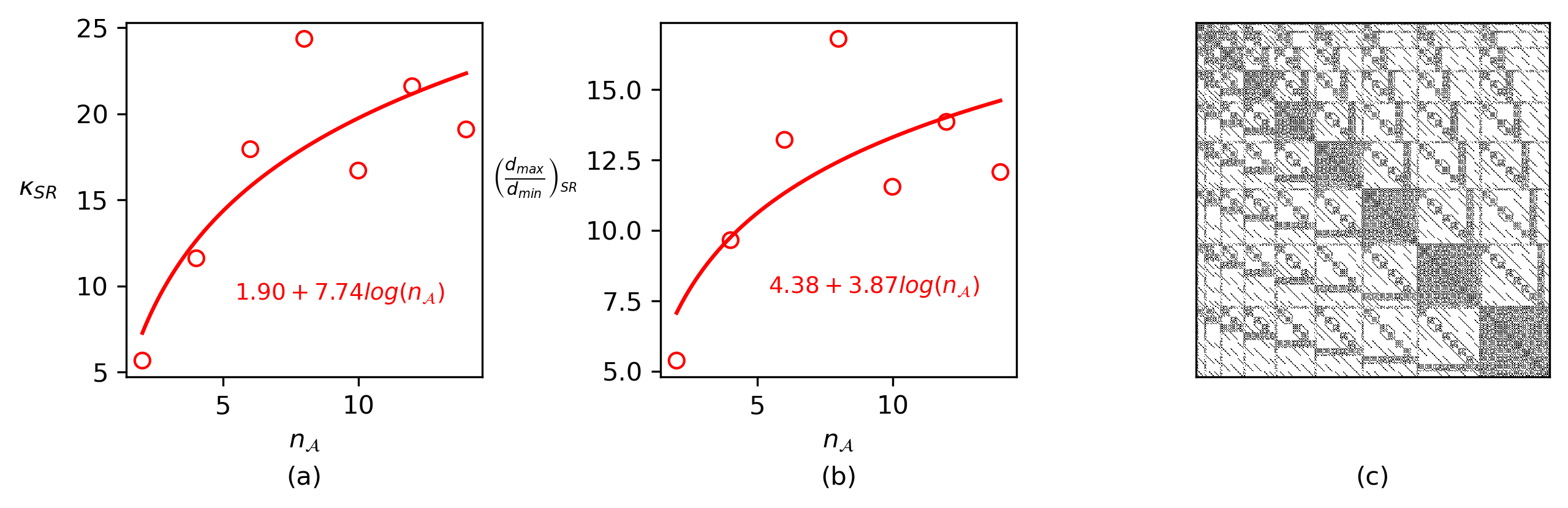}  
\caption{The figure depicts the following: (a) $\kappa$ versus $n_\mathcal{A}$, (b) $d_{\mathrm{max}}\big/d_{\mathrm{min}}$ versus number of atoms, $n_{\mathcal{A}}$, and (c) the $A$ matrix heatmap at $n_{\mathcal{A}}=6$ ($A$ matrix size of $405\times405$), all for $H_6$ chain picked as a representative example and in the SRLCCSD level of theory.} \label{overallchain:h} 
\end{figure} 

\begin{figure}[!h] 
\centering
   \includegraphics[scale=0.65]{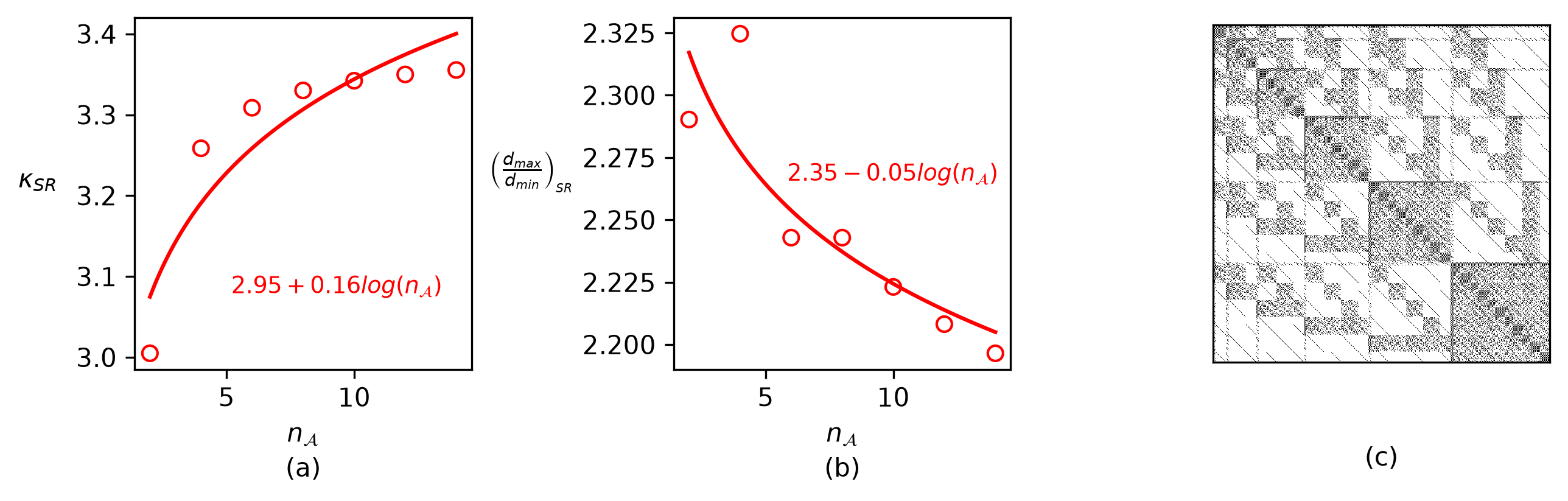}  
\caption{Figure illustrating the behaviour of the same quantities as in the earlier one and for the SRLCCSD case, this time for the $He_6$ chain ($A$ matrix size of $702\times 702$) chosen for the representative example.  } \label{overallchain:he} 
\end{figure} 

\begin{figure}[!h] 
\centering
   \includegraphics[scale=0.65]{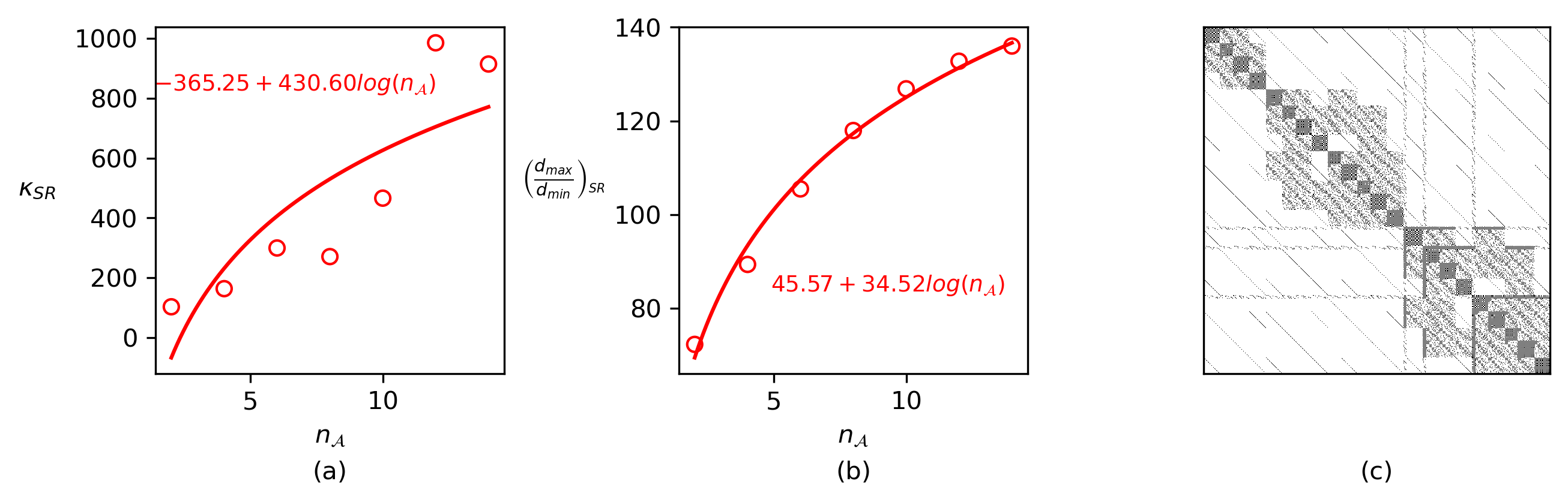}  
\caption{The same data as the earlier two figures and for the SRLCCSD case, but for the $Li_6$ chain ($A$ matrix size of $954\times 954$). The notations are the same as in the earlier two figures. } \label{overallchain:Li} 
\end{figure} 

\begin{figure}[!h] 
\centering
   \includegraphics[scale=0.55]{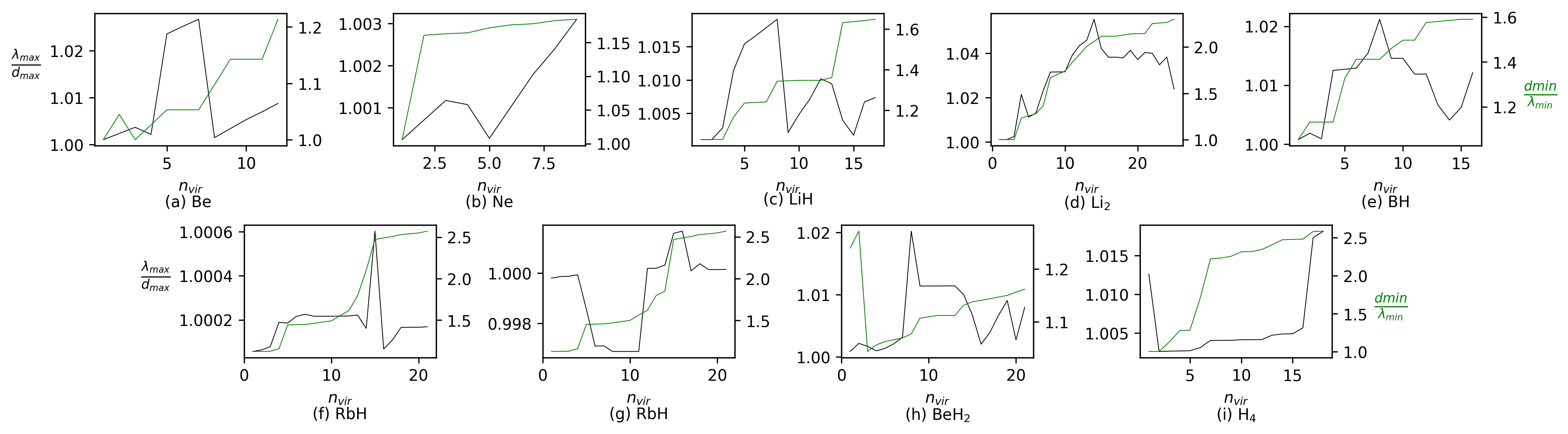}  
\caption{This figure depicts the ratios $\frac{\lambda_{\max}}{d_{\max}}$ and $\frac{d_{\min}}{\lambda_{\min}}$.  Every molecule is taken in the cc-pVDZ basis, except RbH which is taken in the Sapporo DKH3 DZP 2012 (Subfigure (f)) and Sapporo DZP 2012 (Subfigure (g)) bases. All of the data is for the SRLCCSD scenario. } \label{lambdamaxbydmaxdminbylambdamin_ccpvdz} 
\end{figure} 

\begin{figure}[!h] 
\centering
   \includegraphics[scale=0.5]{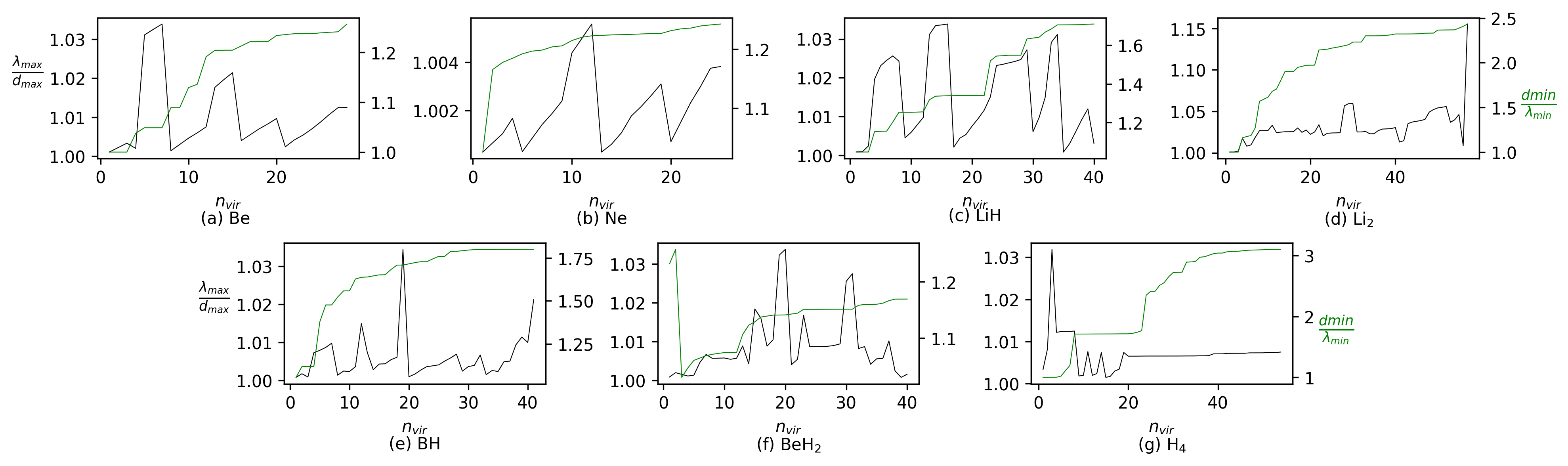}  
\caption{This figure depicts the ratios $\frac{\lambda_{\max}}{d_{\max}}$ and $\frac{d_{\min}}{\lambda_{\min}}$. Every molecule is taken in the cc-pVTZ basis. All of the data is for the SRLCCSD scenario.  } \label{lambdamaxbydmaxdminbylambdamin_ccpvtz} 
\end{figure} 

\begin{figure}[!h] 
\centering
   \includegraphics[scale=0.5]{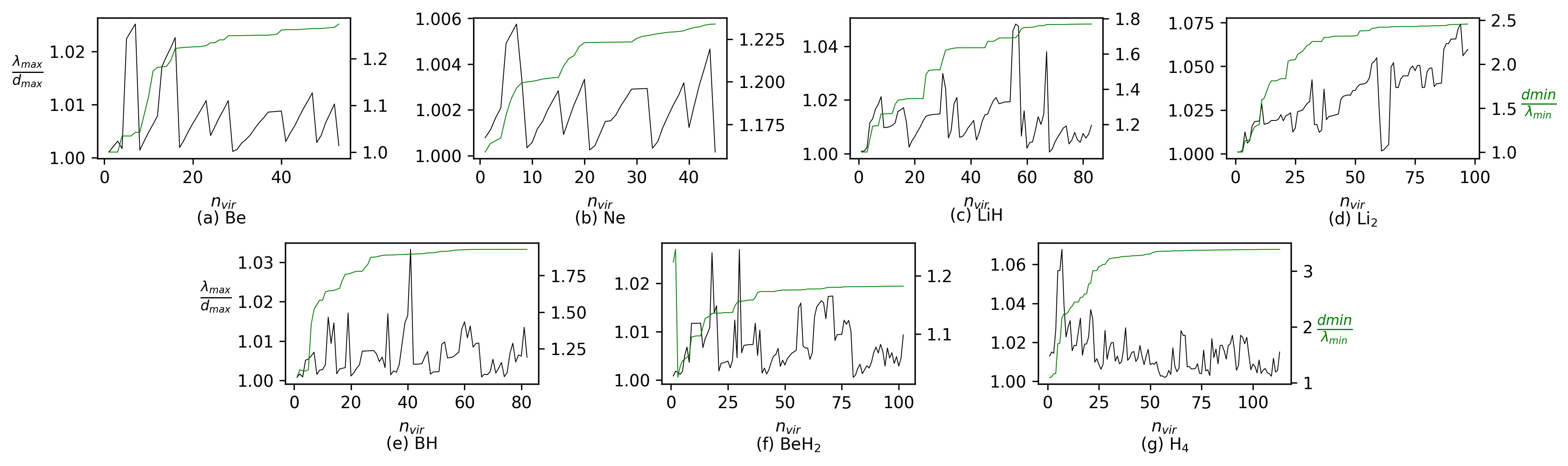}  
\caption{This figure depicts the ratios $\frac{\lambda_{\max}}{d_{\max}}$ and $\frac{d_{\min}}{\lambda_{\min}}$. Every molecule is taken in the cc-pVQZ basis. All of the data is for the SRLCCSD scenario. } \label{lambdamaxbydmaxdminbylambdamin_ccpvqz} 
\end{figure}

\begin{figure}[!h] 
\centering
   \includegraphics[scale=0.55]{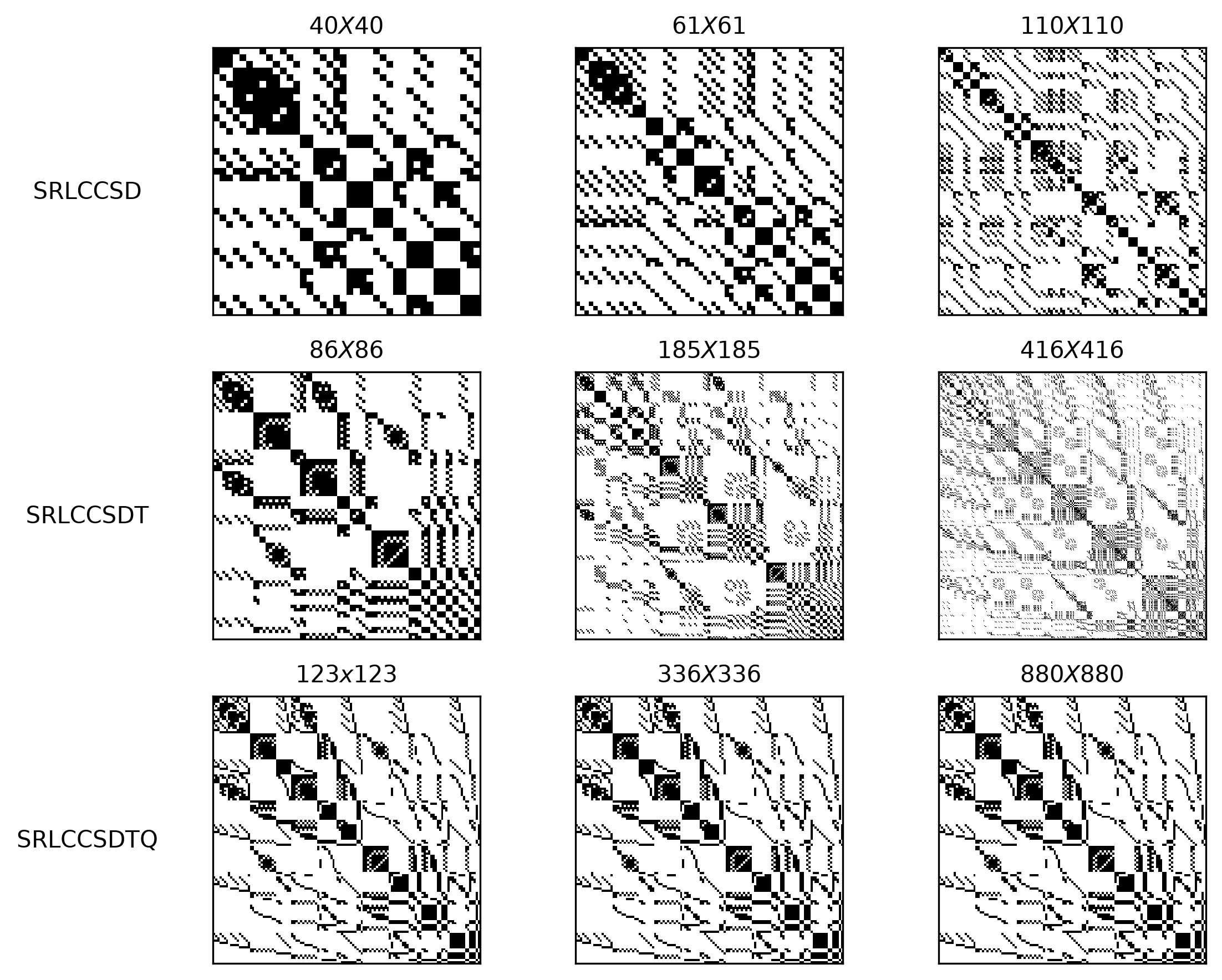}  
\caption{This figure depicts heatmaps of the Be atom, in cc-pVQZ basis for three different matrix sizes (mentioned at the top of every heatmap) and for  SRLCCSD, SRLCCSDT, and SRLCCSDTQ.} \label{heatmapsr:be} 
\end{figure} 

\begin{figure}[!h] 
\centering
   \includegraphics[scale=0.55]{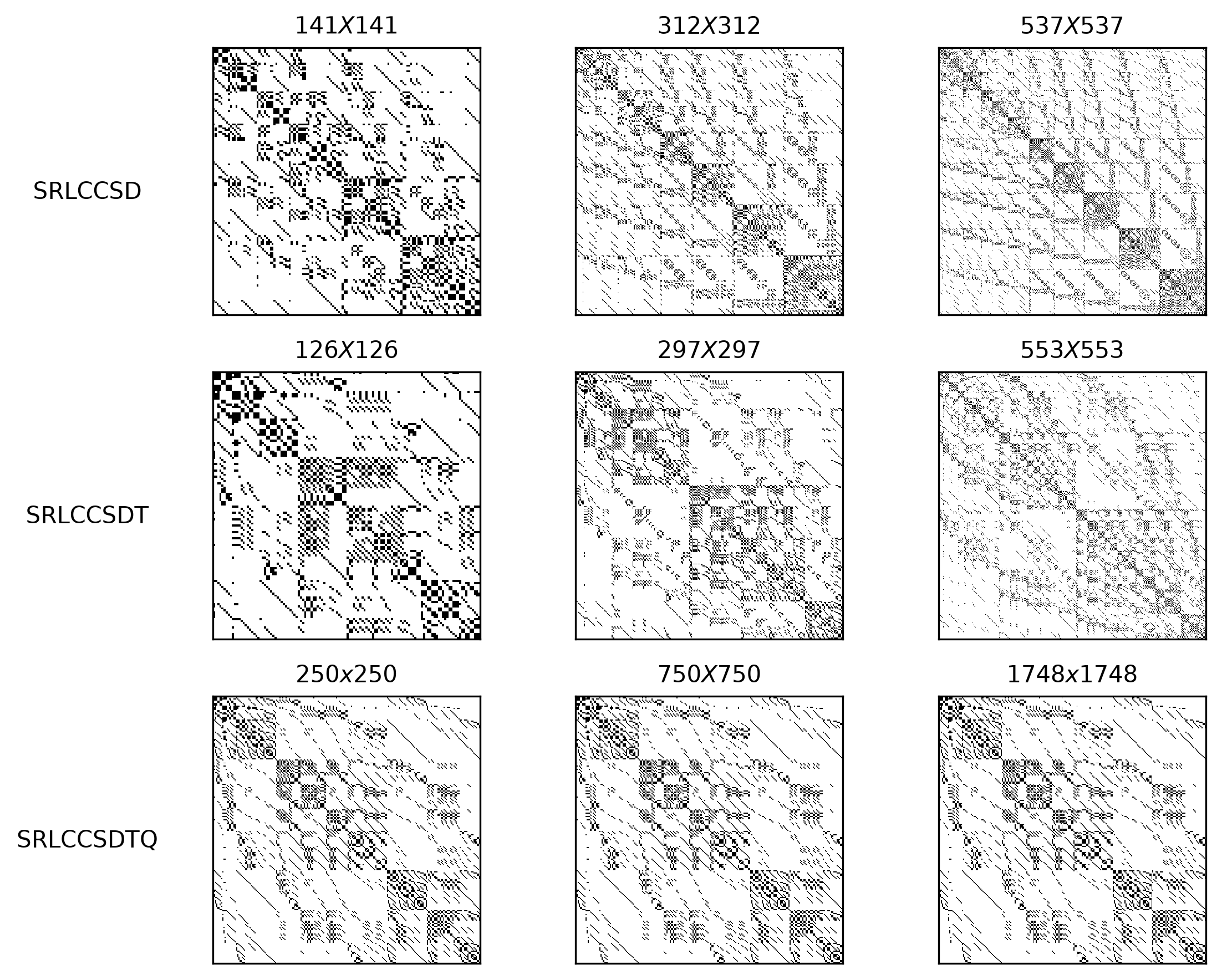}  
\caption{This figure depicts heatmaps of the Ne atom, in cc-pVQZ basis for three different matrix sizes (mentioned at the top of every heatmap) and for  SRLCCSD, SRLCCSDT, and SRLCCSDTQ.} \label{heatmapsr:Ne} 
\end{figure} 

\begin{figure}[!h] 
\centering
   \includegraphics[scale=0.55]{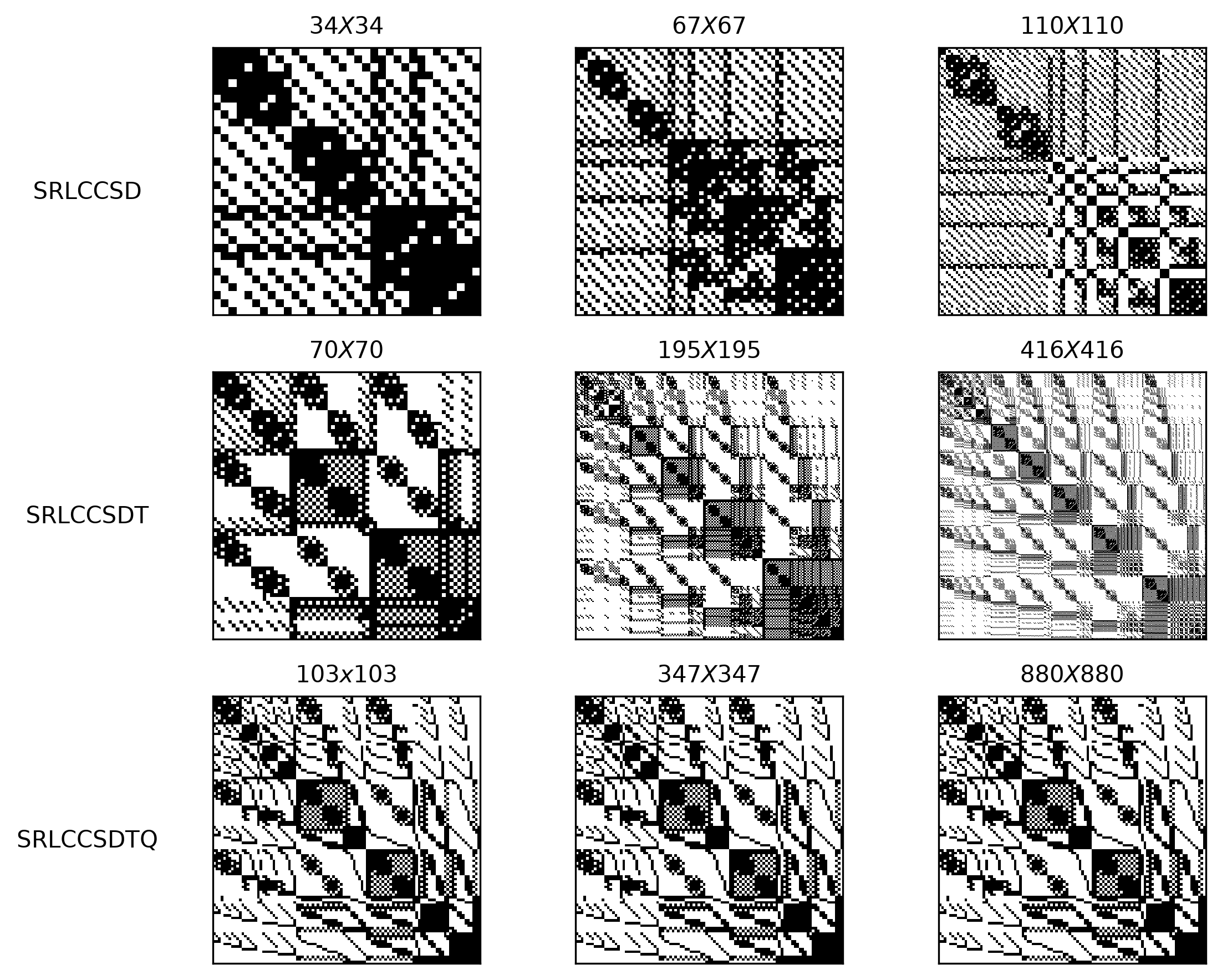}  
\caption{This figure depicts heatmaps of the LiH molecule, in cc-pVTZ basis for three different matrix sizes (mentioned at the top of every heatmap) and for  SRLCCSD, SRLCCSDT, and SRLCCSDTQ.} \label{heatmapsr:lih} 
\end{figure} 

\begin{figure}[!h] 
\centering
   \includegraphics[scale=0.55]{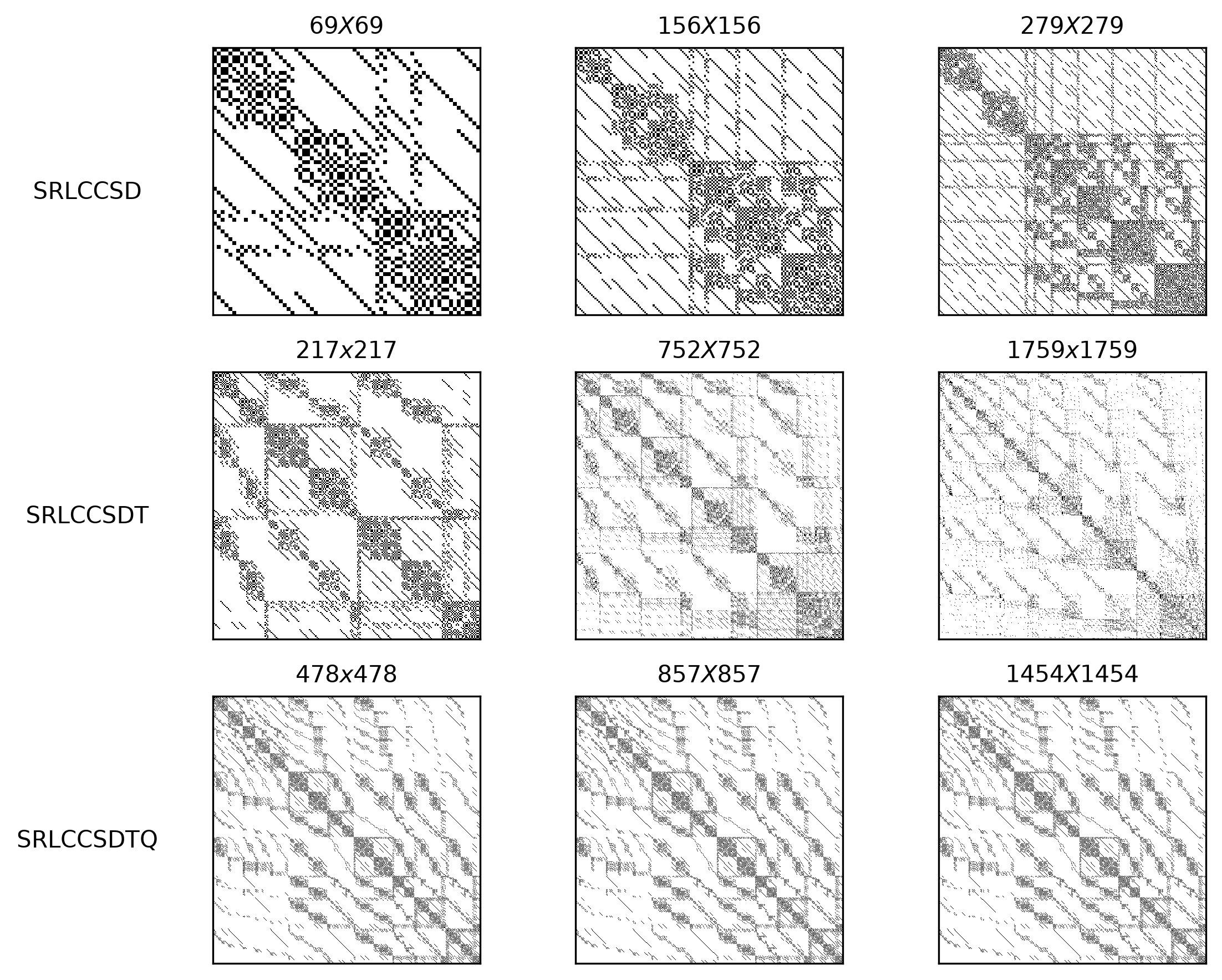}  
\caption{This figure depicts heatmaps of the Li$_2$ molecule, in cc-pVTZ basis for three different matrix sizes (mentioned at the top of every heatmap) and for  SRLCCSD, SRLCCSDT, and SRLCCSDTQ.} \label{heatmapsr:li2} 
\end{figure} 

\begin{figure}[!h] 
\centering
   \includegraphics[scale=0.55]{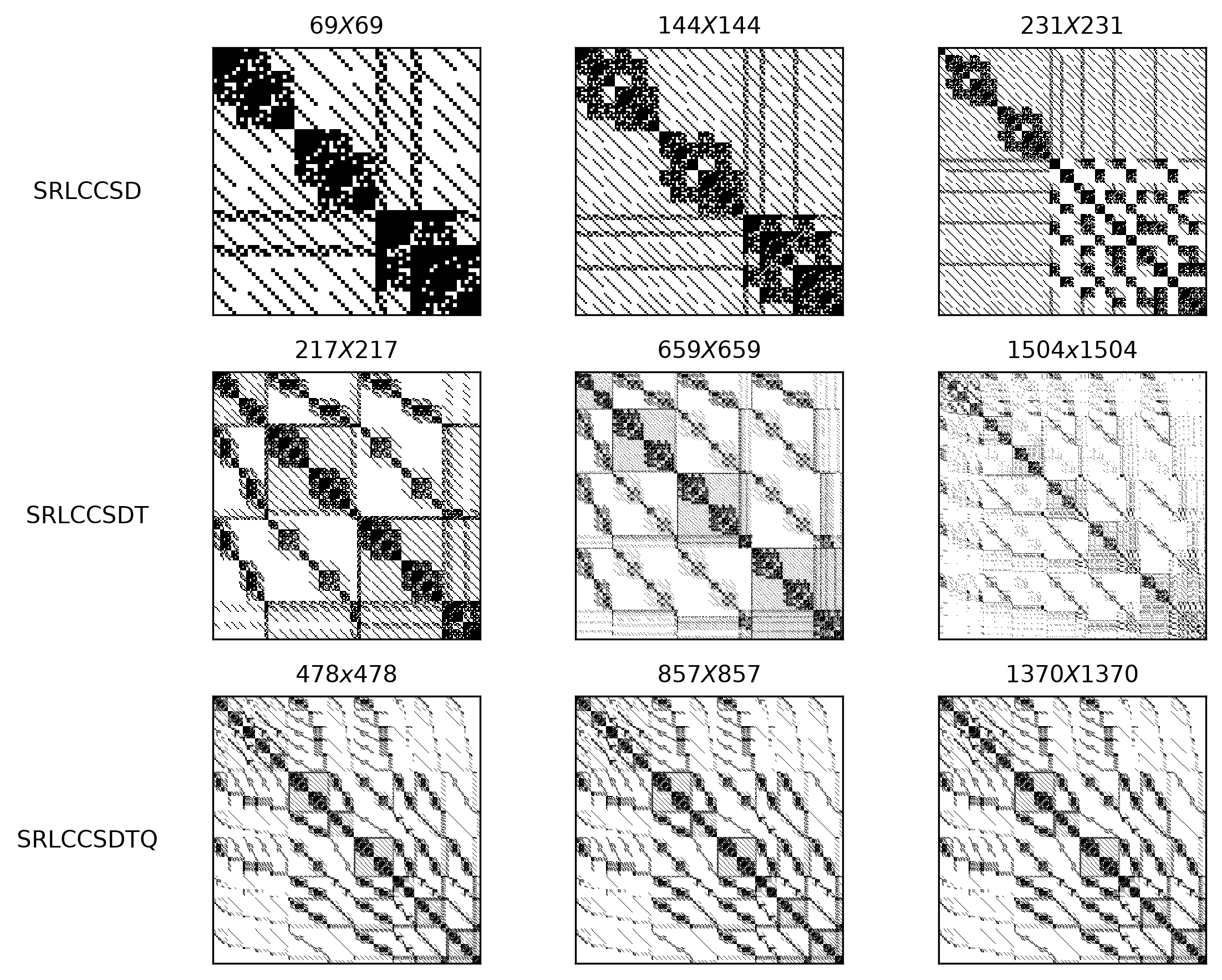}  
\caption{This figure depicts heatmaps of the BH molecule, in cc-pVTZ basis for three different matrix sizes (mentioned at the top of every heatmap) and for  SRLCCSD, SRLCCSDT, and SRLCCSDTQ.} \label{heatmapsr:bh} 
\end{figure}  

\begin{figure}[!h] 
\centering
   \includegraphics[scale=0.55]{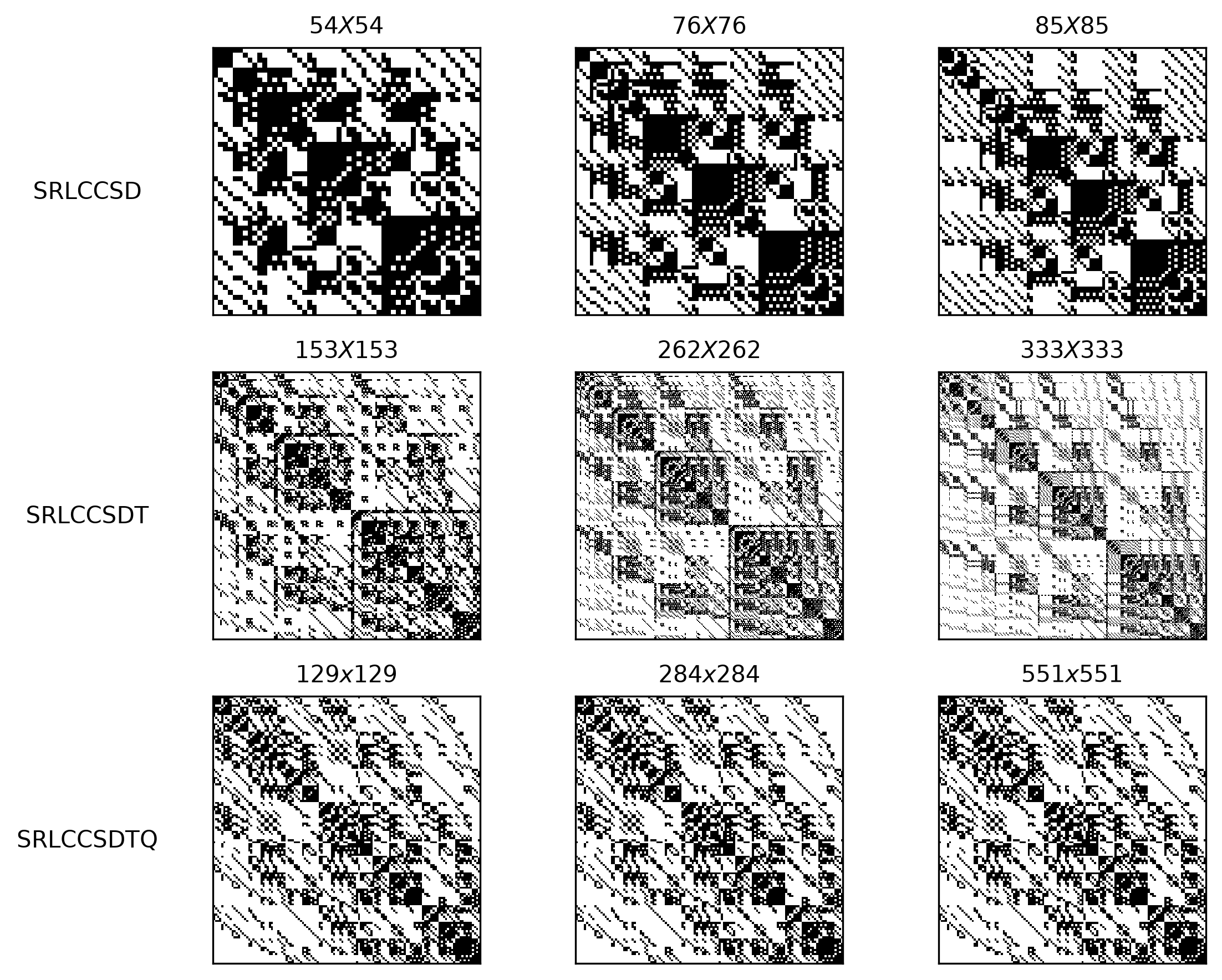}  
\caption{This figure depicts heatmaps of the RbH molecule in Sapporo DKH3 DZP 2012 basis for three different matrix sizes (mentioned at the top of every heatmap) and for  SRLCCSD, SRLCCSDT, and SRLCCSDTQ.} \label{heatmapsr:rbhsapporodkh3dzp2012} 
\end{figure}  

\begin{figure}[!h] 
\centering
   \includegraphics[scale=0.55]{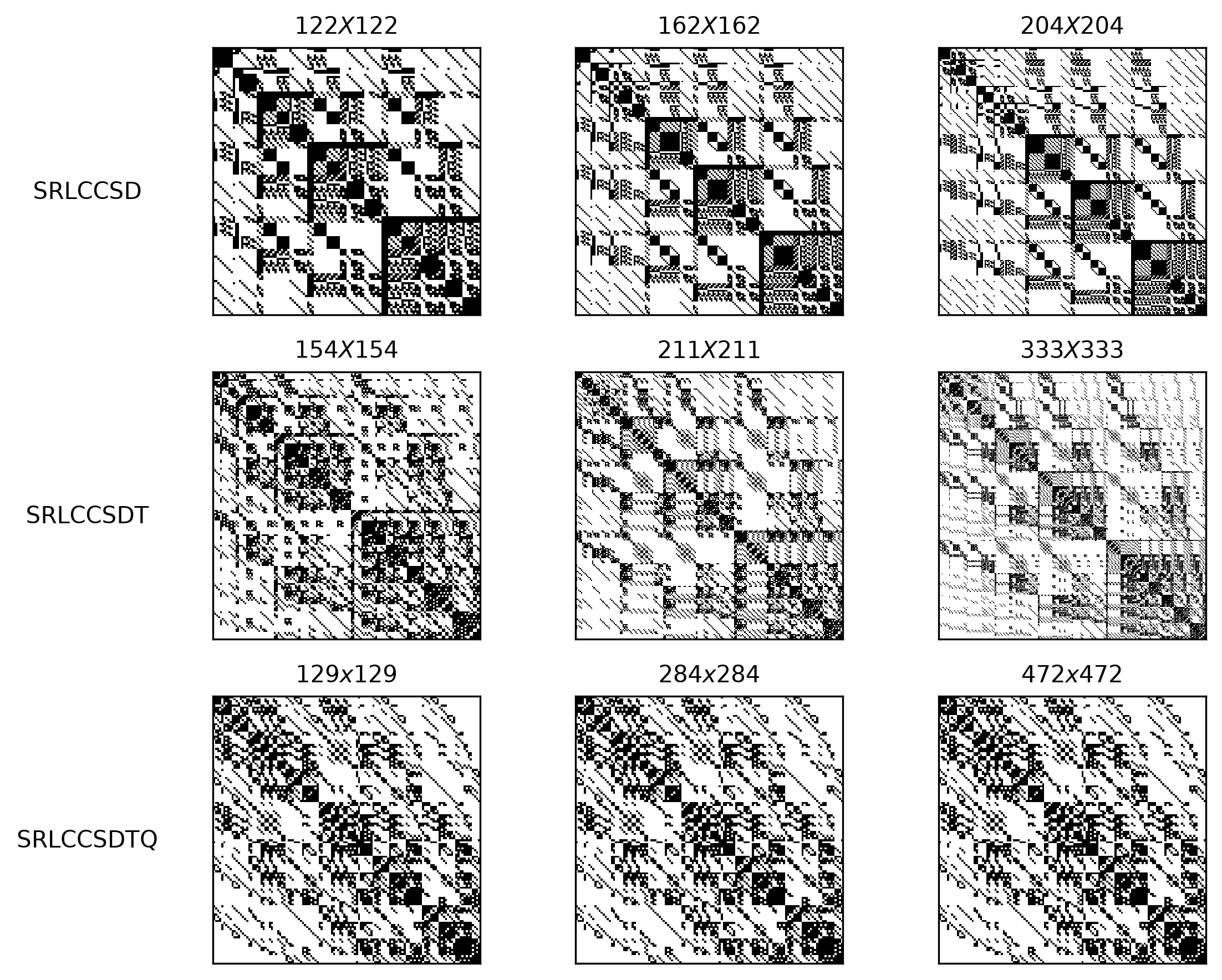}  
\caption{This figure depicts heatmaps of the RbH molecule in Sapporo DZP 2012 basis for three different matrix sizes (mentioned at the top of every heatmap) and for  SRLCCSD, SRLCCSDT, and SRLCCSDTQ.} \label{heatmapsr:rbhsapporodzp2012} 
\end{figure}  

\begin{figure}[!h] 
\centering
   \includegraphics[scale=0.55]{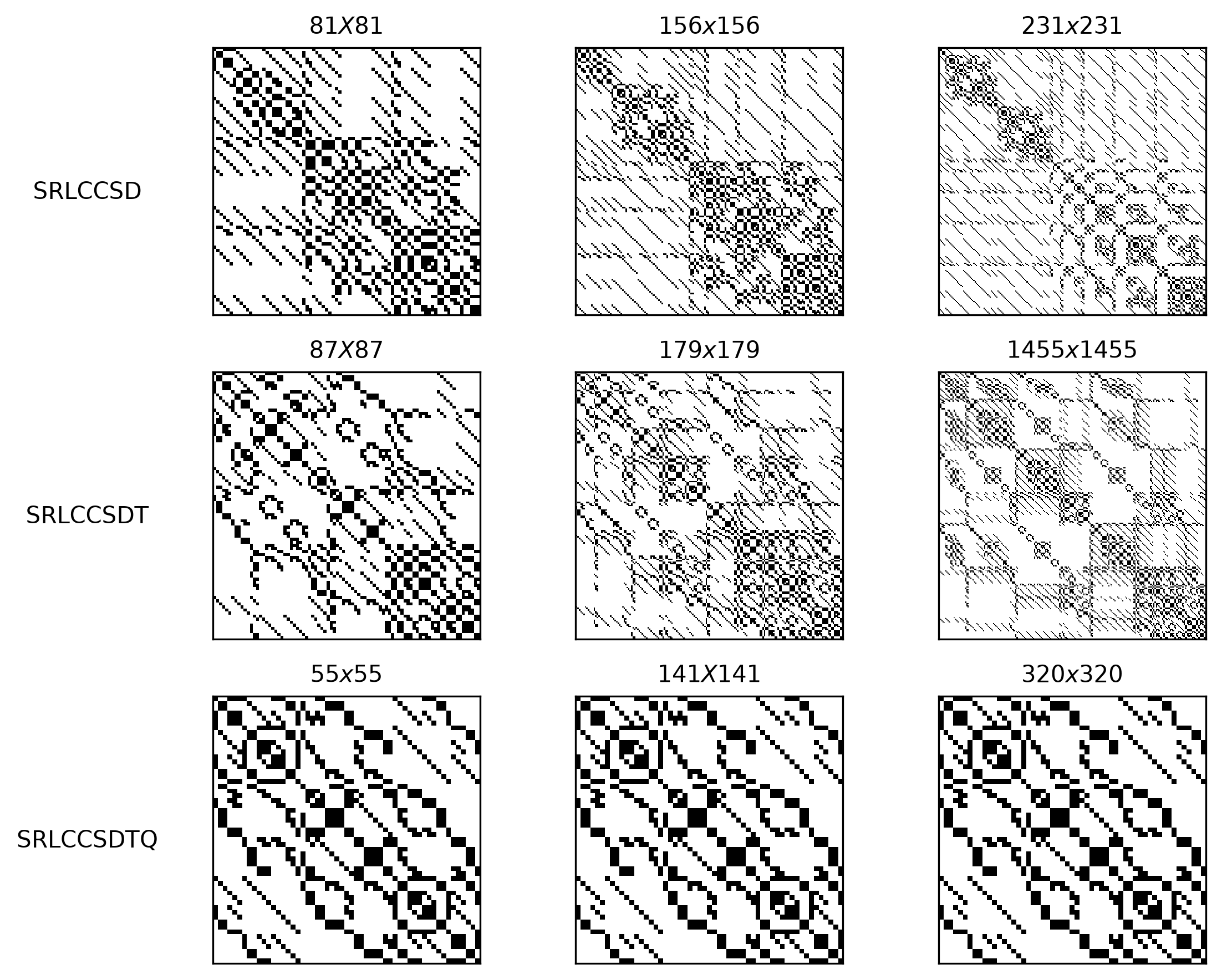}  
\caption{This figure depicts heatmaps of the BeH$_2$ molecule, in cc-pVTZ basis for three different matrix sizes (mentioned at the top of every heatmap) and for  SRLCCSD, SRLCCSDT, and SRLCCSDTQ.} \label{heatmapsr:beh2} 
\end{figure} 

\begin{figure}[!h] 
\centering
   \includegraphics[scale=0.55]{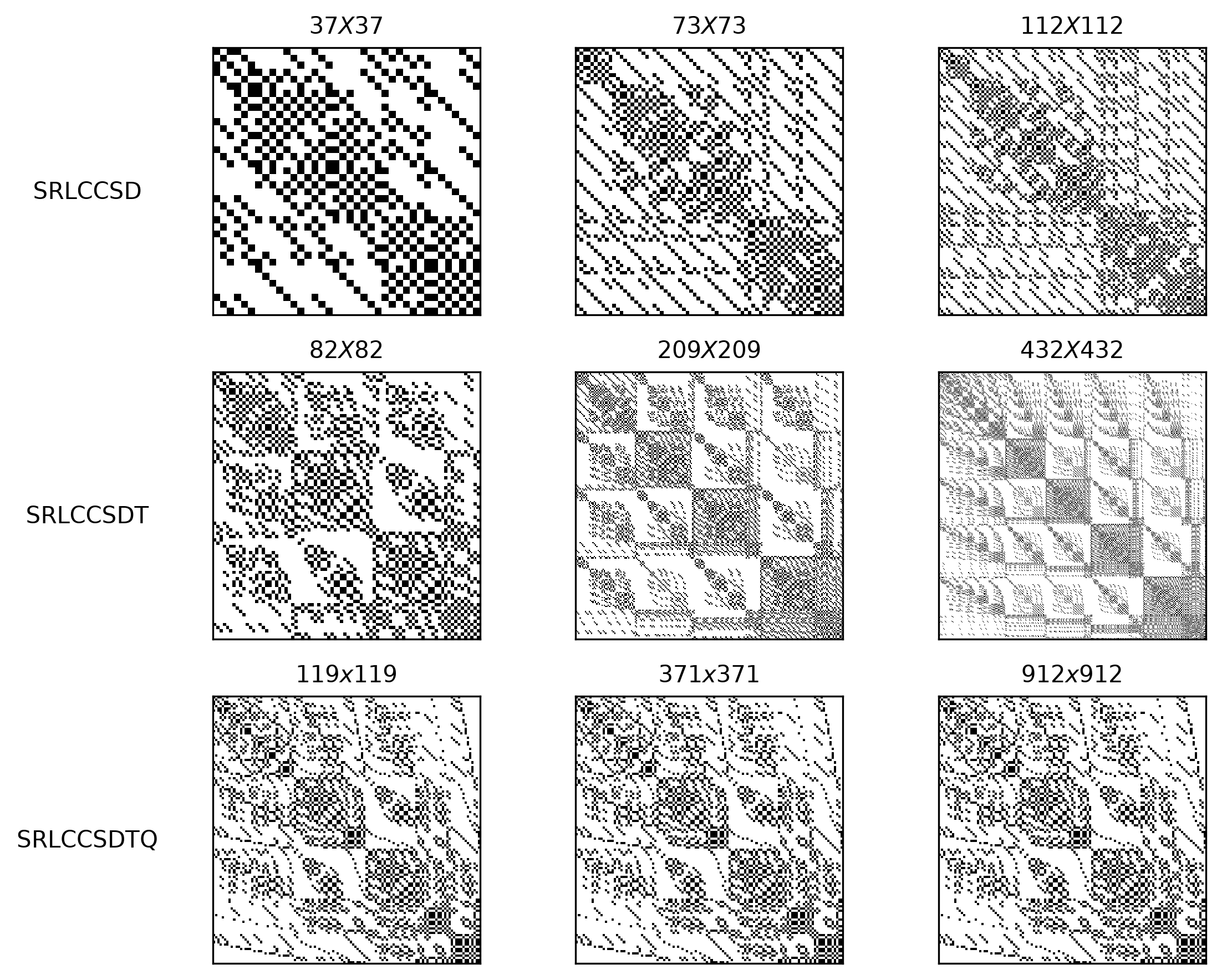}  
\caption{This figure depicts heatmaps of the H$_4$ molecule, in cc-pVTZ basis for three different matrix sizes (mentioned at the top of every heatmap) and for  SRLCCSD, SRLCCSDT, and SRLCCSDTQ.} \label{heatmapsr:h4} 
\end{figure} 


\begin{figure}[!h] 
\centering
   \includegraphics[scale=0.55]{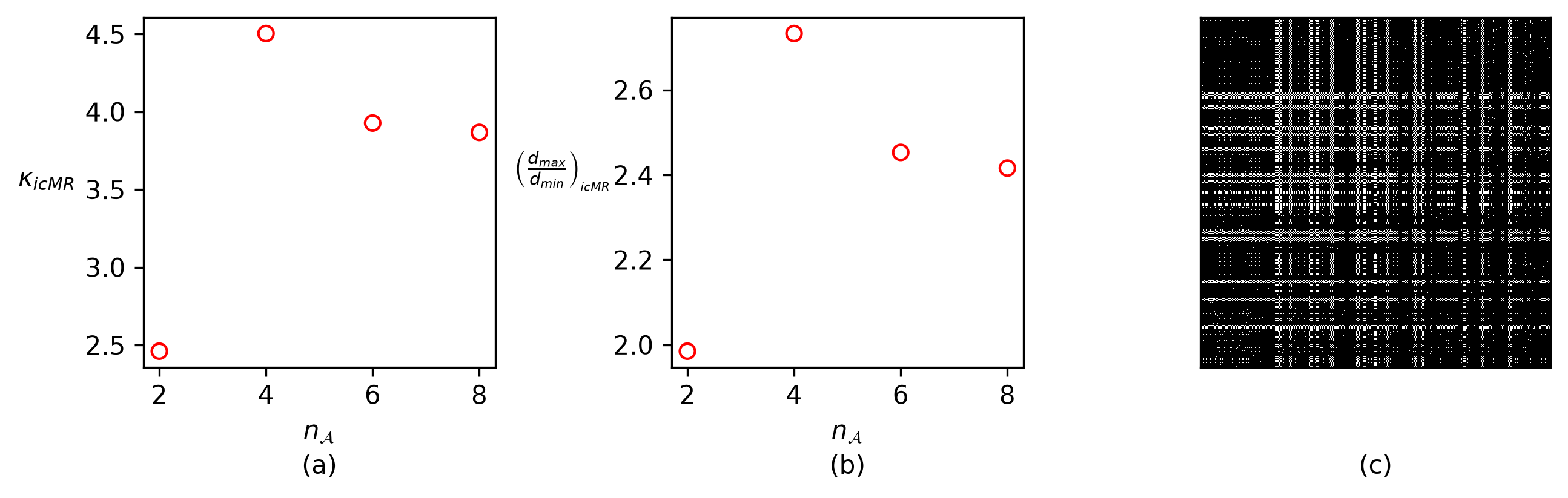}  
\caption{The figure depicts the following: (a) $\kappa$ versus $n_\mathcal{A}$, (b) $d_{\mathrm{max}}\big/d_{\mathrm{min}}$ versus number of atoms, $n_{\mathcal{A}}$, and (c) the $A$ matrix heatmap at $n_{\mathcal{A}}=6$ ($A$ matrix size of $2950\times2950$), all for $H_6$ chain picked as a representative example and in the icMRLCCSD level of theory. } \label{overall_plot:hchain-icmr} 
\end{figure} 

\begin{figure}[!h] 
\centering
   \includegraphics[scale=0.55]{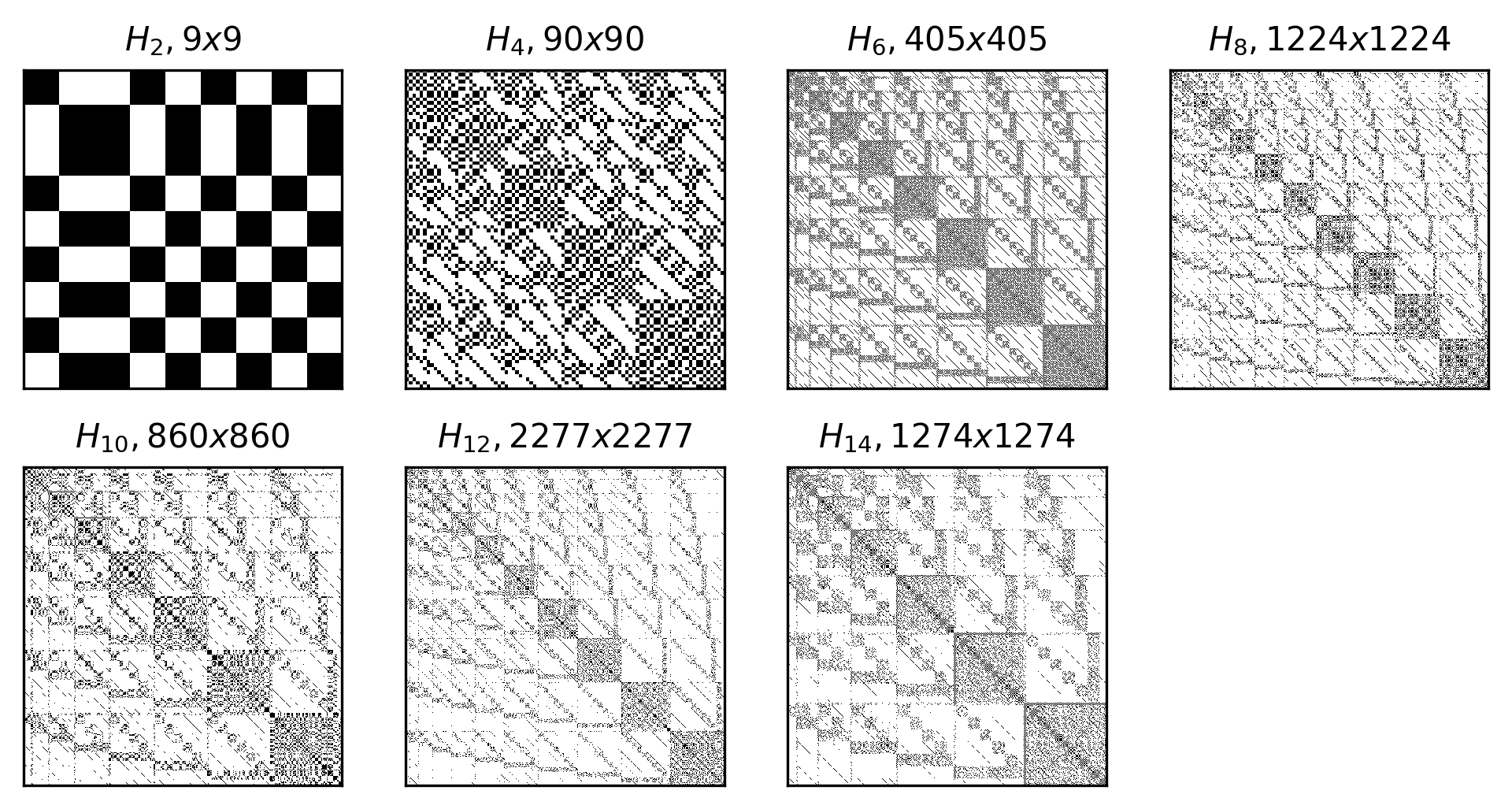}  
\caption{This figure depicts heatmaps of Hydrogen chains going from H$_2$ to H$_{14}$ in the case of SRLCCSD. In view of computational cost, we do not make the plots by considering full spaces, but rather a subset of virtuals again picked at random. } \label{heatmap:h} 
\end{figure} 

\begin{figure}[!h] 
\centering
   \includegraphics[scale=0.55]{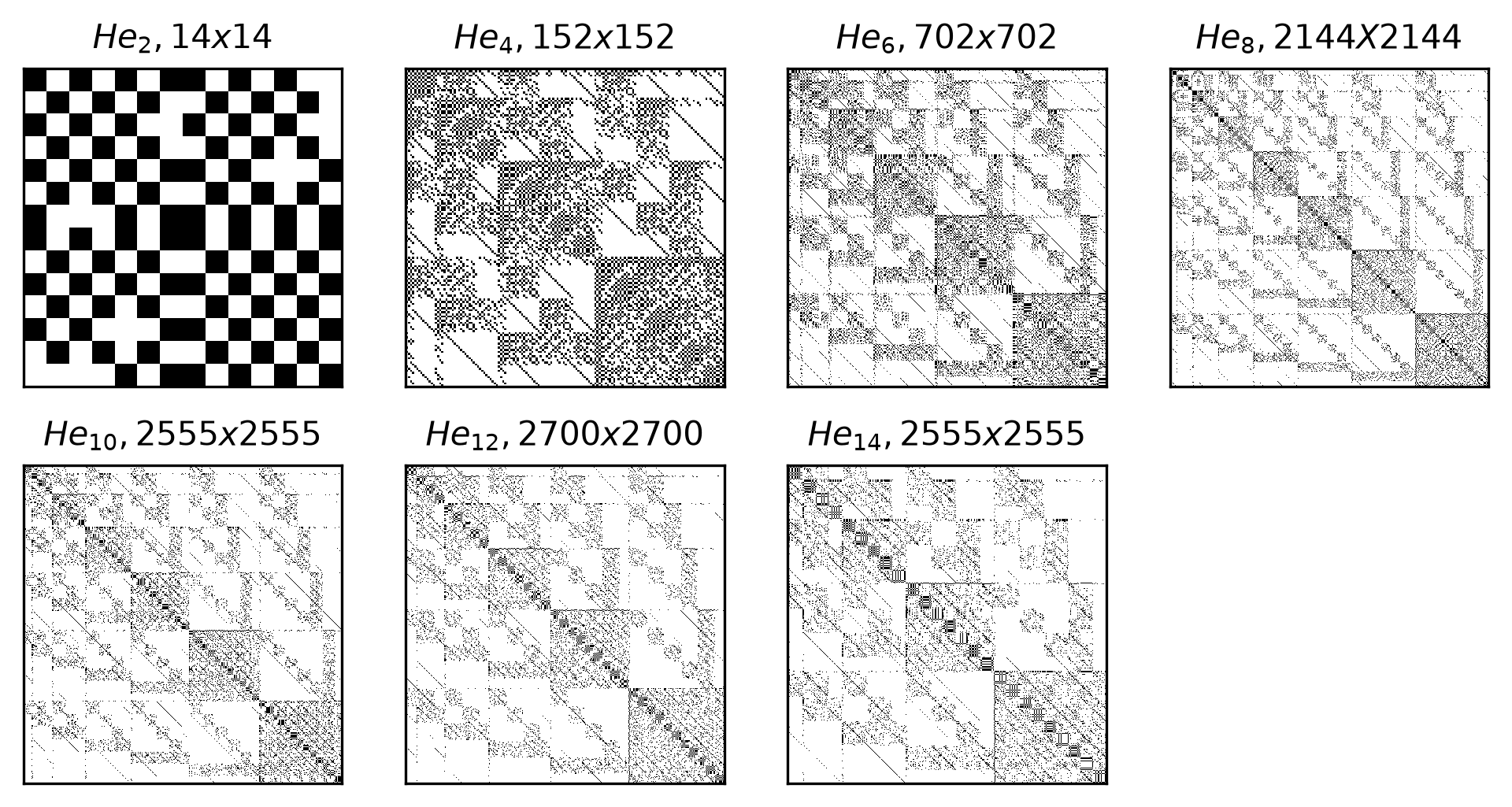}  
\caption{This figure depicts heatmaps of Helium chains going from He$_2$ to He$_{14}$ in the case of SRLCCSD.  In view of computational cost, we do not make the plots by considering full spaces, but rather a subset of virtuals again picked at random. } \label{heatmap:he} 
\end{figure}

\begin{figure}[!h] 
\centering
   \includegraphics[scale=0.55]{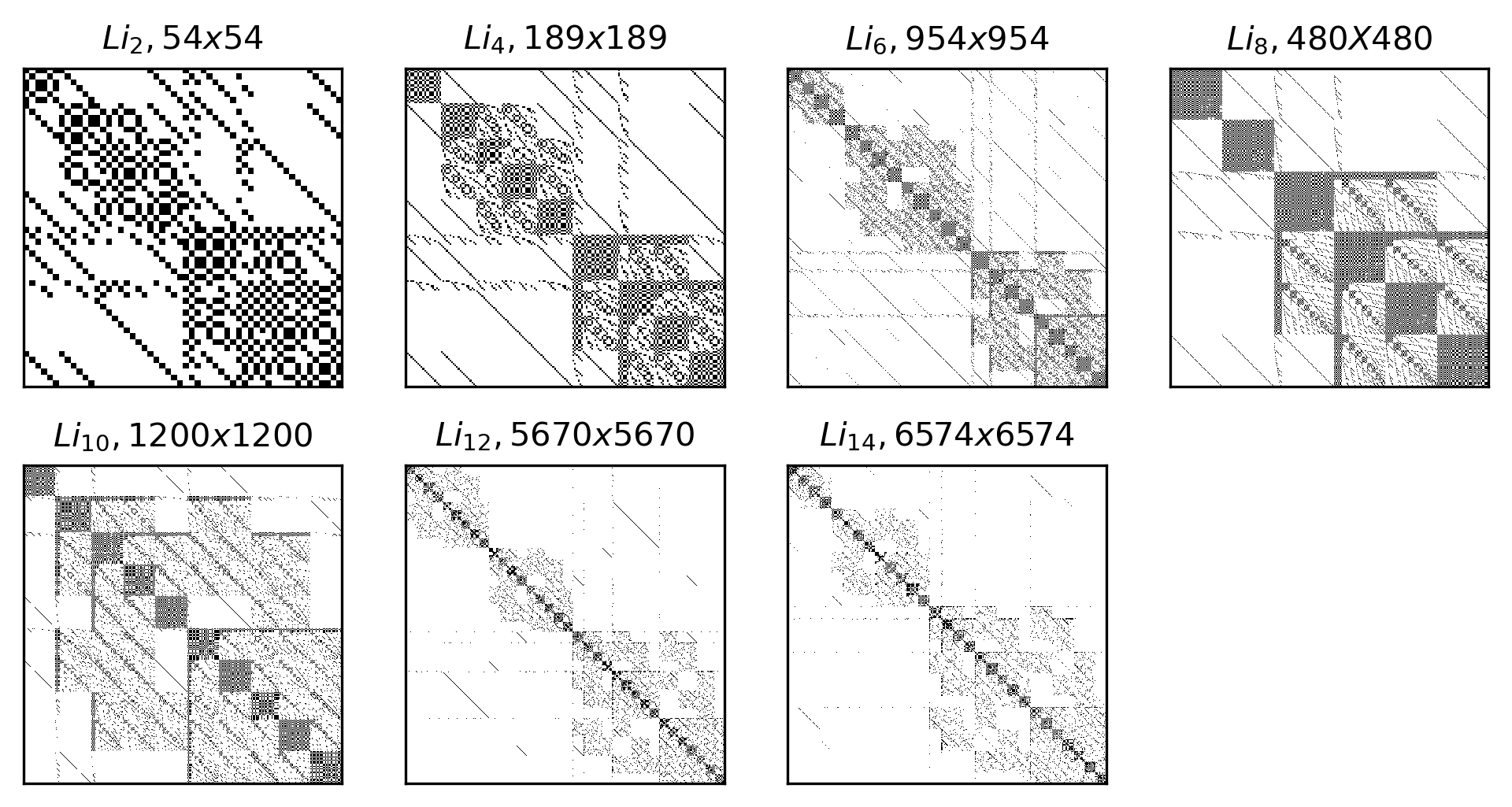}  
\caption{This figure depicts heatmaps of Lithium chains going from Li$_2$ to Li$_{14}$ in the case of  SRLCCSD.  In view of computational cost, we do not make the plots by considering full spaces, but rather a subset of virtuals again picked at random. } \label{heatmap:li} 
\end{figure} 

\clearpage 

\begin{figure}[!h] 
\centering
   \includegraphics[scale=0.5]{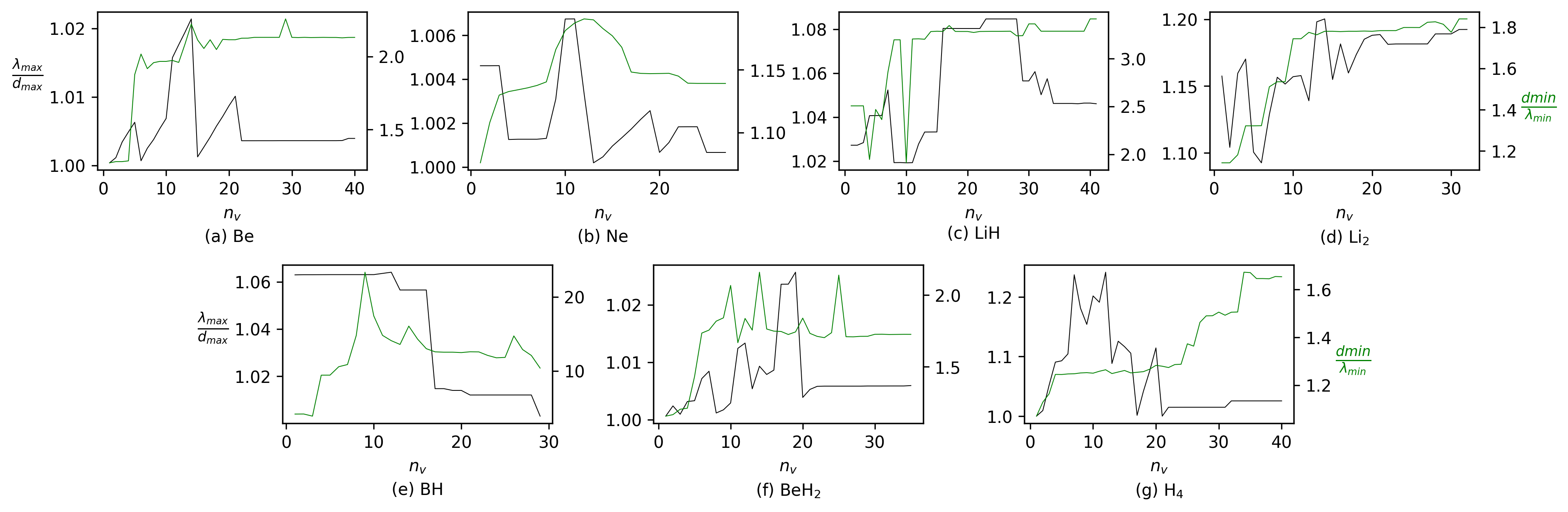}  
\caption{This figure depicts the ratios $\frac{\lambda_{\max}}{d_{\max}}$ and $\frac{d_{\min}}{\lambda_{\min}}$ for the icMRLCCSD case for Be, Ne, LiH, Li$_2$, BH, BeH$_2$ and H$_4$. Be and Ne are in the cc-pVQZ bases, whereas other molecules are in the cc-pVTZ basis.} \label{lambdamaxbydmaxdminbylambdamin_mr_ccpvtz} 
\end{figure} 

\begin{figure}[!h] 
\centering
   \includegraphics[scale=0.45]{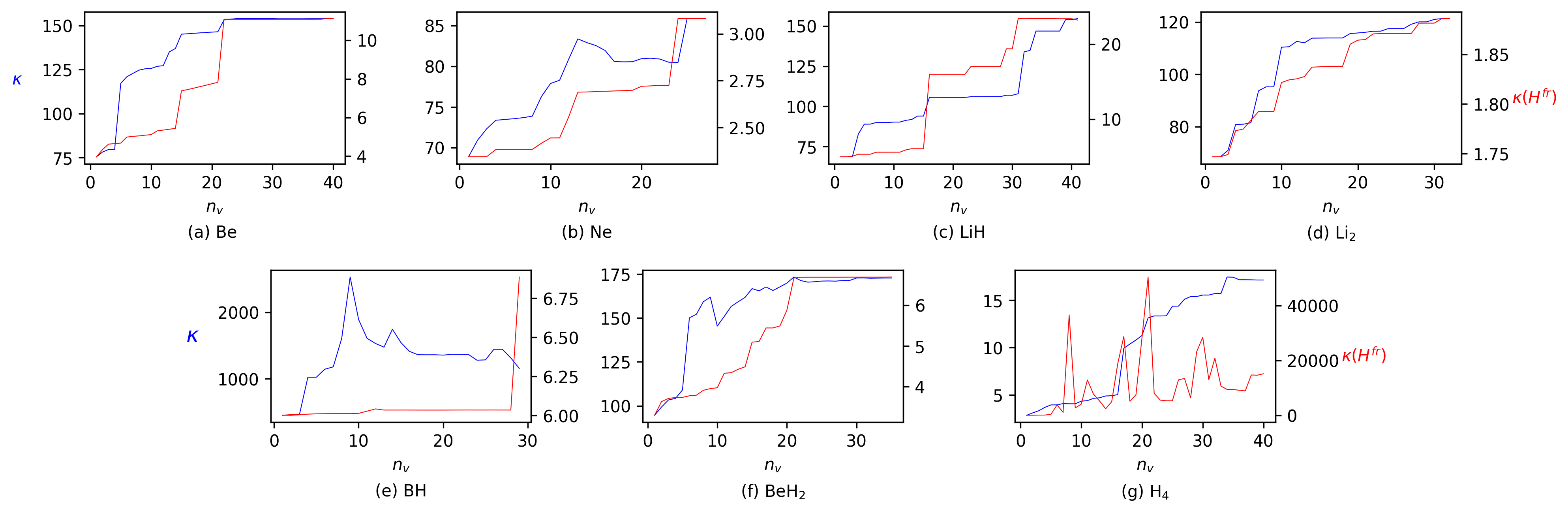}  
\caption{This figure depicts condition numbers of the $H^{\text{fr}}$, $\kappa(H^{\text{fr}})$, and condition number of $A$, $\kappa$ vs the number of virtual orbitals ($n_{v}$) and in the icMRLCCSD case for Be, Ne, LiH, Li$_2$, BH, BeH$_2$ and H$_4$. Be and Ne are in the cc-pVQZ bases, whereas other molecules are in the cc-pVTZ basis. } \label{kappaA_kappaHmr:allmolecules} 
\end{figure} 

\begin{figure}[!h] 
\centering
   \includegraphics[scale=0.45]{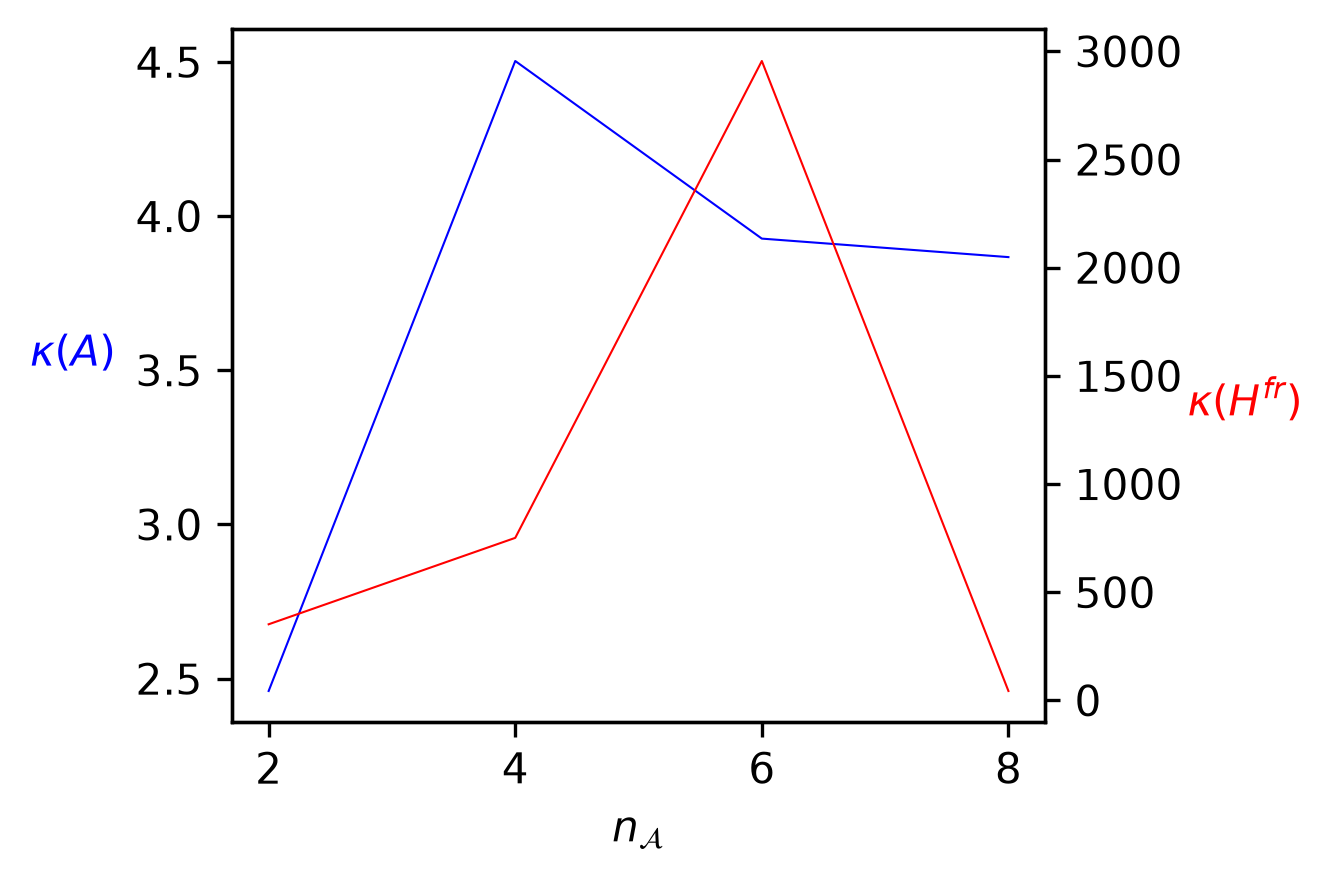}  
\caption{This figure depicts condition numbers of the $H^{\text{fr}}$, $\kappa(H^{\text{fr}})$, and condition number of $A$, $\kappa(A)$ vs the number of atoms ($n_{A}$) and in the icMRLCCSD case for Hydrogen linear chain. } \label{kappaA_kappaHmr:h_chain} 
\end{figure}

\begin{figure}[!h] 
\centering
   \includegraphics[scale=0.55]{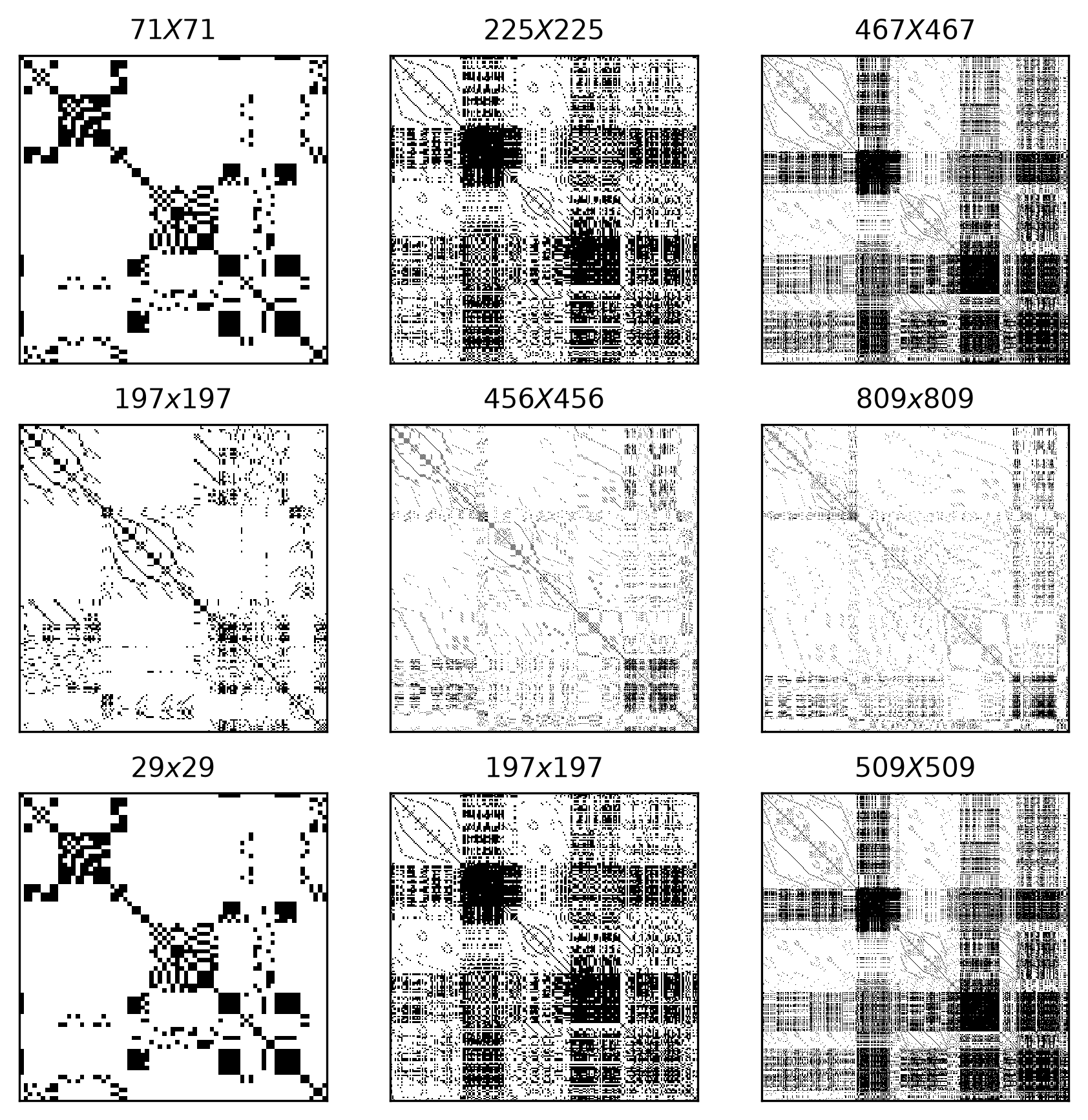}  
\caption{This figure depicts heatmaps of the Be atom (first row), Ne atom (second row) and LiH molecule (third row), for three different matrix sizes (mentioned at the top of every heatmap) and in the case of icMRLCCSD. Be and Ne are in the cc-pVQZ basis set, while LiH is in the cc-pVTZ basis set.} \label{heatmapmr:be_ne_lih} 
\end{figure}

\begin{figure}[!h] 
\centering
   \includegraphics[scale=0.55]{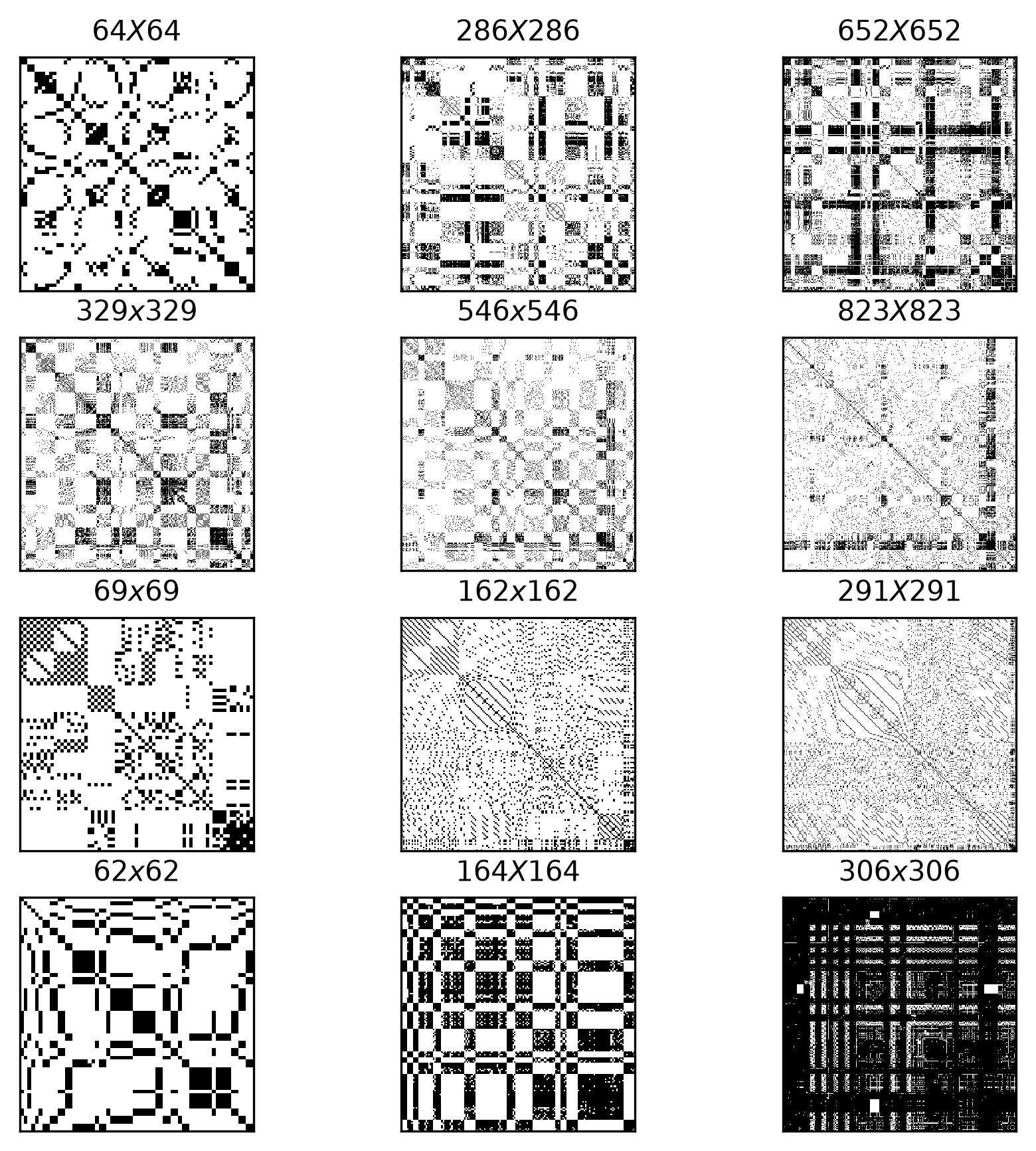}  
\caption{This figure depicts heatmaps of the Li$_2$  (first row), BH (second row), BeH$_2$ (third row) and H$_4$ molecule (fourth row), each in the cc-pVTZ basis set and for three different matrix sizes (mentioned at the top of every heatmap) and for icMRLCCSD.} \label{heatmapmr:li2_bh_h4_beh2} 
\end{figure} 

\begin{table*}[t]
\centering
\caption{Details of chains considered for our SRLCC and icMRLCC
calculations. The notations are the same as in Table \ref{tab:systems} of the main manuscript. } 
\label{tab:systems-chains}
\renewcommand{\arraystretch}{1.25}
\setlength{\tabcolsep}{5pt}
\begin{tabular}{p{0.10\textwidth}
                p{0.42\textwidth}
                p{0.38\textwidth}}
\hline
\textbf{System} & \textbf{SRLCC} & \textbf{icMRLCC} \\
\hline
H chain &
6-31G;
linear geometry with nearest  neighbour
separation of $1.400$. 
&  6-31G;
linear geometry with nearest neighbour
separation of $3.400$. Choice of active spaces from H$_2$ through H$_8$: (2,2), (4,4), (6,6), (8,8) respectively. \\
He chain &
6-31G;
linear geometry with uniform nearest  neighbour
separation of $5.669$
& - \\
Li chain &
6-31G;
linear geometry with uniform nearest neighbour
separation of $5.051$
& - \\
\hline
\end{tabular}
\end{table*}

\begin{table*}[t] 
\centering
\small
\caption{Matrix dimension $N$, condition number $\kappa$, and their amount of increase over the sampled data for each of the molecules considered in this study, for the SRLCC case. Subscripts $i$ and $f$ denote the initial and final values of a considered quantity. We note that the data on chains span 2-14 atoms. We recall that all the $A$ matrices were found to be positive definite, and thus $\kappa$ is simply the ratio of the maximum to minimum eigenvalues. In the table, $\kappa$ is rounded off to one decimal place, and the ratios in the last two columns are rounded off to the nearest whole number. } \label{tab:srlcc-N-kappa-data} 
\setlength{\tabcolsep}{12pt}
\begin{tabular}{llrrrr}
\toprule
Molecule/atom & Basis set & $N$ range & $\kappa$ range
& $N_f/N_i$ & $\kappa_f/\kappa_i$ \\
\midrule
\multicolumn{6}{c}{\textbf{LCCSD}} \\
\midrule
\multicolumn{6}{c}{Atoms} \\
\midrule
Be & cc-pVDZ
& $3$--$110$ & $32.5$--$73.1$ & $37$ & $2$ \\
Be & cc-pVTZ
& $5$--$506$ & $56.6$--$87.3$ & $101$ & $2$ \\
Be & cc-pVQZ
& $5$--$1667$ & $57.1$--$141.4$ & $333$ & $2$ \\
\hline 
Ne & cc-pVDZ
& $9$--$312$ & $34.0$--$44.2$ & $35$ & $1$ \\
Ne & cc-pVTZ
& $11$--$2168$ & $46.7$--$71.0$ & $197$ & $2$ \\
Ne & cc-pVQZ
& $147$--$8343$ & $67.6$--$187.5$ & $57$ & $3$ \\
\midrule
\multicolumn{6}{c}{Diatomic molecules} \\
\midrule
BH & cc-pVDZ
& $6$--$414$ & $63.4$--$123.6$ & $69$ & $2$ \\
BH & cc-pVTZ
& $6$--$2385$ & $59.0$--$163.2$ & $398$ & $3$ \\
BH & cc-pVQZ
& $6$--$8877$ & $55.3$--$263.0$ & $1480$ & $5$ \\
\hline 
Li$_2$ & cc-pVDZ
& $9$--$957$ & $72.5$--$202.0$ & $106$ & $3$ \\
Li$_2$ & cc-pVTZ
& $9$--$4461$ & $71.9$--$371.1$ & $496$ & $5$ \\
Li$_2$ & cc-pVQZ
& $9$--$12030$ & $71.0$--$265.5$ & $1337$ & $4$ \\
\hline 
RbH & Sapporo DZP 2012
& $86$--$21882$ & $9753.4$--$22287.6$ & $254$ & $2$ \\
RbH & Sapporo DKH3 DZP 2012
& $86$--$21882$ & $9586.7$--$22051.9$ & $254$ & $2$ \\
\hline 
LiH & cc-pVDZ
& $5$--$227$ & $38.7$--$101.6$ & $45$ & $3$ \\
LiH & cc-pVTZ
& $5$--$1046$ & $37.8$--$157.4$ & $209$ & $4$ \\
LiH & cc-pVQZ
& $5$--$4127$ & $36.8$--$257.3$ & $825$ & $7$ \\
\midrule
\multicolumn{6}{c}{Triatomic molecules} \\
\midrule
BeH$_2$ & cc-pVDZ
& $6$--$741$ & $19.3$--$43.8$ & $124$ & $2$ \\
BeH$_2$ & cc-pVTZ
& $6$--$2373$ & $18.5$--$54.3$ & $396$ & $3$ \\
BeH$_2$ & cc-pVQZ
& $6$--$13638$ & $17.8$--$89.2$ & $2273$ & $5$ \\
\midrule
\multicolumn{6}{c}{Tetratomic molecules} \\
\midrule
H$_4$ & cc-pVDZ
& $5$--$246$ & $4.1$--$130.9$ & $49$ & $32$ \\
H$_4$ & cc-pVTZ
& $5$--$1794$ & $2.3$--$299.6$ & $359$ & $130$ \\
H$_4$ & cc-pVQZ
& $3$--$7117$ & $7.4$--$478.7$ & $2372$ & 65 \\
\midrule
\multicolumn{6}{c}{Chains} \\
\midrule
H chain & 6-31G
& $9$--$1274$ & $5.7$--$19.0$ & $142$ & $3$ \\
He chain & 6-31G
& $14$--$945$ & $3.3$--$3.4$ & $68$ & $1$ \\
Li chain & 6-31G
& $96$--$6573$ & $137.2$--$914.3$ & $68$ & $7$ \\
\midrule
\multicolumn{6}{c}{\textbf{LCCSDT}} \\
\midrule
LiH & cc-pVDZ
& $5$--$1155$ & $38.7$--$121.7$ & $231$ & $3$ \\
\midrule
\multicolumn{6}{c}{\textbf{LCCSDTQ}} \\
\midrule
LiH & cc-pVDZ
& $5$--$3035$ & $38.7$--$143.4$ & $607$ & $4$ \\
\bottomrule
\end{tabular}
\end{table*} 

\clearpage 

\begin{table}[t] 
\centering
\caption{Matrix dimension $N$, condition number $\kappa$, and their changes across the number of virtuals considered in the icMRLCCSD calculations with positive definite $A$. Subscripts $i$ and $f$ denote the first and last sampled points respectively. The $\kappa$ range lists the minimum and maximum over all sampled points, while the final two columns give ratios of the final to initial values of the respective quantities. $\kappa$ values are rounded to one decimal place, 
$N_f/N_i$ to the nearest integer, and $\kappa_f/\kappa_i$ to two decimal places. } 
\label{tab:icmr_scaling} 
\begin{tabular}{llrrrr} 
\hline 
Molecule/atom & Basis set & $N$ range & $\kappa$ range & $N_f/N_i$ & $\kappa_f/\kappa_i$ \\ 
\hline
\multicolumn{6}{c}{Atoms} \\
\midrule 
Be & cc-pVQZ & $71$--$4467$ & $75.5$--$154.0$ & $63$ & $2.04$ \\
Ne & cc-pVQZ & $197$--$6868$ & $68.9$--$85.9$ & $35$ & $1.25$ \\ 
\midrule 
\multicolumn{6}{c}{Diatomic molecules} \\ 
\midrule 
LiH & cc-pVTZ & $29$--$3743$ & $68.6$--$154.6$ & $129$ & $2.25$ \\
Li$_2$ & cc-pVTZ & $64$--$5415$ & $68.5$--$121.3$ & $85$ & $1.77$ \\ 
BH & cc-pVTZ &  $329$--$6752$&  $453.1$--$2527.8$&  $21$& $5.58$  \\ 
\midrule 
\multicolumn{6}{c}{Triatomic molecules} \\ 
\midrule 
BeH$_2$ & cc-pVTZ & $64$--$5547$ & $94.7$--$173.2$ & $87$ & $1.81$ \\ 
\midrule 
\multicolumn{6}{c}{Tetratomic molecules} \\ 
\midrule 
H$_4$ & cc-pVTZ & $62$--$3420$ & $2.8$--$17.5$ & $55$ & $6.04$ \\ 
\hline
\end{tabular}
\end{table} 

\clearpage 

\vspace{1cm}
\begin{longtable}{cccccc}
\caption{Table containing data of energy differences (in mHa) between the total energy of each method and the FCI energy. Geometry refers to bond lengths (in bohr) for LiH, the distance $d$ for H$_4$ and labels A to I for BeH$_2$ (See Table \ref{tabbeh2}).}\label{taballmolecules} \\ 
\hline 
\multicolumn{1}{c}{Molecule} &
  Geometry (bohr) &
  \multicolumn{4}{c}{$E^{\text{SR}/\text{icMR}} - E_{FCI}$ (mHa)} \\
 &
  &
  CCSD &
  CCSD(T) &
  icMRCISD &
  HHL-icMRLCCSD\\\hline 
  \endfirsthead
  \multicolumn{6}{c}{TABLE \thetable{}: \textit{Continued.}} \\
\hline
  \multicolumn{1}{c}{Molecule} &
  Geometry &
  \multicolumn{4}{c}{$\Delta E \text{ (mHa)}$} \\
 &
  &
  CCSD &
  CCSD(T) &
  icMRCISD &
  HHL-icMRLCCSD \\\hline 
  \endhead
  \hline
  \endfoot
  \hline
  \endlastfoot
 &
   1.880 &
  0.013 &
  0.000 &
  0.002 &
  -0.011 \\
 &
  2.070&
   0.014 &
  0.000 &
  0.002 &
  -0.011 \\
 &
  2.260  &
  0.010 &
  0.000 &
  0.001 &
  -0.011 \\
 &
   2.450 &
  0.010 &
  0.001 &
  0.001 &
  0.014 \\
 &
  2.640 &
   0.010 &
  0.001 &
  0.001 &
  0.012 \\
 &
   2.830 &
   0.010&
  0.001 &
  0.001 &
  -0.025 \\
 &
   3.020 &
   0.011&
   0.001 &
  0.001 &
  -0.007\\
 &
   3.210 &
  0.011 &
   0.001 &
  0.001 &
  -0.013 \\
 &
   3.400 &
   0.012 &
   0.001 &
  0.001 &
  0.004 \\
 &
  3.590 &
   0.001 &
 0.002 &
  0.001 &
  -0.005 \\
 &
  3.770 &
  0.013 &
   0.001 &
  0.001 &
  -0.013 \\
 &
   3.960 &
   0.014 &
   0.002 &
  0.001 &
  -0.045 \\
 &
   4.150 &
   0.016 &
   0.002 &
  0.001 &
  0.038 \\
 &
   4.340 &
   0.017 &
   0.002 &
   0.001 &
  0.007 \\
 &
   4.530 &
   0.019 &
   0.001 &
  0.001 &
  -0.016 \\
 &
 4.720 &
 0.021 &
   0.001 &
  0.001 &
  0.036 \\
 &
 4.910 &
  0.023 &
   0.001&
  0.001 &
  0.000 \\
 &
 5.100 &
 0.026 &
   0.000 &
  0.001 &
  0.042 \\
 &
  5.290 &
   0.030 &
   0.000&
  0.001 &
  0.000 \\
 &
   5.480 &
   0.034 &
   -0.002 &
  0.001 &
  -0.008 \\
 &
  5.660 &
   0.038 &
   -0.004 &
  0.001 &
  0.031 \\
 &
   5.850 &
   0.043 &
   -0.008 &
  0.001 &
  0.002 \\
 &
   6.040 &
  0.048 &
   -0.013 &
  0.001 &
  -0.009 \\
 &
   6.230 &
   0.053 &
   -0.020 &
  0.001 &
  -0.017 \\
 &
  6.420 &
   0.058 &
   -0.029 &
  0.009 &
  0.008 \\
 &
   6.610 &
  0.064 &
    -0.039 &
  0.008 &
  0.001 \\
 &
6.800 &
   0.069 &
   -0.052 &
  0.007 &
  -0.001 \\
 &
  6.990 &
  0.074 &
   -0.065 &
  0.005 &
  -0.023 \\
 &
  7.180 &
  0.079 &
  -0.080&
  0.004 &
  -0.009 \\
 &
   7.360 &
   0.083 &
  -0.095 &
  0.004 &
  -0.009 \\
\multirow{-31}{*}{LiH} &
    7.550 &
   0.087 &
   -0.111 &
  0.003 &
  -0.001 \\\hline
 &
  1.800 &
  0.718 &
 0.005 &
   0.002 &
  0.257 \\
 &
  1.900 &
  1.361 &
 -0.414  &
 0.022 &
0.058
\\
 &
  1.990 &
  1.987&
  -3.164 &
   0.001&
0.033
  \\
 &
  2.010 &
    1.909&
   -3.216 &
   0.001 &
  0.024\\
 \multirow{-3}{*}{H$_4$} 
 &
  2.100 &
  1.084 &
   -0.073 &
   0.016&
   0.072\\
 &
  2.200 &
  0.566 &
   -0.165 &
0.015 &
   0.195 
  \\\hline 
 & A &
 0.290 &
  0.109 &
  0.362 &
  0.081 \\
 &
  B &
  0.278 &
  0.091 &
  0.326 &
  0.039 \\
 &
  C &
  0.397 &
  0.112 &
  0.552 &
  -0.100 \\
 &
  D &
  0.670 &
   0.195 &
  0.168 &
  -0.024 \\
 &
  E &
  1.073 &
   0.349 &
  1.325 &
  -1.454 \\
 &
  F &
  104.843 &
   97.429 &
  1.544 &
  -0.104 \\
 &
  G &
  47.605 &
   47.240 &
  0.624 &
  -0.301 \\
 &
  H &
  20.721 &
  20.560 &
  0.495 &
  -0.296 \\
\multirow{-9}{*}{BeH$_2$} &
  I &
  -0.459 &
  -0.477 &
  0.495 &
  -0.303 \\
\end{longtable} 

\end{document}